\documentclass[acmsmall,screen,nonacm]{acmart}

\usepackage{subcaption}
\usepackage{enumitem}
\graphicspath{{}{plots/}}

\definecolor{myorange}{RGB}{0,0,0}

\DeclareFixedFont{\ttb}{T1}{txtt}{bx}{n}{7} 
\DeclareFixedFont{\ttm}{T1}{txtt}{m}{n}{7}  

\definecolor{deepblue}{rgb}{0,0,0.5}
\definecolor{deepred}{rgb}{0.6,0,0}
\definecolor{deepgreen}{rgb}{0,0.5,0}

\usepackage{listings}
\newcommand\pythonstyle{\lstset{
language=Python,
basicstyle=\ttm,
morekeywords={self},              
keywordstyle=\ttb\color{deepblue},
emph={MyClass,__init__},          
emphstyle=\ttb\color{deepred},    
stringstyle=\color{deepgreen},
frame=tb,                         
showstringspaces=false
}}

\lstnewenvironment{python}[1    ][]
{
\pythonstyle
\lstset{#1}
}
{}

\AtBeginDocument{%
  \providecommand\BibTeX{{%
    \normalfont B\kern-0.5em{\scshape i\kern-0.25em b}\kern-0.8em\TeX}}}

\setcopyright{acmcopyright}
\copyrightyear{2021}
\acmYear{2021}
\acmDOI{10.1145/1111111.1111111}

\acmConference[ACM Sigmetrics ']{ACM Sigmetrics '}{July 03--05, ----}{n.d., n.d.}

\newcommand\bs{\mathbf{s}}
\newcommand\bk{\mathbf{k}}
\newcommand\bx{\mathbf{x}}
\newcommand\be{\mathbf{e}}
\newcommand\bm{\mathbf{m}}

\newcommand\bS{\mathbf{S}}
\newcommand\bX{\mathbf{X}}

\newcommand\by{\mathbf{y}}

\newcommand\bz{\mathbf{z}}
\newcommand\bR{\mathbf{R}}
\newcommand\bg{\mathbf{g}}

\newcommand\bv{\mathbf{v}}
\newcommand\bw{\mathbf{w}}
\newcommand\bdrift{f}
\newcommand\bphi{\text{\boldmath$\phi$}}
\newcommand\bpsi{\text{\boldmath$\psi$}}

\newcommand\bQ{\mathbf{Q}}

\newcommand\calX{\mathcal{X}}
\newcommand\calI{\mathcal{I}}

\newcommand\calS{\mathcal{S}}
\newcommand\toN{^{(n)}}
\newcommand\bSn{\bS\toN}
\newcommand\bXn{\bX\toN}

\newcommand\Sn{S\toN}
\newcommand\Xn{X\toN}
\newcommand\Zn{Z\toN}

\newcommand{\ncopies}{^{[C\text{ copies}]}}
\newcommand{\onecopy}{^{[1\text{ copy}]}}

\newcommand\ns{n\times|\calS|} 

\newcommand{\R}{\ensuremath{\mathbb{R}}}

\DeclareMathOperator{\E}{\mathbb{E}} 
\newcommand\Proba[1]{\mathbb{P}\left(#1\right)} 
\newcommand\esp[1]{\E\left[#1\right]} 
\newcommand\norm[1]{\left\|#1\right\|}      
\newcommand\conv[1]{\mathrm{conv}(#1)}      
\newcommand\abs[1]{\left|#1\right|}         
\newcommand\cro[1]{\langle#1\rangle}         

\newcommand\floor[1]{\left\lfloor#1\right\rfloor}
\newcommand\p[1]{\left(#1\right)} 
\newcommand\ind[1]{\mathbf{1}_{\{#1\}}} 

\newcommand{\dddd}[5]{ \frac{\partial^4 #1 }{ \partial x_{#2} \partial x_{#3} \partial x_{#4} \partial x_{#5} }}
\newcommand{\ddd}[4]{ \frac{\partial^3 #1 }{ \partial x_{#2} \partial x_{#3} \partial x_{#4} }}
\newcommand{\dd}[3]{ \frac{\partial^2 #1 }{ \partial x_{#2} \partial x_{#3}}}

\usepackage{colortbl}
\newcolumntype{o}{>{\color{myorange}}c}

\begin{document}

\title[Mean Field Approximation(s) for Heterogeneous Systems: It Works!]{Mean Field and Refined Mean Field Approximations for Heterogeneous Systems: It Works!}

\author{Sebastian Allmeier}
\author{Nicolas Gast}
\affiliation{%
  \institution{Univ. Grenoble Alpes, Inria}
  \streetaddress{F-38000}
  \city{Grenoble}
  \state{}
  \country{France}
}


\begin{abstract}
  Mean field approximation is a powerful technique to study the performance of large stochastic systems represented as $n$ interacting objects. Applications include load balancing models, epidemic spreading, cache replacement policies, or large-scale data centers. Mean field approximation is asymptotically exact for systems composed of $n$ homogeneous objects under mild conditions. In this paper, we study what happens when objects are \emph{heterogeneous}. This can represent servers with different speeds or contents with different popularities. We define an interaction model that allows obtaining asymptotic convergence results for stochastic systems with heterogeneous object behavior, and show that the error of the mean field approximation is of order $O(1/n)$. More importantly, we show how to adapt the refined mean field approximation, developed by the authors of \cite{gastSizeExpansionsMean2019}, and show that the error of this approximation is reduced to $O(1/n^2)$. To illustrate the applicability of our result, we present two examples. The first addresses a list-based cache replacement model RANDOM($m$), which is an extension of the RANDOM policy. The second is a heterogeneous supermarket model. These examples show that the proposed approximations are computationally tractable and very accurate. They also show that for moderate system sizes ($n\approx30$) the refined mean field approximation tends to be more accurate than simulations for any reasonable simulation time.

\end{abstract}




\maketitle


\section{Introduction}

Mean field approximation method is a widely used tool to analyze large-scale and complex stochastic models composed of interacting objects. The idea of the approximation is to assume that objects within the system evolve independently. Following this assumption, interactions of objects in the system are approximated by a ``mean'' behavior, which allows to model the system's evolution as a set of deterministic ordinary differential equations. Mean field approximation finds widespread use in fields such as epidemic spreading \cite{montalbanHerdImmunityIndividual2020,dearrudaFundamentalsSpreadingProcesses2018}, load balancing strategies \cite{mukhopadhyayAnalysisLoadBalancing2015,mitzenmacherPowerTwoChoices2001}, caching \cite{gastTransientSteadystateRegime2015} or SSDs \cite{vanhoudtMeanFieldModel2013}. Building on this approximation, a refined mean field approximation is introduced in \cite{gastRefinedMeanField2018,gastSizeExpansionsMean2019} that greatly improves the accuracy of mean field approximations for populations of $n=10$ to $n=100$ objects. The popularity of mean field approximation lies in the ease of defining and solving the differential equations as well as the increasingly high accuracy for large systems. 

Most of the theoretical work, however, has been done for systems where the interacting objects have homogeneous transitions, as for density-dependent population processes of Kurtz \cite{kurtzStrongApproximationTheorems1978}, or can be clustered into a finite number of groups with similar statistical behavior. Yet, in many models, heterogeneity plays an important role. This is particularly relevant to model caches, where object popularities vary broadly among contents, or epidemic spreading, where variations of sensibility among agents can greatly influence the long-term dynamics and vaccination strategies \cite{gomesIndividualVariationSusceptibility2020}.  Using a finite number of clusters with homogeneous behavior simplifies the underlying models and essentially ignores the actual heterogeneity. Up to now, there are virtually no fully heterogeneous models with theoretical guarantees on why mean field approximation is a valid technique.

In this paper, we generalize the notion of mean field approximation and refined mean field approximation to stochastic systems composed of $n$  heterogeneous objects and show that similar asymptotic results as for the homogeneous case hold. For such a system, we show that it is possible to construct a set of ordinary differential equations (ODEs) which approximate $\Proba{S_k(t)=s}$, the probability for an object $k$ to be in a state $s$ at time $t$. {\color{myorange}This can be used to approximate the expectation of a function of the state of an object (such as the average queue length in a queuing system).}

To give some intuition, the way we construct our approximations is to consider a scaled model with $C$ identical copies of each object. This allows one to define the mean field approximation $x_{(k,s)}(t)$ and a $1/C$-expansion term $v_{(k,s)}(t)$ defined in \cite{gastSizeExpansionsMean2019}. These approximations are shown in \cite{gastSizeExpansionsMean2019} to be asymptotically accurate as the number of copies $C$ goes to infinity. As the fully heterogeneous system corresponds to having one copy, the heuristic reasoning is then to apply the approximation with $C=1$. Up to now, there was no theoretical foundation on why this should work because all results assume that the number of copies $C$ goes to infinity.

We provide the first rigorous justification of the validity of this approach. The main contribution of our paper is to show that if $x_{(k,s)}(t)$ is the mean field approximation and $v_{(k,s)}(t)$ is the expansion term defined in \cite{gastSizeExpansionsMean2019}, then
\begin{align*}
  \Proba{S_k(t)=s} &= x_{(k,s)}(t) + O(1/n),\\
  \Proba{S_k(t)=s} &= x_{(k,s)}(t) + v_{(k,s)}(t) + O(1/n^2).
\end{align*}
{\color{myorange}As an important corollary, if the system is composed of $C$ copies of $n$ heterogeneous objects, we then have $\Proba{S_k(t)=s} = x_{(k,s)}(t) + (1/C) v_{(k,s)}(t) + O(1/(Cn)^2)$. This shows that the accuracy of the mean field and refined mean field approximation does not depend on the level of heterogeneity of the system but only on the total number of objects (being $n$ or $nC$).}

To do so, we develop a {\color{myorange} heterogeneous interaction model in which each object changes state either unilaterally or by interacting with $d-1$ other objects. The main assumption that we make in our model is that the rate at which a given tuple of $d$ individuals interact scales as $O(1/n^{d-1})$. As there are $O(n^d)$ such tuples, this guarantees a uniform bound on the interactions between tuples.} This model covers an extensive range of models with pairwise interactions, such as infection models, load balancing strategies, or cache replacement policies. These approximations can be computed by solving a differential equation that can be easily integrated numerically. For the mean field approximation, the number of variables grows linearly with the number of different objects $n$. For the refined approximation, it grows quadratically with $n$. Our proposed framework naturally extends mean field models for homogeneous population processes and the results are comparable with \cite{gast2017expected,kurtzStrongApproximationTheorems1978,ying2016rate}. Our approach does not assume any homogeneity in the system and does not cluster objects into a finite number of classes. Hence, it can be applied to interacting systems where all objects are different.

To illustrate our results we provide two examples that show how the mean field and refined mean field approximation can be applied. They also show that the hidden constants in the $O(1/n)$ and $O(1/n^2)$ error terms given by the theorems are small in practice. Our first example is a list-based cache replacement algorithm studied in \cite{gastTransientSteadystateRegime2015} consisting of $n$~objects whose popularities follow a Zipf-like distribution. We study how the cache popularities depend on the replacement policy for the transient and steady-state regime. For transient results, we compare the mean field and refined mean field approximation with simulations, which indicates that the refined mean field provides a significant improvement of accuracy. The results are even more striking for the steady-state regime for which it is possible to compute the exact steady-state distribution if the system size is small. This allows us to compare the accuracy of the two approximations and the simulation to the exact value. We observe that, for any reasonable computational power, the confidence intervals provided by the simulation are higher than the error of the refined mean field approximation as soon as $n$ exceeds a few tens. In a second example, we apply the approximation techniques to a heterogeneous two-choice load balancing model. The heterogeneity in the model is introduced by considering varying server rates. As for the previous examples, we give a full description justifying the use of the mean field models and show by numerical computation that the obtained results confirm the theoretical statements. Numerical calculations for values such as the average queue length and queue length tail distribution are given. We also compare the approximation results to a homogeneous variant of the system where the server rates are set to the average server rate of the heterogeneous model. This shows that, as expected, taking heterogeneity into account strikingly improves the accuracy of the results. For both examples, we adapt numerical methods from the toolbox developed by \cite{gastRefinedMeanField2018}, which allowed us to implement and solve the differential equations with relative ease.

\paragraph*{Roadmap} The remainder of this paper is organized as follows. We describe more related work in Section~\ref{sec:related_work}. We introduce the heterogeneous interaction model in Section~\ref{sec:model}. We define and study the accuracy of the approximation in Section~\ref{sec:main_results}. This section contains the main results of the paper. We present the two numerical examples in Section~\ref{sec:numerical}. The proofs are given in Section~\ref{sec:proofs}. Finally, we conclude in Section~\ref{sec:conclusion}.  Some technical lemmas are postponed to the appendix.


\section{Related Work}
\label{sec:related_work}

\paragraph*{Generator and Stein's method}
Our paper builds on the recent line of work regarding the use of Stein's method \cite{stein1986approximate}. This method allows one to estimate precisely the distance between two random variables by looking at the distance between the generators of two related stochastic systems. This method has seen a resurgence of interest in the stochastic network community since the work of \cite{braverman2017stein,braverman2017stein2}. There is still an active development in this area. For instance, this method has been used to develop higher-order diffusion approximations \cite{bravermanHighOrderSteadystate2020,bravermanSteinMethodSteadystate2017a}. It is used in \cite{hodgkinsonNormalApproximationsDiscretetime2018} to develop a normal approximation of a heterogeneous discrete time population process. One of the key differences between our work and theirs is that the two aforementioned papers consider one-dimensional processes (i.e., the state of each object of the system is either $0$ or $1$), and the extension to more complex dynamics is not direct, at least from a computational point of view. One contribution of our work is to demonstrate how to deal with multiple states, by changing the state representation. 

\paragraph*{Refined mean field methods} Stein's method has been successfully used to study the accuracy of mean field methods in \cite{gast2017expected,kolokoltsovMeanFieldGames2012,ying2016rate,ying2017stein}. These works show that the accuracy of the mean field approximation is $O(1/n)$ for a system with $n$ \emph{homogeneous} objects. By using an expansion of the generator, these results have been extended in \cite{gastSizeExpansionsMean2019,gastRefinedMeanField2017} to propose what the authors called a \emph{refined} mean field approximation, that is similar to the system size expansion introduced in mathematical biology \cite{vankampen2007,grima2010effective,grima2011}. 

In fact, there exists a close link between the refined approximation that we propose in the paper and the approach of \cite{gastSizeExpansionsMean2019,gastRefinedMeanField2017}. Since the results of \cite{gastSizeExpansionsMean2019,gastRefinedMeanField2017} only apply to systems composed of homogeneous objects, let us consider a hypothetical model composed of $C$ identical copies of each of the $n$ object, and let us denote by $X_{(k,s)}\ncopies(t)$ the number of copies of the object $k$ that are in state $s$ at time $t$. By \cite{gastSizeExpansionsMean2019,gastRefinedMeanField2017}, there exists a constant $v_{(k,s)}(t)$ such that 
\begin{align}
  \label{eq:copies}
  \esp{X_{(k,s)}\ncopies(t)} = x_{(k,s)}(t) + \frac1C v_{(k,s)}(t) + O(1/C^2).
\end{align}
Since our original model corresponds to $C=1$, the rationale behind our approximation is to use $x_{(k,s)}(t)$ as a first order approximation and $x_{(k,s)}(t) + v_{(k,s)}(t)$ as the refined approximation. Yet, this is no a priori guarantee of why $O(1/C^2)$ should be small for $C$ equal to one. As a key technical contribution, our paper gives the theoretical foundation of this method. 
{\color{myorange}
This requires to overcome several difficulties, that are our main contributions compared to the aforementioned papers:
\begin{enumerate}
  \item We define an interaction model that can be dealt with by bounding the interaction rates. 
  \item A key idea of our work is to use the indicators $X_{k,s}(t)\in\{0,1\}$ and not the proportion of objects in a given state. This allows us to construct the expansion terms and do the proofs.
  \item Similar to previous works \cite{gast2017expected,kolokoltsovMeanFieldGames2012,ying2015rate}, we use generators to reduce the analysis of the mean field error to the study of the sensibility of an ODE with respect to its initial conditions (Section~\ref{sec:proof_lemma}). The extra difficulty in our case is to carefully analyze the remainder terms: while to analyze \eqref{eq:copies}, the error is a finite sum of $O(1/C)$ terms, here we have a sum of $n$ (or $n^2$) terms, some of them being of order $O(1/n^2)$ and others being of order $O(1/n^3)$. Dealing with all these different cases requires quite some care and is the subject of Appendix~\ref{apx:proofs}.
\end{enumerate}
}

Note that to study heterogeneous population models, it is quite common to assume that there are $n$ classes with $C$ copies of the same objects of class $k\in\{1\dots n\}$. This has been used for instance to study load-balancing strategies \cite{mukhopadhyayAnalysisLoadBalancing2015} or cache replacement policies \cite{hirade1,tsukada1}. Our approach generalizes such methods as we assume that the objects can be fully heterogeneous. This is for instance what is used in replica models \cite{mezard1987sk,baccelli2019replica}. 

\paragraph*{Heterogeneous populations and caches} In the performance evaluation community, heterogeneous population models are very common when studying cache replacement policies, where the popularity of objects is typically assumed to follow a Zipf-like distribution. As the dynamics of caches are intrinsically complicated, many mean field like approximations have been proposed, such as the famous TTL-approximation of \cite{fagin1977asymptotic} (sometimes misleadingly called the Che-approximation after it was rediscovered in \cite{che2001analysis}) or fixed-point approximation like \cite{dan1990approximate}. Theoretical support exists to prove that these approximation are asymptotically correct \cite{fricker2012versatile,jiang2018convergence,gastTransientSteadystateRegime2015}. 

The generic method that we propose in this paper has two advantages: First, we prove that the accuracy of the mean field method is $O(1/n)$ whereas the above papers only obtain bounds in $O(1/\sqrt{n})$. Second, we develop a refined approximation that can greatly improve the accuracy compared to the cited method, at the price of being computationally more expensive.

{\color{myorange}
In particular, our result applies directly to the cache replacement model of \cite{gastTransientSteadystateRegime2015} in which a mean field approximation for the RAND($\bm$) policy is derived. This paper also contains a theorem that shows that the mean field approximation is $O(1/\sqrt{n})$-accurate. Yet, we do not think that the proof of the main result of \cite{gastTransientSteadystateRegime2015} is correct because of the use of a martingale inequality combined with the infinite-norm (and not a $L_2$-norm). More precisely, we believe that the problem in their proof is just below their Equation (13) when Lemma 1(ii) is used.  What their Lemma 1 implies is that $M(t)$ is a Martingale such that $\esp{ \norm{M(t+1)-M(t)}^2} \le c$. This is used below their Equation (13) to imply that $\esp{\norm{M(t)}^2} \le ct$. The problem is that this is true if the norm $\norm{M}$ is a $L_2$ norm (or any norm that can be written as a scalar product $\norm{M}^2 =\langle M,M\rangle$ ) but not if $\norm{M}$ it is a supremum norm. The norm used in \cite{gastTransientSteadystateRegime2015} is a supremum norm and we do not believe it can be derived from a scalar product. The approach that we take in this paper is radically different as we work with a comparison of generators. This allows us to correct the proof of \cite{gastTransientSteadystateRegime2015} by obtaining a tighter bound. Note that we do not claim that their result is false, but only that the proof is false. We explain in Appendix~\ref{apx:cache_adaptation2} that their result is a consequence of ours (and that our results give a finer bound).
}


\section{The Heterogeneous Population Model}
\label{sec:model}


\subsection{Interaction Model}
\label{ssec:interaction_model_and_definition}

We consider a population model composed of $n$ interacting objects. Each object evolves in a finite\footnote{The fact that objects share the same state space is done without loss of generality as we do not assume that Markov chains are irreducible. } state space $\calS$. The state of the $k$-th object at time $t$ is denoted by $S\toN_k(t)\in\calS$ and the state of the system at time $t$ is given by $\bSn(t)=(\Sn_1(t), \dots, \Sn_n(t))\in \calS^n$. We assume that the stochastic process $\bSn=(\bSn(t))_{t\ge0}$ is a continuous time Markov chain (CTMC) whose transitions are the results of interactions between objects. 
\color{myorange}
More precisely, we assume for any tuple of $d$ objects $\bk=(k_1,..,k_{d})$, that these objects jump \emph{simultaneously}\footnote{Note that for the transitions caused by interactions, we do not impose that all objects jump: we may have $s_i=s_i', i\in \{1,\dots,d\}$ in which case some objects keep their state.} from their states $\bs=(s_1, \dots,s_{d})$ to $\bs'=(s_1', \dots, s_{d}') \ne (s_1,\dots,s_{d})$ at rate $\frac{1}{dn^{d-1}}r_{\bk, \bs \rightarrow \bs'}\toN$. All such interactions occur independently. We also assume that $d\le d_{\max}$ is a constant independent of $n$, i.e. the maximal amount of interacting objects does not scale with the system size. 

Throughout the paper, we will refer to such a model as a \emph{heterogeneous population model}.  Note that while a transition can affect up to $d_{\max}$ objects, all the examples studied in the paper will be with $d_{\max}=2$ for which there are two types of transitions:
\begin{itemize}
  \item $d=1$: An object jumps without interacting with others. We call this a unilateral transition.
  \item $d=2$: Two objects interact. We call it a pairwise interaction. 
\end{itemize}

The critical assumption of our model is that the interactions between $d$ objects scale as $O(\frac{1}{n^{d-1}})$. In particular, the rates of unilateral transitions scale like $O(1)$ and the one of pairwise interactions like $O(1/n)$.  This $1/n^{d-1}$ factor is here because there are $O(n^{d-1})$ tuples of $d$ objects. Hence, our condition implies that the total rate of transitions is $O(n)$ and that one tuple cannot have much higher rates than other tuples.

To simplify notations, we assume that for any permutation $\sigma$ of the set $\{1\dots d\}$, the rate satisfy $r_{k_1, \dots, k_{d}, (s_1, \dots, s_{d}) \rightarrow (s_1', \dots, s_{d}')}\toN = r_{k_{\sigma(1)}, \dots, k_{\sigma(d)}, (s_{\sigma(1)}, \dots, s_{\sigma(d)}) \rightarrow (s_{\sigma(1)}', \dots, s_{\sigma(d)}')}\toN$. This does not imply that objects are homogeneous but should be seen at a notation artifact. An alternative notation would be to consider tuples such that $k_1<k_2<\dots <k_d$ and to multiply all rates by $d!$. This would lead to the same model but at the price of heavier notations.


\color{black}

\subsection{State Representation}
\label{ssec:state_representation}
The key element of our analysis is to use an alternative, binary based, representation of the state space. For an object $k\in\{1\hdots n\}$ and a state $s\in\calS$, we define $\Xn_{(k,s)}(t)$ as
\begin{align*}
  \Xn_{(k,s)}(t) = \ind{S_k\toN(t)=s} &:=\left\{\begin{array}{ll}
    1 &\text{if object $k$ is in state $s$ at time $t$}, \\
    0 &\text{otherwise.}
  \end{array}\right.
\end{align*}
The state of the system is described by $\bXn(t)=(\Xn_{(k,s)}(t))_{k\in\{1\hdots n\},s\in\calS}$ with $\calX\toN\subset\{0,1\}^{\ns}$ being the set of attainable states for $\bXn=(\bXn(t))_{t\geq0}$. In particular, for all $\bx\in\calX\toN$, one has $\sum_{s\in\calS} x_{(k,s)} = 1$ which follows from the fact that an object can only be in one state at any time. The notation $\bXn$ is less compact than the original representation $\bSn$ but will allow for an easier definition of the mean field and refined mean field approximation. 

The transitions of the model can all be expressed in terms of $\bX\toN$. Let $\be\toN_{(k,s)}$ denote a matrix of size $n\times |S|$ whose $(k,s)$ component is equal to $1$, all others being equal to $0$. If object $k$ transitions from state $s$ to state $s'$, $\bXn$changes into $\bXn+\be\toN_{(k,s')}-\be\toN_{(k,s)}$. Hence, expressed as a function of $\bXn$, the process $\bX^{(n)}$ jumps to (for $k_1,k_2, \hdots \in\{1\hdots n\}$ and $s_1,s_1',s_2,s_2',\hdots \in\calS$):
\color{myorange}
\begin{subequations}
  \begin{align}
    &\bX^{(n)} {+} \be\toN_{(k_1,s_1')}{-}\be\toN_{(k_1,s_1)} && \text{ at rate $r_{k_1,(s_1)\rightarrow (s_1')}\toN X_{(k_1,s_1)}^{(n)}$}\label{eq:unilateral}\\
    &\bX^{(n)} {+} \be\toN_{(k_1,s_1')}{-}\be\toN_{(k_1,s_1)} + \be\toN_{(k_2,s_2')}{-}\be\toN_{(k_2,s_2)} && \text{ at rate $\frac{1}{2n}r_{k_1,k_2,(s_1,s_2) \rightarrow (s_1',s_2')}\toN X_{(k_1,s_1)}^{(n)}X_{(k_1,s_1)}^{(n)}$}\label{eq:interactions} \\
    &\bX^{(n)} {+} \be\toN_{(k_1,s_1')}{-}\be\toN_{(k_1,s_1)}{+}\hdots{+}\be\toN_{(k_{d},s_{d}')}{-}\be\toN_{(k_{d},s_{d})} &&\text{ at rate $\frac{1}{d n^{d-1}}r_{k_1,\hdots,k_{d},\bs \rightarrow \bs'}\toN X_{(k_1,s_1)}^{(n)}\hdots X_{(k_{d},s_{d})}^{(n)}$} \label{eq:d-type_interactions}
\end{align}
\label{eq:transitions}
\end{subequations}
In the above equations \eqref{eq:unilateral} corresponds to an unilateral transition of object $k_1$ from $s_1$ to $s_1'$, \eqref{eq:interactions} corresponds to transitions caused by a pairwise interaction between object $k_1$ and $k_2$ and \eqref{eq:d-type_interactions} describes the general form of the transitions of $d$ interacting objects. Recall that we assume that $(s_1,s_2,\hdots)\ne(s_1',s'_2,\hdots)$ but we do not necessarily assume all states change, i.e. a pairwise interaction might result in either one object changing state or two objects that change state simultaneously.
\color{black}

\subsection{Main Notations}

Throughout the paper, we use bold letters (like $\bX\toN,\bx,\hdots$) to denote matrices and regular letters (like $X_{(k,s)}\toN,x_{(k,s)},n,\hdots$) to denote scalars. Capital letters (like $\bX\toN, \bS\toN$) denote random variables whereas lower case letters $(\bx,v_{(k,s)}\toN,\hdots$) are for deterministic values. The indices $k,k_1,k',\hdots$ are reserved for objects while $s,s_1,s'\hdots$ are reserved for the states. 

In the results below, when we write that a quantity $h$ satisfies $h = O(1)$ or $h=O(1/n)$, this means that there exists a constant $C$ independent of $n$ such that $h\le C$ or $h\le C/n$. In general, these constants do depend on other parameters of the problem (like $|\calS|$, $\bar{r}$, or $t$).


\section{Main Results}
\label{sec:main_results}

In this section, we define the mean field and refined mean field approximation for the heterogeneous system and formulate the corresponding theorems.

\subsection{Drift and Mean Field Approximation}
\label{ssec:drift_mean_field}


For a given state $\bx\in\calX\subset \{0,1\}^{\ns}$, we define the drift of the system in $\bx$ as the expected variation of the stochastic process $\bX\toN$ at time $t$: $\bdrift\toN(\bX\toN(t))=\lim_{dt\to0}\frac{1}{dt}\esp{\bX\toN(t+dt)-\bX\toN(t) \mid \bX\toN(t)=\bx}$. \color{myorange} Based on the transitions \eqref{eq:transitions}, if $d\le 2$ (\emph{i.e.}, only unilateral and pairwise interactions)\footnote{We give the general drift definition in Appendix \ref{apx:drift_def}}, the $(k,s)$ component of the drift can be expressed as: 
\begin{align}
  & \sum_{s' \ne s} (r\toN_{k,(s') \rightarrow (s)} x_{(k,s')} - r\toN_{k,(s) \rightarrow (s')} x_{(k,s)})  \label{eq:drift}\\
  & \quad + \frac1n \sum_{s',k_1, s_1,s_1'} (r\toN_{k,k_1,(s',s_1')\rightarrow (s,s_1)} x_{(k,s')}x_{(k_1,s_1)} - r\toN_{k,k_1,(s,s_1) \rightarrow (s',s'_1)} x_{(k,s)}x_{(k_1,s_1')}) 
  \nonumber
\end{align} 
The first term corresponds to unilateral transitions while the second term corresponds to transitions caused by pairwise interactions. Note that compared to \eqref{eq:interactions}, there seems to be an extra factor $2$ in front of the pairwise interactions. This is not an error and it is due to the fact that we fixed the position of $s$ in the above equation.

\color{black}Note that if the conditional expectation is only defined for $\bx\in\calX\toN\subset\{0,1\}^{\ns}$, the above expression \eqref{eq:drift} can be extended to a function $f\toN:\conv{\calX\toN}\subset[0,1]^{\ns}\to\R^{\ns}$, where $\conv{\calX\toN}$ denotes the convex hull of $\calX$. For a given initial condition $\bx \in \calX\toN$, we define the mean field approximation of the heterogeneous population model as the solution  of the ODE $\frac{d}{dt} \bphi\toN(\bx,t) = f\toN(\bphi\toN(\bx, t))$ that starts in $\bx =\bX\toN(0)$, and we denote by $\bphi\toN(\bx, t)$ the value of this solution at time $t$. The solution is unique as $f\toN$ is Lipschitz-continuous (all elements of $f\toN$ are polynomials) and $\bphi(\bx, t)$ takes values in a bounded set. Notice from the definition of the ODE, that $\phi\toN_{(k,s)}\in[0,1]$, and  $\sum_{s} \phi\toN_{(k,s)}(\bx,t) = 1$ for all $t$ and $\bx$ in $\conv{\calX\toN}$. 

\subsection{Accuracy of the Mean Field Approximation}

\color{myorange}
To obtain asymptotic properties we require that there exists a uniform bound $\bar{r}$, independent of $n$, such that for all $s_1,s_1',s_2,s_2',... \in \calS$ and $k_1,k_2,\hdots \in \{1,\ldots, n\}$ we have:
\begin{align}
  r\toN_{k_1,\hdots,k_d, (s_1,\hdots,s_d)\rightarrow(s_1',\hdots,s_d')} \leq \bar{r}. \label{eq:bound:r}
\end{align}


\begin{theorem}
  \label{th:MF}
  Assume that the model, defined in Section \ref{ssec:state_representation}, satisfies \eqref{eq:bound:r}. Let $\bphi\toN(\bx, t)$ be the solution of the ODE introduced in Section \ref{ssec:drift_mean_field} with initial condition $\bX\toN(0)= \bx \in \calX\toN $ and drift $f\toN$. Then, for $(k,s)\in \{1,\ldots,n\}\times \calS$ and $t < \infty$,
  \begin{align}
    \label{eq:MF_theorem}
    \Proba{S_k(t)=s} &= \esp{X\toN_{(k,s)}(t) } = \phi\toN_{(k,s)}(\bx,t) + O(1/n).
  \end{align}
\end{theorem}


\color{black}

\begin{proof}[Key ideas of the proof]
  A complete proof is provided in Section~\ref{sec:proof_mf}. Here we only give a brief overview. Note that by definition of the initial value problem and the drift being a polynomial, the differentiability of the solution $\bphi$ w.r.t. the initial condition and time is given. In the first part of the proof, we borrow ideas from \cite{kolokoltsovMeanFieldGames2012,gast2017expected} and use a Taylor's expansion of $\bphi\toN$ to show that $\Proba{S_k(t)=s} - \phi\toN_{(k,s)}(\bx,t)$ can be bounded by a weighted sums of $\frac{\partial^2 \phi_{(k,s)}(\bx,t)}{\partial x_{(k',s')} \partial x_{(k'',s'')}}$.  The key technical difficulty of the proof is then to show that these terms are small when the number of objects is large. To do so, we distinguish the cases where $k'$ and $k''$ refer to the same object as $k$ or not. 
\end{proof}

The statement of the theorem can be interpreted as saying that the probability of an object~$k$ in the Markov chain $\bX$ to be in state $s$ is approximated by $\bphi_{(k,s)}$ with an accuracy of $O(1/n)$. Indeed, with this result we can obtain similar statements as for the homogeneous case where in many cases asymptotic results are  proven for $\Zn_s(t) = \frac1n\sum_{k=1}^n X_{(k,s)}(t)$, the stochastic process describing the fraction of objects which are in state $s$. It should be noted that the solution of the ODE $\bphi_{(k,s)}$, taking values in $[0,1]$, is not close to the value of $\bX_{(k,s)}$, which indicates if object~$k$ is in state $s$ and takes the values zero or one. Hence, single trajectories of the stochastic process are not comparable to the approximation. 

To illustrate this result, let us consider a cache model with a total of $n=4$ objects and a cache  that can store $2$ objects. We assume that the popularities follow a Zipf distribution of parameter $0.8$, meaning that object $k$ is requested at rate $\lambda_k=1/k^{0.8}$ and use the RANDOM replacement policy. The policy exchanges objects the following way: When an object is requested and inserted in the cache, we evict another object picked uniformly at random among the two objects in the cache. Initially, the cache contains the objects $3$ and $4$.

In Figure~\ref{fig:cache_illustration}, we plot the behavior of the cache as a function of time. Each plot corresponds to a different object and contains three curves: In grey we plot one stochastic trajectory of the cache, $X_{(k,\text{in})}(t)$, where $X_{(k,\text{in})}=1$ means that the object~$k$ is in the cache and $0$ that it is not. In blue, we plot the probability for object $k$ to be in the cache at time $t$, $\Proba{\text{object $k$ in cache at time $t$}} = \esp{X_{(k,\text{in})}(t)}$, which is computed by averaging over $1000$ trajectories.  In green, we plot the solution of the mean field approximation, $\phi\toN_{(k,\text{in})}(\bx,t)$. We emphasize that $X_{(k,\text{in})}(t)$ is never close to $\phi\toN_{(k,\text{in})}(\bx,t)$, because the former can only take the values $0$ and $1$ whereas the latter takes values between $0$ and $1$. Moreover, the latter, which is the mean field approximation, seems to provide a very good approximation for the object to be in the cache. 
\begin{figure}[ht]
  \centering
  \includegraphics[width=\linewidth]{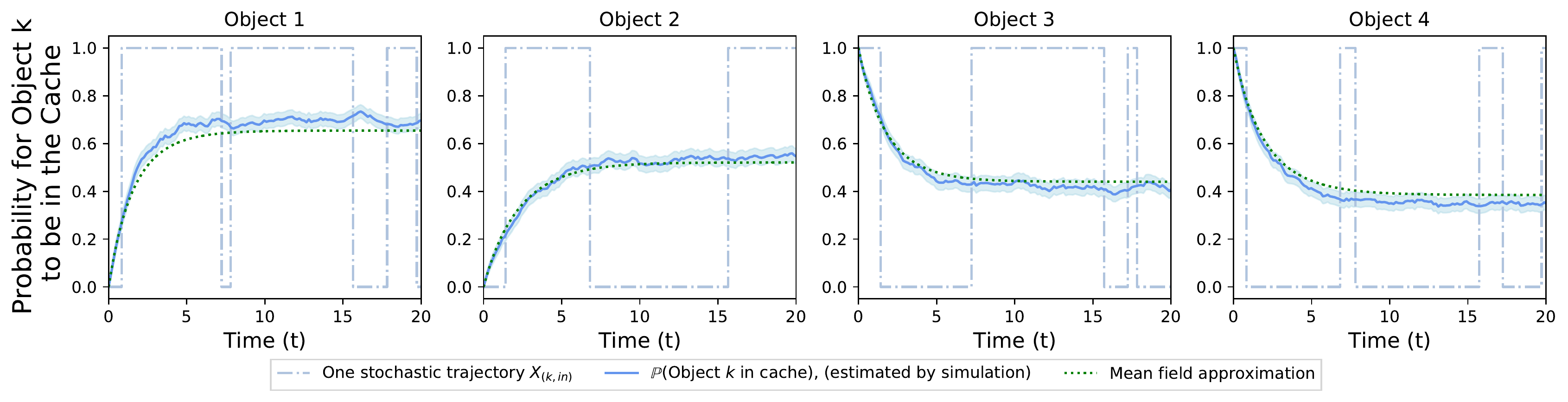}
  \caption{Behavior of the RANDOM policy for a cache of size two and a total of four objects. For each of the four objects, we compare the stochastic system with the mean field approximations.}
  \label{fig:cache_illustration}
\end{figure}

\subsection{Accuracy of the Refined Mean Field Approximation}
\label{ssec:rmf:definition}
\color{myorange}
In \cite{gastSizeExpansionsMean2019,gastRefinedMeanField2017}, Gast et al. introduced a refined mean field approximation which provides significantly more accurate approximations in the homogeneous case. The idea is to study higher moments of $X_{(k,s)} - \phi_{(k,s)}$ and to derive refinement terms. Taking the refinements into account is especially important for small to moderate system sizes, i.e. $n\approx10-50$, where the mean field approximation does not capture the dynamics of the stochastic system well. In this section, we show how to derive the refinement for our heterogeneous model.

To construct the refinement term, we consider an (imaginary) system in which there are $C$ replicas of the same object. Let us denote $X\ncopies_{(k,s)}(t)$ the fraction of replicas of type $k$ that are in state $s$ at time $t$. The process $\mathbf{X}\ncopies$ is a density-dependent population process and our original process is given by $\bX=\bX\onecopy$. The mean field approximation of $\bX\ncopies$ is also $\bphi\toN$.  The idea of \cite{gastSizeExpansionsMean2019} is to study the stochastic fluctuation of $X\ncopies_{(k,s)}(t)$ around its mean field approximation. The authors show that there exists a set of deterministic values $v\toN_{(k,s)}(t)$ such that
\begin{align}
  \label{eq:rmf_Ncopies}
  \esp{\mathbf{X}\ncopies_{(k,s)}(t)} = \phi\toN_{(k,s)}(\bx,t) + \frac1C v\toN_{(k,s)}(\bx,t) + O(1/C^2).
\end{align}
The values $v\toN_{(k,s)}(t)$ are shown in \cite{gastSizeExpansionsMean2019} to satisfy a system of linear ordinary differential equations whose solution can be expressed in integral form as:
\begin{align*}
v_{(k,s)}(\bx, t) = \frac{1}{2}\int_0^t\sum_{(k_1,s_1),(k_2,s_2)\atop \in \{1,\hdots,n\}\times \calS} Q_{(k_1,s_1),(k_2,s_2)}(\bphi(\bx,\tau))\dd{\phi_{(k,s)}}{(k_1,s_1)}{(k_2,s_2)}(\bphi(\bx,\tau),t-\tau)d\tau
\end{align*}
where $Q_{(k_1,s_1),(k_2,s_2)}(\bx)$ corresponds to the expected change of the covariance between the values of $X_{(k_1,s_1)}$ and $X_{(k_2,s_2)}$ of the stochastic system at some given point $\bX = \bx$. We formally introduce and elaborate more on the refinement terms in Appendix~\ref{apx:rmf_def}.

In our heterogeneous population model, we have no replica which corresponds to setting $C=1$. Hence, the above bound does not guarantee that the $O(1/C^2)$ should be small for $C=1$. The next theorem shows that, surprisingly, using the refined approximation \eqref{eq:rmf_Ncopies} with $C=1$ copy leads to an approximation that is an order of magnitude more accurate than the mean field approximation provided before. When comparing Equation~\eqref{eq:MF_theorem} and \eqref{eq:RMF_theorem}, what this theorem shows is that the correction $v\toN_{(k,s)}(x,t)$ is of order $O(1/n)$ and is the leading term of the $O(1/n)$-term of Equation~\eqref{eq:MF_theorem}. 

Note that to obtain the accuracy bound, no further assumption is needed compared to the case of the mean field approximation. That is, we assume that the parameters $r$ are uniformly bounded.
\begin{theorem}
  \label{th:RMF}
  Assume that the model, defined in Section \ref{ssec:state_representation}, satisfies \eqref{eq:bound:r}. Let $\bphi\toN(\bx, t)$ be the solution of the ODE introduced in Section \ref{ssec:drift_mean_field} with initial condition $\bX\toN(0) = \bx \in \calX\toN$ and drift $f\toN$. Let $\bv\toN(\bx, t)$ be the solution of the refinement term explicitly defined in Appendix \ref{apx:rmf_def}. Then, for $(k,s)\in \{1,\ldots,n\}\times \calS$ and $t < \infty$,
  \begin{align}
    \label{eq:RMF_theorem}
    \Proba{S_k(t)=s } & = \phi\toN_{(k,s)}(\bx,t) + v\toN_{(k,s)}(\bx,t) + O(1/n^2).
  \end{align}
\end{theorem}
\color{black}
\begin{proof}[Key ideas of the proof]
  A complete proof is provided in Section~\ref{sec:proof_rmf}. The proof of Theorem~\ref{th:RMF} uses the same methodology as the proof of Theorem~\ref{th:MF} but refines the analysis to extract the term $v$. First, we use a second-order Taylor expansion instead of the first order expansion used in the proof of Theorem~\ref{th:RMF}, this allows us to derive an expansion term in integral form. Second, we show that this expansion term is essentially equal to the refinement term $v$. Last, we show that the remainder terms are of order $O(1/n^2)$ by carefully studying how small the third and fourth derivatives of $\bphi$ with respect to its initial condition are.
\end{proof}

The theorem shows that, by adding the refinement term, the error of the refined approximation is of order $O(1/n^2)$, which is an order of magnitude better than the $O(1/n)$ of the classical mean field approximation. This implies that both equations are asymptotically exact as the number of interacting objects goes to infinity. Hence, the refinement is especially interesting to approximate systems with few interacting objects. Note that in theory, it is possible to obtain a refined-refined approximation that has an accuracy of $O(1/n^3)$. For that, one can adapt the $1/N^2$-expansion of \cite{gastSizeExpansionsMean2019} to compute a second expansion term. This expansion depends on up to the fourth derivative of $\bphi$. Yet, proving carefully that this expansion is $O(1/n^3)$-accurate seems difficult as it requires obtaining precise estimates of up to the sixth derivative of $\bphi$. Also, from a practical point of view, computing such an expansion involves solving an ODE with $O((nS)^4)$ variables which seems difficult as soon as $n$ grows. Hence, in this paper, we restrict our attention to the first expansion term.

To illustrate how this refinement improves the accuracy compared to the classical mean field, we consider the same cache replacement policy as the one studied in Figure~\ref{fig:cache_illustration}, with four objects and a cache of size $2$. Compared to the previous figure, we now added an orange curve that corresponds to the refined mean field approximation. We observe that, if the mean field approximation was good, the refined mean field approximation seems almost exact. 

\color{myorange}
In fact, the refined mean field approximation lies within the confidence interval of the sample mean which is calculated from $1000$ sample trajectories of the underlying system. It is noticeable that computing the mean field and refinement term takes about 150ms whereas simulating $1000$ sample paths and calculating the sample mean takes several seconds. \color{black} 
This suggests that for the same computational budget, the refined mean field approximation is more accurate than the simulation. We will elaborate more on that in Section~\ref{ssec:cache}. 

\begin{figure}[ht]
  \centering
  \includegraphics[width=\linewidth]{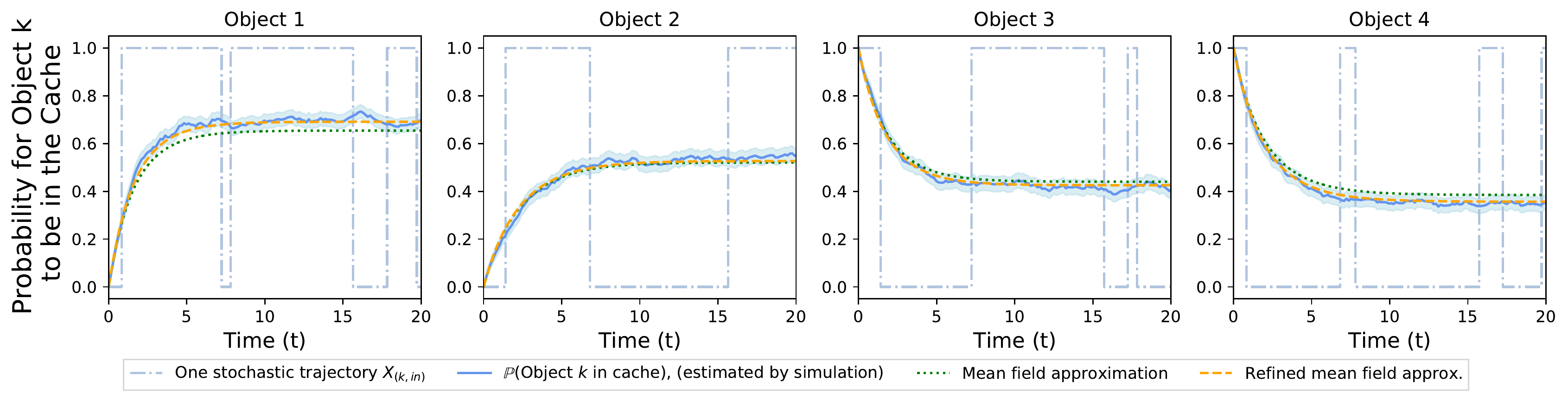}
  \caption{Behavior of the RANDOM policy for a cache of size $2$ and a total of four objects. For each of the four objects, we compare the stochastic system with the mean field and refined mean field approximations. }
  \Description{Illustration of the accuracy of the refined mean field approximation.}
  \label{fig:cache_illustration_rmf}
\end{figure}

\color{myorange}
\subsubsection*{On the applicability to partially heterogeneous systems}

Consider the case that the stochastic system is not fully heterogeneous, i.e. there exists a finite number of classes (say $n$), each class having $C$ objects that have the same behavior. Such a model corresponds to our model of $C$ copies. As this model is a density dependent process, the results from \cite{gast2017refined,gastSizeExpansionsMean2019} show that the mean field approximation is $O(1/C)$-accurate and that the refined mean field approximation is $O(1/C^2)$ accurate for the $C$-copy model.  We now show that Theorems~\ref{th:MF} and \ref{th:RMF} can be used to obtain the following much sharper bound:
\begin{align}
  \label{eq:rmf_Ncopies_refined}
  \esp{\mathbf{X}\ncopies_{(k,s)}(t)} = \phi\toN_{(k,s)}(\bx,t) + \frac1C v\toN_{(k,s)}(\bx,t) + O\p{\frac{1}{(Cn)^2}}.
\end{align}
The (only) difference between this equation and \eqref{eq:rmf_Ncopies} is that the term $O(1/C^2)$ of \eqref{eq:rmf_Ncopies} is replaced by the much smaller term $O(1/(Cn)^2)$. This implies that the accuracy of mean field interaction model does not depend on the number of homogeneous copies but only on the total number of objects. 

To see why \eqref{eq:rmf_Ncopies_refined} works, we remark that our model with $C$ copies of $n$ classes can be represented by a fully heterogeneous model with $n'=Cn$ objects (one just have to use equal rates for objects that are similar). The result of Theorem~\ref{th:MF} and \ref{th:RMF} imply that the mean field $\phi^{(n')}$ and refined mean field $\phi^{(n')}+v^{(n')}$ approximations are $O(1/n')$ and $O(1/(n')^2)$ accurate.  By replacing $n'$ by $nC$ and summing over identical objects, one obtains \eqref{eq:rmf_Ncopies_refined}.

\color{black}

\subsection{Numerical Complexity}
\label{sec:numerical_complexity}

From a computational point of view, the mean field and refined mean field approximations greatly fasten the estimation of transient or steady-state values compared to a direct study of the original Markov process $\bSn$. Indeed, the continuous-time Markov chain $\bSn$ has up to $|\calS|^n$ states where the mean field approximation can be computed by solving a non-linear ODE with $n|\calS|$ variables. As shown in Appendix~\ref{apx:rmf_def}, the refinement term $\bv\toN$ is the solution of a linear ODE with $n|\calS|+(n|\calS|)^2$ variables. This means that both approximations can be solved by using standard numerical integrators. Moreover, their complexity grows linearly (for the mean field) or quadratically (for the refinement) in the number of objects whereas an exact analysis grows exponentially with the number of objects. 

\begin{table}[ht]
  \color{myorange}
  \caption{Computation time of the mean field and refined mean field approximation for the RANDOM model as a function of the number of objects $n$.}
  \label{table:computation_time_random}
  \begin{tabular}{|c|c|c|c|c|}
    \hline
    &  \multicolumn{2}{c|}{Transient up to $T=1000$}
    & \multicolumn{2}{c|}{Steady-state (=Fixed point)}\\\hline
    $n$& mean field & refined m.f. & mean field & refined m.f.\\\hline
    10 & 30ms &180ms& 2ms &2ms\\\hline
    30 & 30ms &370ms& 3ms &5ms\\\hline
    50 & 35ms &1s& 2ms &6ms\\\hline
    100 & 60ms &14s& 3ms &30ms\\\hline
    300 & 170ms &--$^*$& 9ms &200ms\\\hline
    500 & 300ms &--$^*$& 10ms &700ms\\\hline
    1000 & 970ms &--$^*$& 30ms &9s\\\hline
  \end{tabular}\\
  ${}^*$ ``--'' means that the ODE solver did not finish before 30 seconds. 
\end{table}

{\color{myorange}
  To study the time taken to compute the mean field and refined mean field approximation in more detail, we consider the RANDOM model already presented in Figure~\ref{fig:cache_illustration_rmf} and we vary the number of objects $n$ from $10$ to $1000$. For each system size $n$, we measure the time to compute the four values described next and we report them in Table~\ref{table:computation_time_random}. 
  \begin{itemize}[wide, labelwidth=!, labelindent=0pt]
    \item (Transient) The first two columns correspond to the computation of the mean field approximation $x\toN_{(k,\cdot)}(t)$, and the refined approximation $(x+v)\toN_{(k,\cdot)}(t)$ for $t\in[0,1000]$. The computation is done by using a straightforward implementation of the ODEs given in Appendix~\ref{apx:rmf_def}, which is solved by using the function \texttt{solve\_ivp} of \texttt{scipy}. There is no particular optimization of the code to use that a large number of terms are $0$. We observe that computing the mean field approximation seems to scale linearly with $n$ and can be done for a system of more than a $n=1000$ objects. For the refined mean field approximation, the computation cost grows quickly when $n$ exceeds $100$ (it takes several minutes for $n=200$ objects).
    \item (Steady-state). The last two columns correspond to the computation of the limiting value as $t$ goes to infinity: $x\toN_{(k,\cdot)}(\infty) = \lim_{t\to\infty}x\toN_{(k,\cdot)}(t)$ and $v\toN_{(k,\cdot)}(\infty) = \lim_{t\to\infty}v\toN_{(k,\cdot)}(t)$. We observe that the computation of these values is much faster: the computation of the mean field approximation is essentially instantaneous whereas the computation of the refined mean field is doable for $n=1000$. This is since the computation of the steady-state values corresponds to finding the fixed point of a linear system of ODEs, which is done by solving a linear system.
  \end{itemize}
}

The choice between the mean field and refined mean field approximation can certainly depend on the system size as the complexity of the former grows linearly with $n$ and quadratically for the latter. Hence, for system sizes larger than $n\ge 100$, the computation time of the refined mean field increases rapidly. As shown in Theorem~\ref{th:MF}, for large $n$ the mean field approximation already gives a good estimate of the true values. This makes the refined approximation more interesting for reasonable system sizes (say $n\le100$). Note that for a homogeneous system, the complexity of computing the refined approximation does not depend on the number $C$ of replicas. For a fully heterogeneous system, this is no longer the case because the $v\toN_{(k,\cdot)}$ depends on the object's identity.


\section{Numerical Experiments}
\label{sec:numerical}

In this section, we illustrate our main results with two examples, a cache replacement model and a two-choice load balancing model. We will see that both models fulfill the requirements for Theorems~\ref{th:MF} and \ref{th:RMF} and that the hidden constant given in the theorems is small. The two examples are chosen to illustrate models for which the classical, homogeneous mean field approximation cannot be used but our heterogeneous framework applies. 

\subsection{Application to a Cache Replacement Algorithm: the RANDOM($\bm$) Model}
\label{ssec:cache}


As a first example we consider the list-based RANDOM($\bm$) cache replacement policy which was introduced in \cite{gastTransientSteadystateRegime2015,hazewinkelStochasticAnalysisComputer1987}. List-based cache replacement policies are used to order content in a cache separated into lists.  The cache is separated into $\ell:=|\calS|-1$ lists with sizes $m_1\toN,\ldots,m_\ell\toN$. When an object is requested, if it is not in the cache, it is inserted in the first list and replaces a randomly chosen object from it. If the object is in a list $s$, it is promoted to list $s+1$ and a randomly chosen object from list $s+1$ is moved to list $s$. If an object that is not in the cache is requested we call it a `miss' otherwise a `hit'. It is shown in \cite{gastTransientSteadystateRegime2015} that list-based cache replacement policies can greatly improve the hit rate compared to the classical RANDOM or LRU policies, at the price of being less responsive. The authors of \cite{gastTransientSteadystateRegime2015} used a mean field approximation for which some theoretical support was given (essentially by showing that the error of the mean field approximation is $O(1/\sqrt{n})$). In this section, we push this analysis further in two directions. First, we show that our framework improves on the bound of \cite{gastTransientSteadystateRegime2015} by showing that the error of the mean field approximation is $O(1/n)$ and not $O(1/\sqrt{n})$. Second, we show that our refined approximation provides an extremely accurate approximation (essentially more accurate than simulation). While this last fact was empirically observed in \cite{casale2020performance}, Theorem~\ref{th:RMF} provides theoretical support by showing that the error of the refined approximation is $O(1/n^2)$. 

\subsubsection{Model and Approximations}

We consider that there are $n$ objects with identical sizes. Requests for an object $k$ arrive according to a Poisson process of intensity $\lambda_k$. In our framework, the state of object~$k$ at time $t$ is $\Sn_k(t)\in\{0,1\dots \ell\}$, that represents the list in which  object~$k$ is ("0" means that the object is not stored in the cache). Following our framework, we denote by $X_{(k,s)}(t)$ the random variable that equals $1$ if object $k$ is in list $s$ and $0$ otherwise.  According to the RANDOM($\bm$) policy, if object~$k$ is in list $s\in \{0,1, \ldots, \ell-1\}$ and gets requested, then it is moved into list $s+1$ and a randomly selected object (say $k_1$) from list $s+1$ moves into list $s$.  The corresponding transitions for the Markov chain $\bX$ are: 
\begin{align}
  \bX \mapsto \bX + \be_{(k,s+1)} - \be_{(k,s)} + \be_{(k_1,s)} - \be_{(k_1,s+1)} & \qquad \text{ at rate } \frac{\lambda_k}{m_{s+1}\toN} X_{(k,s)}X_{(k_1,s+1)}. \label{eq:rand_m:transition}
\end{align}
Here, $\lambda_kX_{(k,s)}$ is the rate at which object $k$ is requested while being in list $s$, and $X_{(k_1,s+1)}/m_s\toN$ is the probability that object $k_1$ is in list $s+1$ and is chosen to be exchanged.

The drift of the stochastic system is the vector $(f_{(k,s)}(\bx))_{(k,s)}$. By definition, for object~$k$ in list $s$, the $(k,s)$ component of the drift is 
\begin{align*}
f_{(k,s)}(\bx) = & \lambda_k x_{(k,s-1)} - \sum_{k_1=1}^n\frac{\lambda_{k_1}}{m_s} x_{(k_1,s-1)}x_{(k,s)} + \bigl( \sum_{k_1=1}^n \frac{\lambda_{k_1}}{m_{s+1}} x_{(k_1,s)}x_{(k,s+1)} - \lambda_k x_{(k,s)} \bigr)\mathbf{1}_{\{s < \ell\}}.
\end{align*}
The mean field approximation is the solution of the ODE $\dot{\bx}=f(\bx)$. Note that this ODE is the same as the one given in \cite{gastTransientSteadystateRegime2015}.

The transitions of this model are pairwise interactions with $b\toN_{k,k_1,s,s+1,s+1,s} = \frac{\lambda_k}{m\toN_{s+1}/n}$ as defined in \eqref{eq:interactions}. There is no unilateral transition. To apply Theorems~\ref{th:MF} and \ref{th:RMF}, we assume that the list sizes $m_i\toN$ grow linearly with $n$ which guarantees that $b$ remains bounded by $\bar{b} = \frac{\max_k \lambda_k}{\min_s (m_s/n)}$. Therefore, the assumptions of the theorems are satisfied as soon as the values $\lambda_k$ are bounded and the list sizes grow linearly with the number of objects.  This guarantees that the mean field approximation is $O(1/n)$ accurate whereas the refined mean field approximation is $O(1/n^2)$ accurate.

\subsubsection{Transient Analysis}

We calculate the solution of the ODEs for the mean field and refined mean field approximation and compute the simulations by adapting the toolbox \cite{gastRefinedMeanField2018}. We implement the Markov chain, the drift $f$, the drift derivatives and $\bQ$ for the approximations based on the transitions \eqref{eq:rand_m:transition}. For our numerical example, we consider a cache with $n=20$ objects for which the request rates follows a Zipf distribution with parameter $\alpha=0.8$, that is, $\lambda_k = A / k^{\alpha}$ with $A$ being a normalizing constant. We consider a cache with three lists of sizes $m_1 = 5$, $m_2 = 3$, and $m_3 = 2$. 

In Figure~\ref{fig:transient_cache_popularities} we compare the mean field and refined mean field approximations of the cache popularities, \emph{i.e.}, $\sum_{k=1}^n \lambda_k x_{(k,s)}$ and $\sum_{k=1}^n \lambda_k (x_{(k,s)} + v_{(k,s)})$ for $s=0,1,\ldots, 3$, against the ``true'' value $\sum_{k=1}^n \lambda_k \E[X_{(k,s)}], \ s=0,1,\ldots, 3$ that is estimated by simulation. We compute the sample mean and the 95-percent confidence interval of the cache popularities by running 2000 Markov chain simulations. This figure shows that the mean field approximation captures the qualitative behavior of the stochastic process very well. Quantitatively, the mean field provides a good approximation but does not accurately capture the behavior of the system, especially for the third list. The values of the refined mean field approximation give a considerably better approximation. It lies within the 95-percent confidence interval of the sample mean and seems to be almost exact. 

Yet, evaluating how precise the refined approximation is difficult since it lies within the confidence interval of the simulation. To study this error in more detail, next we study the steady-state behavior of the cache, for which an exact analysis is doable when $\bm$ is small enough.

\begin{figure}[ht]
  \includegraphics[width=\textwidth]{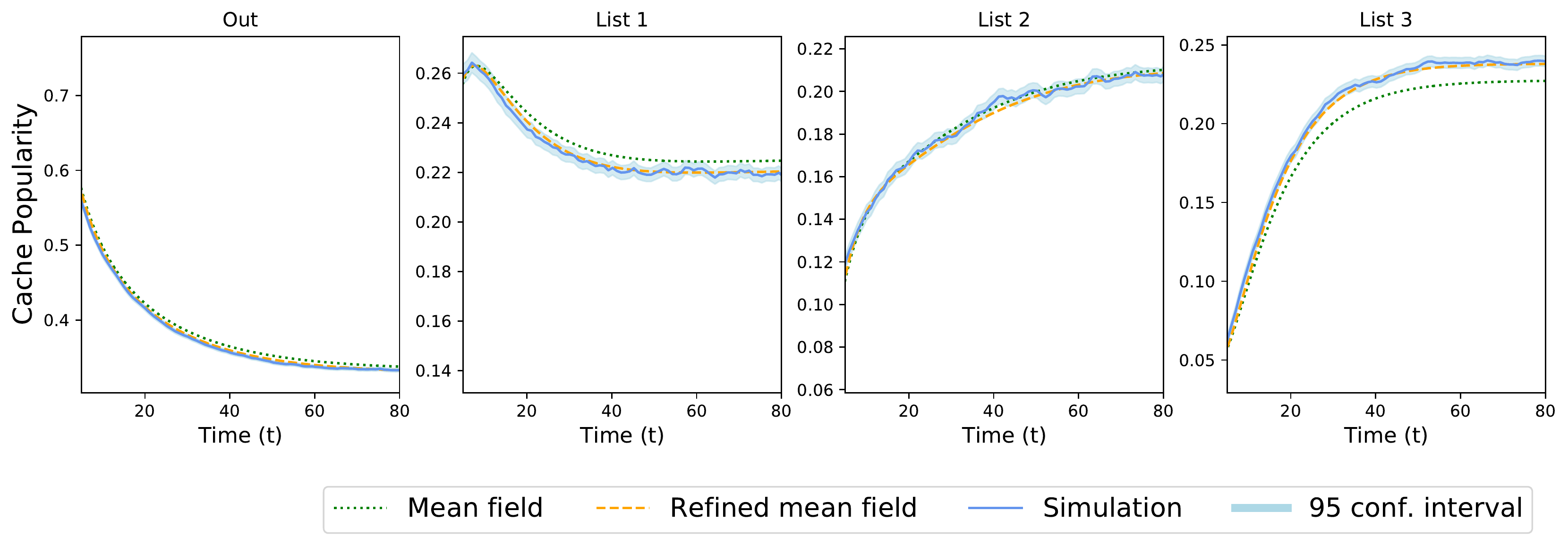}
  \caption{Transient state comparison of cache popularities.}
  \label{fig:transient_cache_popularities}
  \Description{Comparison of the mean field and refined mean field approximated values against a simulated mean estimate in transient regime.}
\end{figure}

\subsubsection{Steady-state Analysis}

The previous results show that the mean field and refined mean field can accurately approximate the transient behavior of the RANDOM($\bm$) policy. In Figure~\ref{fig:steady_state_cache}, we compare the steady-state values of the simulation, mean field and refined mean field approximation against an exact solution. To make the figure visible, we consider a case with $8$ objects having a Zipf popularity with parameter $0.5$ and three lists of sizes $2$. This figure shows that even for the steady-state, the mean field and refined mean field approximation are very good estimates of the true mean. As for the transient regime, this figure shows that the refined mean field approximation captures the cache popularities more closely than the mean field approximation: the curve provided by simulation and by refined mean field approximation are almost indistinguishable. Note that the bound obtained in Theorem~\ref{th:MF} and \ref{th:RMF} are only for the transient regime. We believe that obtaining a similar bound for the steady-state is possible but requires to precisely control how fast the mean field approximation converges to its fixed point. We leave this for future work.

\begin{figure}[ht]
\includegraphics[width=\linewidth]{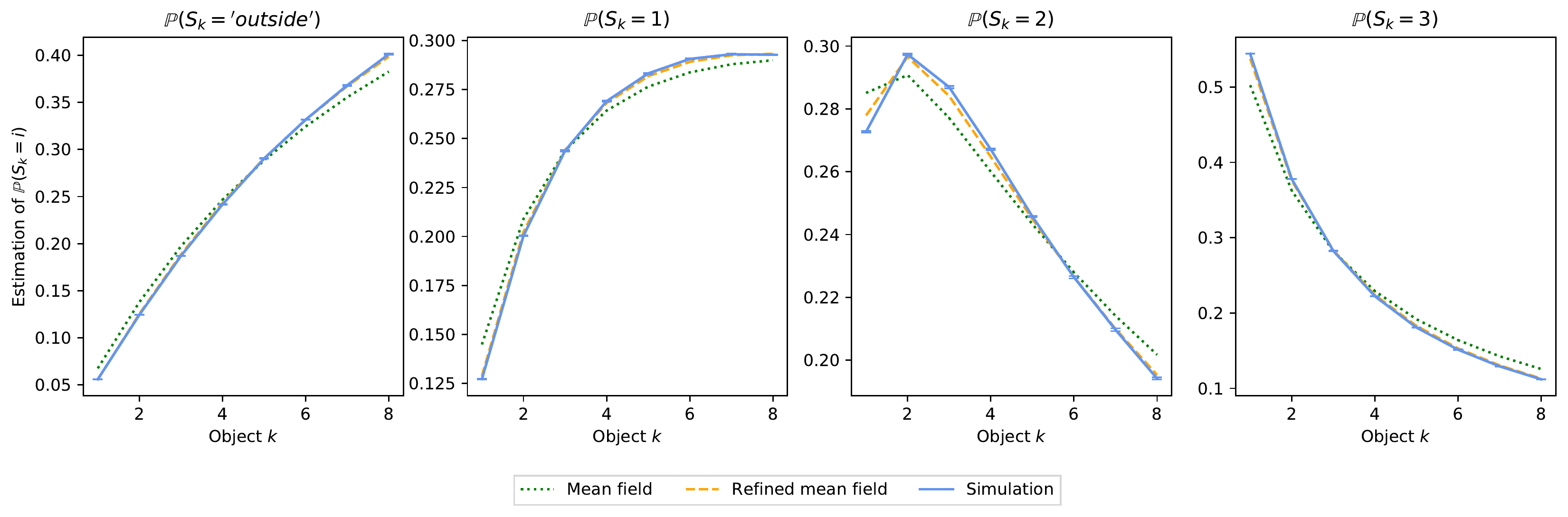}
  \caption{Steady-state probabilities estimated by simulation, mean field and refined mean field approximation.}
  \label{fig:steady_state_cache}
\end{figure}

While the previous figure suggests that the refined mean field is extremely accurate, it does not give a precise idea of how accurate the approximation is. To go one step further, we consider a cache model with $n$ different objects following a Zipf popularity with parameter $\alpha=0.5$, and a cache with two lists of size $m_1=m_2=\floor{0.3n}$. We study the accuracy of the mean field and refined mean field approximation as $n$ grows. One difficulty to do so is that when the number of objects $n$ is large, obtaining an accurate simulated estimation of $\Proba{S_k=s}$ for all $(k, s) \in\{1\dots n\}\times \calS$ is difficult. As we show below, the refined mean field seems more accurate than the simulation as soon as $n$ is more than $20$. 

We show in Appendix~\ref{apx:cache_adaptation} that one can use the product form of the steady-state distribution to obtain a recurrence equation for the steady-state probability of $\Proba{S_k=s}$. While the complexity of computing this is quite large for large caches (our implementation does not allow us to compute it for more than $3$ lists of size $10$), it is possible to compute the exact steady-state distribution for relatively small values of $\bm$. We call this value $\pi_{(k,s)}^{\mathrm{exact}}$. We also compute an estimation $\pi_{(k,s)}^{\text{method}}$ for each $\text{method}\in\{\text{mean field}, \text{refined mean field}, \text{simulation}\}$. For the estimate computed by using simulations, we simulated $10^{8}$ requests and estimate the steady-state probability after a warp-up period of $10^7$ requests. The average error of a method is defined as
\begin{align}
  \label{eq:cache_error}
  \mathrm{Error}(\text{method}) &= \frac1{n}\sum_{k,s} \abs{\pi_{(k,s)}^{\text{method}} - \pi_{(k,s)}^{\text{exact}}}.
\end{align}

We report in Table~\ref{tab:cache_steady-state} the error of the three estimation methods (mean field, refined mean field and simulation). By Theorem~\ref{th:MF} and \ref{th:RMF}, we expect the average error of the mean field to be of order $O(1/n)$ and the error of the refined mean field to be of order $O(1/n^2)$. This is what we observe in Table~\ref{tab:cache_steady-state}, in which we also show the error multiplied by $n$ or $n^2$ (depending on the method), to emphasize the convergence rate. We also observe that when $n$ is larger than $20$, the simulation makes more errors than the refined mean field. Note that, the value obtained by simulation is an unbiased estimator of the true value, and the error that we report arises because we can only simulate a finite number of requests. For our simulation, we choose $10^8$ requests to have a reasonably fast method (it takes between $5$ and $10$ seconds to simulate the $10^8$ requests by using an optimized C++ simulator). {\color{myorange} As a matter of comparison, we also indicate in parenthesis the time taken by our implementation to compute the mean field and refined mean field approximation. Since we consider a relatively small system (at most $n=50$ heterogeneous objects), the computation of the refined mean field approximation is fast (less than $30$ms).  This shows for $n\ge30$, the refined mean field is much more accurate than the simulation, while being much faster to compute.

\begin{table}[tb]
  \caption{Average per-object error of three estimation methods: mean field, refined mean field and simulation. {\color{myorange} We also indicate in parenthesis the time taken to compute these numbers. }}
  \label{tab:cache_steady-state}
  \begin{tabular}{|c|cco|cco|co|}
    \hline
    & \multicolumn{3}{c|}{Mean field}
    & \multicolumn{3}{c|}{Refined mean field}
    & Simulation&\\\hline
    $n$ & Error & $n\ \times$ Error & (time) & Error & $n^2\ \times$ Error & (time) & Error & (time)\\\hline
    10 & 0.0142 &0.142 &(10ms) &0.00197 &0.197 &(10ms)& 0.00026&(4.3s)\\\hline
    20 & 0.0074 &0.149 &(11ms) &0.00049 &0.197 &(13ms)& 0.00043&(4.6s)\\\hline
    30 & 0.0050 &0.151 &(14ms) &0.00022 &0.196 &(17ms)& 0.00047&(4.9s)\\\hline
    40 & 0.0038 &0.153 &(13ms) &0.00012 &0.195 &(22ms)& 0.00055&(6.1s)\\\hline
    50 & 0.0031 &0.154 &(17ms) &0.00008 &0.193 &(30ms)& 0.00055&(5.7s)\\\hline
  \end{tabular}
\end{table}

The time taken to simulate a system of $n$ objects grows with the number of heterogeneous objects $n$ because the complexity of sampling from a Zipf distribution with $n$ object grows with $n$. Yet, this additional computation cost is low (sampling from $n$ objects can be done in $O(\log n)$).  For the refined mean field, the situation is different and the computation time might be large for high values of $n$. Our experiment suggests that it is possible to compute a refined mean field approximation for a few hundred objects in a relatively fast time (less than $5$ seconds). Note that for $n=500$ objects, the error of the refined mean field approximation is in theory $100$ times smaller than the error reported in Table~\ref{tab:cache_steady-state} for $n=50$ (\emph{i.e.}, extremely small). 
} 



\subsection{Application to a Load Balancing Algorithm: The two-choice Model}
\label{ssec:lbm}

\subsubsection{Model and Approximations}

In our second example, we consider a variation of the well studied two-choice model \cite{mitzenmacherPowerTwoChoices2001}. In contrast to the homogeneous case, where all servers have equal service rate parameters, we consider a heterogeneous setup in which processors can have different speeds. We study the impact of the heterogeneity of servers on their performance. Note that a similar analysis was done in \cite{mukhopadhyayAnalysisLoadBalancing2015} by using two classes of servers. The purpose of this example is twofold. It illustrates that our framework can incorporate load balancing models with heterogeneous servers. It also shows that taking heterogeneity into account in such systems is important if one wants to characterize the performance precisely.

\color{myorange}
The model consists of $n$ servers with heterogeneous service rate parameters $\mu_k, k=1,\ldots,n$ and a finite buffer of size $b$, including the job in service. Jobs arrive according to a Poisson process of rate $\lambda n$, and we call $\lambda$ the arrival rate. For each incoming job, we randomly pick two servers. The job is then assigned to the server which has the least number of unfinished jobs. If both servers have the same queue length and a full buffer, the job is discarded. Otherwise, the assignment between the two servers is done at random. The service time of a job in the queue of server $i$ is exponentially distributed with mean $\mu_i$. The state of a server $k$ at time $t$ is its queue length $S_k(t) \in \calS= \{0,1,\ldots, b\}$, state $0$ is referring to the idle state. 

We denote by $X_{(k,s)}$ the random variable that equals $1$ if server $k$ has $s$ jobs. The process $\bX=(X_{(k,s)})_{(k,s)}$ is a Markov chain whose transitions are (for all $k,k_1\in\{1\dots n\}$) \color{black}:
\begin{subequations}
  \begin{align}
    &\bX \mapsto X - e_{(k,s)} + e_{(k,s-1)} && \text{ at rate } \mu_k X_{(k,s)}, \label{eq:lbm:removal_transition}\\
    &\bX \mapsto X + e_{(k,s+1)} - e_{(k,s)} && \text{ at rate } (2\lambda n \mathbf{1}_{\{s_1 \geq s+1\}} +  \lambda n \mathbf{1}_{\{s_1 = s\}}) \frac{X_{(k,s)}}{n} \frac{X_{(k_1,s_1)}}{n}. \label{eq:lbm:arrival_transition}
  \end{align}
\end{subequations}
\color{myorange} In the above equation, the first type of transition \eqref{eq:lbm:removal_transition} corresponds to the completion of a job by server~$k$ when the job queue is of size $1 \leq s \leq b$. It reduces the queue length from $s$ to $s-1$ which sets $X_{(k,s)}$ to $0$ and $X_{(k,s-1)}$ to $1$. The second type of transitions, in equation \eqref{eq:lbm:arrival_transition}, corresponds to adding a job to server $k$ having $0 \leq s \leq b-1$ jobs in the buffer. In this case, the queue size is increased by one from $s$ to $s+1$. To explain the transition rate we see that the servers~$k, k_1$ can be selected in two ways, by selecting $k$ or $k_1$ first and the other second. In the case that both queues have equal length, the chance to add the job to server~$k$ is $1/2$. If both buffers are full, the job is discarded. 

This model has both unilateral transitions with $r_{k,(s) \rightarrow (s-1)}\toN = \mu_k$ and pairwise interactions with $r_{k,k_1,(s,s_1)\rightarrow (s+1,s_1)}\toN = (2\lambda \mathbf{1}_{\{s_1 \geq s+1\}} +  \lambda \mathbf{1}_{\{s_1 = s\}})/n$. The bound \eqref{eq:bound:r}, required to apply Theorems \ref{th:MF} and \ref{th:RMF}, is verified when the values of $\lambda$ and $\mu_k$ are bounded independently of $n$. Moreover, we assume a possibly very buffer for the servers. \color{black}

To simplify notations, let $g_{s}(t)= \sum_{k=1}^{n} \sum_{\tilde{s} \geq s} X_{(k,\tilde{s})}(t)/n$ be the fraction of servers at time $t$ with queue length at least $s$. By summing over all possible values of $k_1$ and $s_1$, the transition \eqref{eq:lbm:arrival_transition} can then be rewritten as 
\begin{align*}
  &X \mapsto X + e_{(k,s+1)} - e_{(k,s)} && \text{ at rate } \lambda \ X_{(k,s)} (g_{s} + g_{s+1}).
\end{align*}
By using this notation, the drift for index $(k,s)$ is
\begin{align*}
  f_{(k,s)}(\bx) = & \mu_k x_{(k,s+1)} - \lambda x_{(k,s)}(g_s+g_{s+1}) - \bigl( \mu_k x_{(k,s)}  - \lambda x_{(k,s-1)} (g_{s-1} + g_{s}) \bigr)\mathbf{1}_{\{s \geq1\}}.
\end{align*}

\subsubsection{Numerical Comparison}

As for the caching example, we adapt methods of the toolbox \cite{gastRefinedMeanField2018} to perform a numerical comparison of mean field and refined mean field approximation against an estimation of the expected value of the system. Following the equations  \eqref{eq:lbm:removal_transition} and \eqref{eq:lbm:arrival_transition} we implement the Markov chain and define the drift $f$, the derivatives of $f$ and the tensor $\bQ$. To obtain the plots we consider models with systems sizes of $n = 10, 20, 30, 40$ and an arrival rate $\lambda = 1$. The heterogeneity is introduced by the consideration of differing server rates. For every system size, we consider a model having service rates as follows. One fifth of the server rates is equal to $2.0$, one fifth is equal to $0.5$ and the remaining rates are sampled uniformly between $1.0$ and $1.4$. In transient state we calculate the sample mean for the system sizes by averaging over $2000$ simulations for $n=10,20$, and over $3000$ simulations for $n=30, 40$. For the steady-state, the estimations are computed by calculating the independent time-average of $19\times10^6$ events of the Markov chain after a warp-up of $5 \times10^5$ events. To compute the mean field and refined mean field approximation faster, we restricted the queue size of the system to a maximum of $12$ (for simulation, we assume unbounded queue lengths). This is justified by two facts: First, the refined mean field seems to be very accurate even with this bounded queue size. Second, we also show in Figure \ref{fig:queue_len_distr} that the queue length distribution vanishes very quickly for high queue sizes. We collect all simulation results in Figures \ref{fig:average_queue_size_fig}, \ref{fig:mean_error_comp_het_vs_hom} and \ref{fig:queue_len_distr}.

\begin{figure}[ht]
  \includegraphics[width=\textwidth]{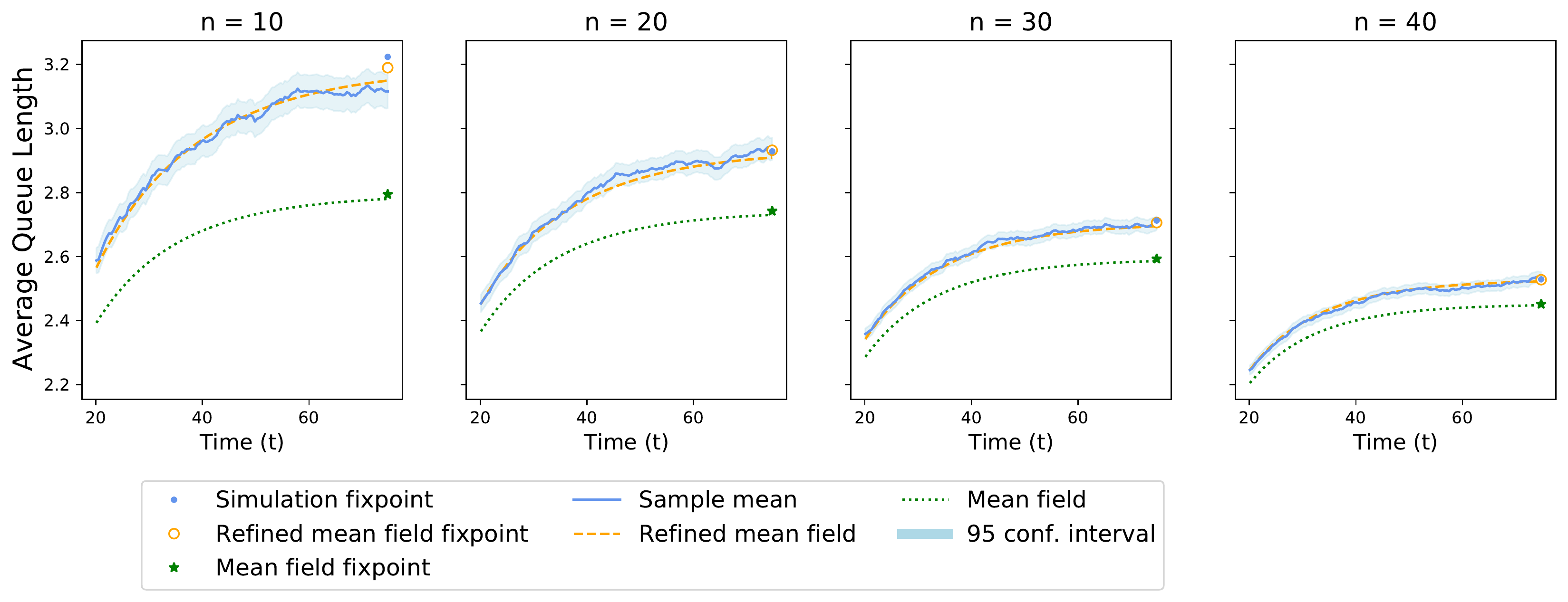}
  \caption{Average Queue Size of simulation mean vs. mean field vs. refined mean field approximation.}
  \label{fig:average_queue_size_fig}
  \Description{Comparison of the average queue size of the system approximated by mean and refined mean field against a simulated estimation of the true value.}
\end{figure}


Figure~\ref{fig:average_queue_size_fig} shows the average queue size of the system. We plot the sample mean of the average queue size with a 95-percent confidence interval against the average queue size calculated from the mean field and refined mean field approximation. We observe that as $n$ grows, both the mean field and the refined mean field approximations seem to be asymptotically exact. Note that the mean field approximation depends on $n$ because it depends on the exact server speeds.  Also, in all cases, the mean field approximation underestimates the average queue length, whereas the refined mean field approximation lies within the confidence interval. On each plot, we also show the steady-states estimates (as a single point on the right of each panel). The observation is the same as for the transient regime: the refined mean field approximation is extremely accurate also for the steady-state regime.

To demonstrate the impact of heterogeneity, we also consider an approximation (that we call the "homogeneous" approximation) in which there are $n$ servers with speed $\bar{\mu}=(\sum_k\mu_k)/n$. We consider the corresponding mean field and refined mean field approximation. For these four approximation methods, we denote by $\mathrm{Error}(\mathrm{method})=\frac1n \sum_{(k,s)}\abs{\esp{X_{(k,s)}(\infty)} - \pi^{\mathrm{method}}_{(k,s)}}$ the mean error, where $\esp{X_{(k,s)}(\infty)}$ is the steady-state of the stochastic system, approximated by simulation, and $\pi^{\mathrm{method}}_{(k,s)}$ is the estimation of the steady-state probability for the given method. We plot these four errors as a function of $n$ in the Figures \ref{fig:mean_error_comp_het_vs_hom_large_server_var} and \ref{fig:mean_error_comp_het_vs_hom_small_server_var}. The setup of the first figure is as described before, one fifth of the servers are of speed $2.0$, one fifth are of speed $0.5$ and the remaining are uniformly chosen between $1.0$ and $1.4$. For the second figure all server speeds are uniformly chosen between $1.0$ and $1.4$. 
We observe that, as expected, the error of the heterogeneous mean field and refined mean field approximation decrease with $n$ (at rate $O(1/n)$ and $O(1/n^2)$) while the error of the homogeneous mean field or refined mean field does not improve much with $n$. This indicates that taking heterogeneity into account is necessary to obtain accurate performance metrics in any case. We also see that for larger variance in the server rates, i.e. stronger heterogeneity, the error of the homogeneous approximation increases whereas our heterogeneous approach gives good estimates.

\begin{figure}[ht]
\begin{subfigure}[b]{0.47\linewidth}
  \includegraphics[width=\linewidth]{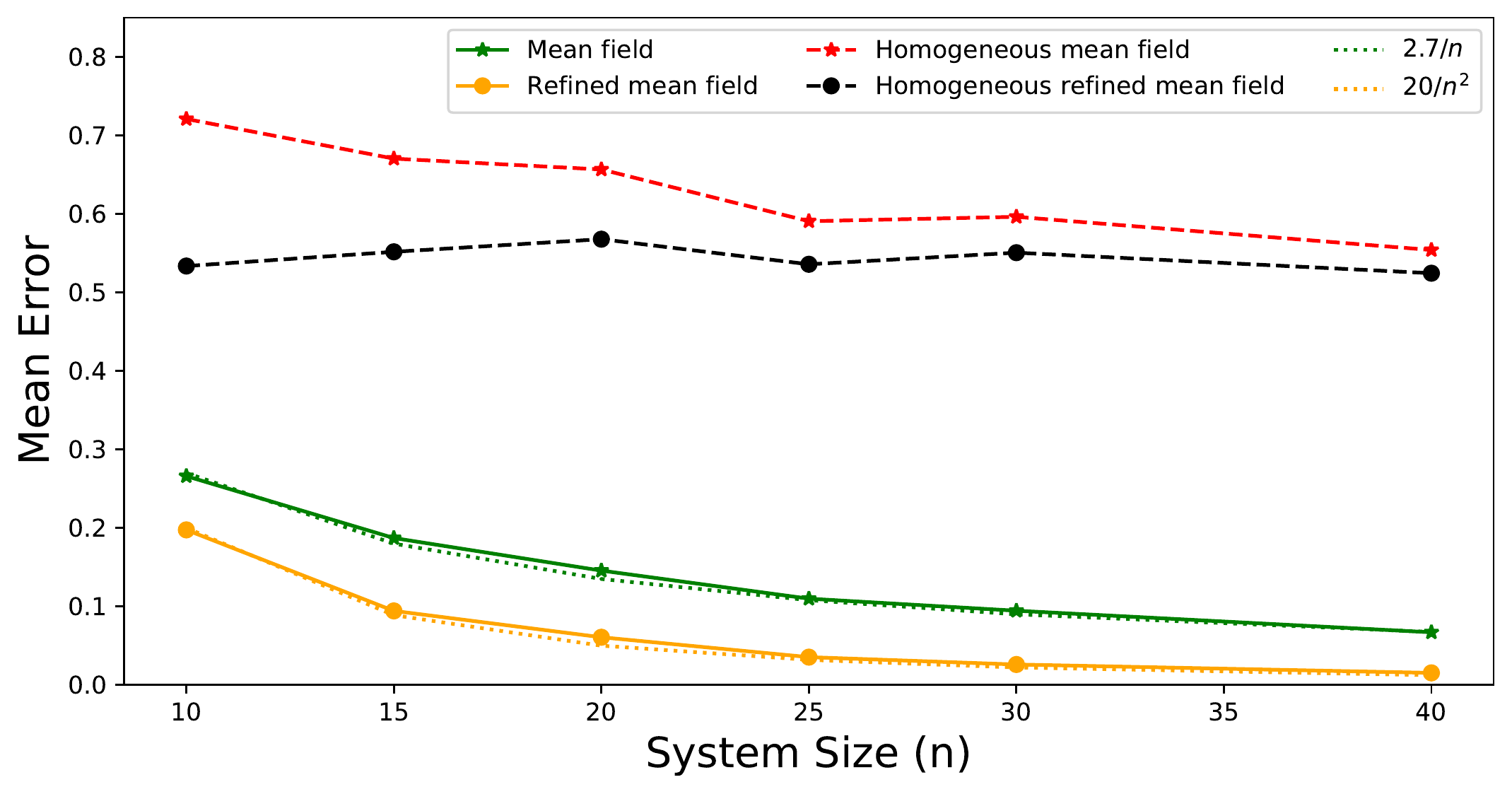}
  \caption{Mean error for strongly varying server rates.}
  \label{fig:mean_error_comp_het_vs_hom_large_server_var}
\end{subfigure}
\hfill
\begin{subfigure}[b]{0.47\linewidth}
  \includegraphics[width=\linewidth]{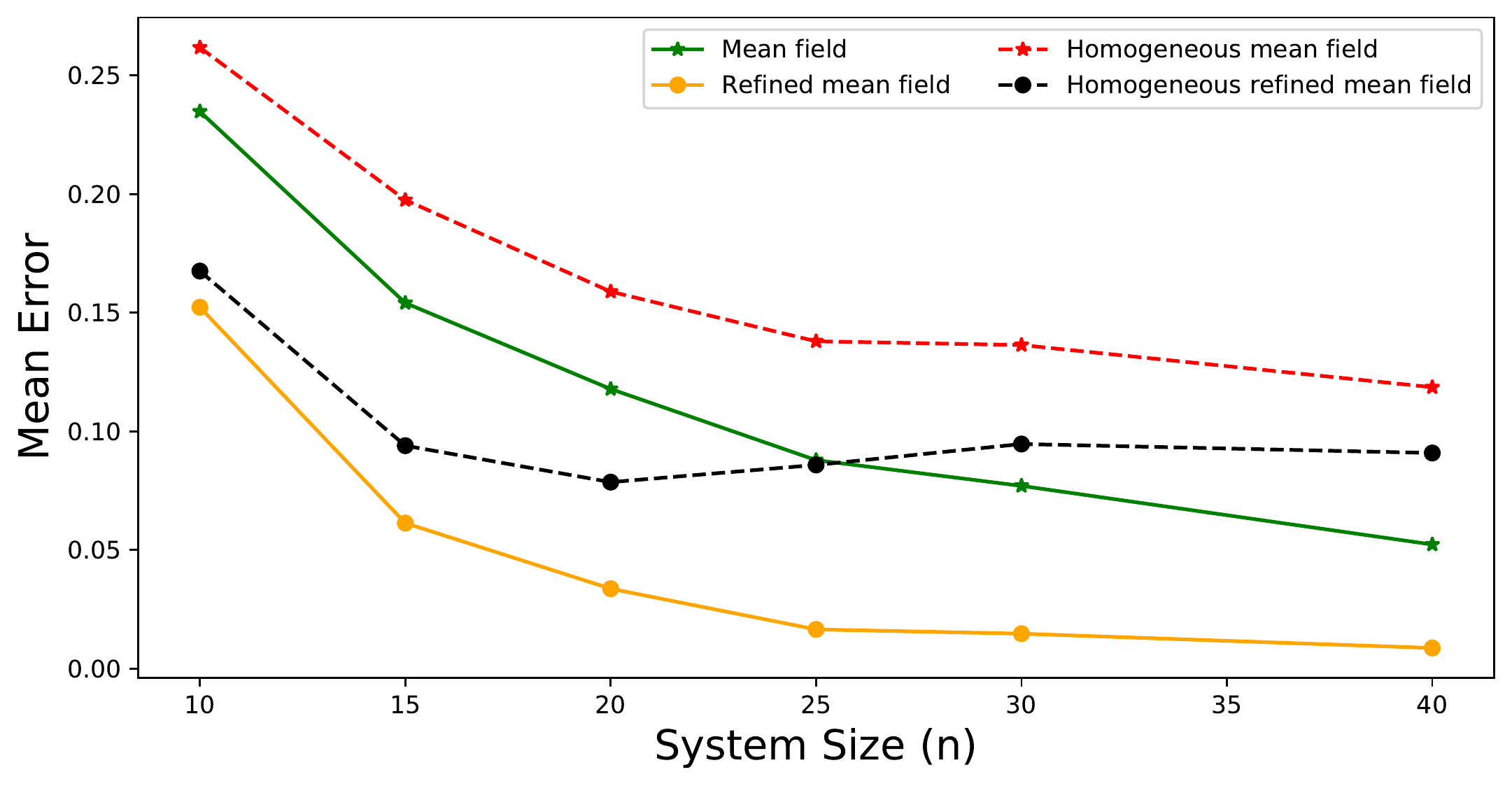}
  \caption{Mean error for lightly varying server rates.}
  \label{fig:mean_error_comp_het_vs_hom_small_server_var}
  \Description[Fly 1 and Fly 2 look identical]{}
\end{subfigure}
\caption{Steady-state mean error comparison of heterogeneous and homogeneous models.}
\label{fig:mean_error_comp_het_vs_hom}
\Description{The plots shows the mean error of the of steady-state approximation error of the heterogeneous and homogeneous mean and refined mean field approximation for different system sizes. The first plot shows the error comparison for systems with strong variation of server rates. The second show the same error for less varying server rates.}
\end{figure}

Last, in Figure \ref{fig:queue_len_distr} we plot the queue length distribution tail. We plot $\frac1n\sum_{k}\Proba{S_k\ge s}$, the probability that a server picked at random has a queue length larger than $s$ as a function of $s$.  The top panel is in normal scale whereas the bottom figure is in log-scale, to zoom on the tail. We observe that for all system sizes the mean field and the refined mean field predict the shape of the distribution well. Yet, they both underestimate the actual tail distribution. The refined approximation improves notably upon the mean field method for ``small'' $s$. It does not fully correct the tail distribution for large $s$. Note that a similar observation was made in \cite{gastRefinedMeanField2017} for the refined mean field for homogeneous systems. 

\begin{figure}[ht]
  \begin{subfigure}{\linewidth}
    \includegraphics[width=0.95\textwidth]{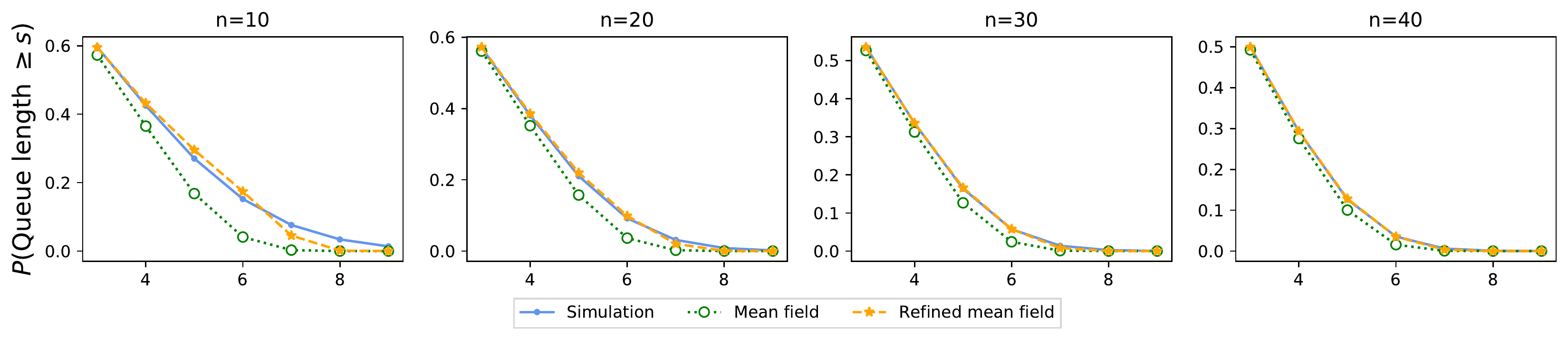}
    \caption{Tail-distribution of queue length (normal-scale)}
  \end{subfigure}
  \begin{subfigure}{\linewidth}
    \includegraphics[width=0.95\textwidth]{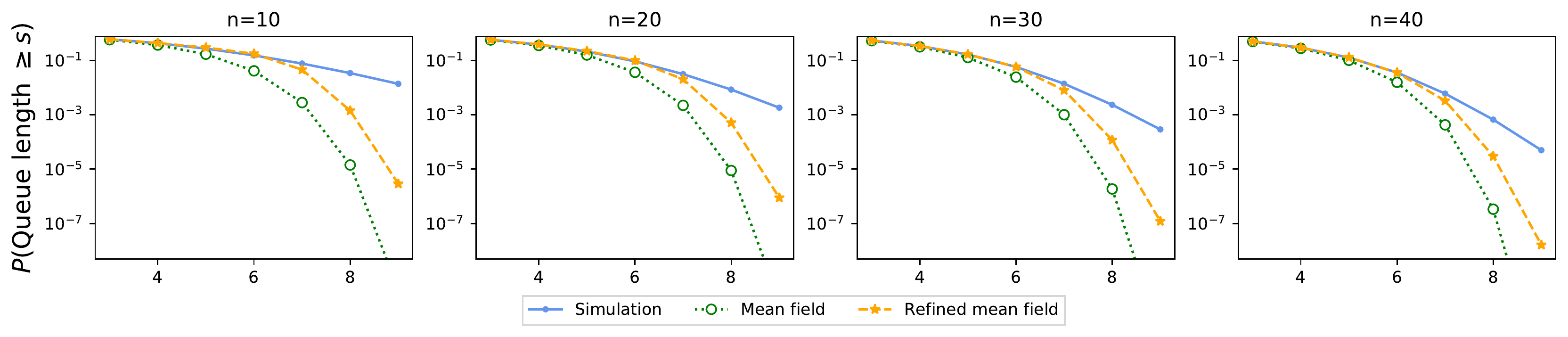}
    \caption{Same figure with log-scale}
  \end{subfigure}
  \caption{Queue length distribution of simulated expectation, mean field approximation and refined mean field approximation in steady-state for system sizes $n=10,20,30,40$. }
  \label{fig:queue_len_distr}
  \Description{The two figures show the queue length tail distribution given by the approximation methods and simulation in normal- and log-scale. Both figures show the probability of the system having servers with queue length larger than three, four, ... . We plot the probabilities for increasing system sizes.}
\end{figure}


\section{Proofs}
\label{sec:proofs}

This section contains the proofs of the main theorems. After recalling some notations in Section~\ref{sec:proof_notation}, we start with a first technical lemma in Section~\ref{sec:proof_lemma} in which we show that the difference between the stochastic and deterministic systems depends on the difference between the generator of the stochastic systems and the one of the ODE. Then, we prove Theorem~\ref{th:MF} in Section~\ref{sec:proof_mf} and Theorem~\ref{th:RMF} in Section~\ref{sec:proof_rmf}. To ease the reading, some technical lemmas -- whose proof are not complicated but long and technical -- are postponed to the appendix.

\subsection{Notation}
\label{sec:proof_notation}

In all the proofs, to ease the reading, we drop the superscript $n$. It should be kept in mind that all quantities $\bX$, $f$,... depend on $n$. Also, instead of indexing the vectors by a pair $(k,s)$, we will use an index $i\in\calI$, where $\calI=\{1\dots n\}\times\calS$ is the set of object-state pairs. For a function $h:\calX\times \R^+\to\R$, we denote by $D_\bx h$ the derivative of $h$ with respect to the first coordinate $\bx$ and by $D_th$ the derivative with respect to the second coordinate. This means that for a given pair $(\by,s)\in\calX \times \R^+$, the quantity $D_\bx h(\by,s)$ and $D_t h(\by,s)$ are the derivatives of $h$ with respect to $\bx$ and $t$ evaluated at the point $(\by,s)$.

For convenience, we will denote by $K_{\bx,\bx'}^{(n)}$ the rate at which the Markov chain $\bX$ jumps from $\bx$ to $\bx'$ for $\bx,\bx'\in\calX$. With this notation, for intuition, the transitions \eqref{eq:unilateral} and \eqref{eq:interactions} correspond to (for $\bx\in\calX$ and $k,k_1\in\{1,\dots, n\}$, $s,s',s_1,s'_1\in\calS$ with $(s,s_1) \ne (s',s_1')$)
\begin{align*}
  K_{\bx, \bx + \be_{(k,s')}-\be_{(k,s)}}\toN &= r_{k,(s)\rightarrow(s')} X_{(k,s)}\toN, \\
  K_{\bx, \bx + \be_{(k,s')}-\be_{(k,s)} + \be_{(k_1,s_1')}-\be_{(k_1,s_1)}}\toN &= \frac1n r_{k,k_1,(s,s_1)\rightarrow(s',s'_1)} X_{(k,s)}X_{(k_1,s_1)}.
\end{align*}

\subsection{Comparison of the Generators}
\label{sec:proof_lemma}

\begin{lemma}
  \label{lemma:generator_comparison}
  Let $\bX(t)$ be the continuous time Markov chain defined in Section~\ref{ssec:state_representation}. Let $\bphi(\bx, t)$ be the value at time $t$ of the solution of the ODE $\frac{d}{dt}\bphi(\bx, t) = f(\bphi(\bx, t))$ with initial condition $\bx\in\calX$. We have
  \begin{align*}
  \esp {\bX(t) - \bphi(\bX(0),t))} & = \int_0^t \E\Big[ \sum_{x'\in \calX} K_{\bX(\tau),x'} \bigl( \bphi(\bx',t-\tau) -  \bphi(\bX(\tau),t-\tau) \bigr) \\
  & \qquad\qquad - D_\bx \bphi(\bX(\tau),t-\tau)f(\bX(\tau)) \Big] d\tau.
  \end{align*}
  \end{lemma}
  
  \begin{proof} \color{myorange}
The following calculations are based on the ideas used in the proofs of \cite[Theorem 1]{kolokoltsovMeanFieldGames2012} and \cite[Theorem 3.1]{gast2017expected}. By defining $\bpsi(\tau) = \esp{\bphi(\bX(\tau),t-\tau)}$ we can rewrite $\E[ \bX(t) - \bphi(\bX(0),t)] = \bpsi(t) - \bpsi(0)$. At first, we derive the time derivative of $\bpsi$. We start by looking at the expected change at a given time $\tau$, $\frac{d}{ds} \bigl. \esp{ \bphi(\bX(\tau + s),t - (\tau + s)) \mid X(\tau)}\bigr|_{s=0}$, which can be written as
\begin{align*}
&  \lim_{ds \downarrow 0} \frac{1}{ds} \bigl( \esp{ \bphi(\bX(\tau + ds),t - (\tau + ds)) \mid X(\tau)} - \bphi(\bX(\tau),t - (\tau + ds))\\
& \phantom{++++++++} + \bphi(\bX(\tau),t - (\tau + ds)) - \bphi(\bX(\tau),t - \tau) \bigr).
\end{align*}
In the limit, the first difference corresponds to the generator of $\bX$ at $\tau$ and the second difference to the change of $\bphi$ due to the decrease of $t-\tau$. By taking the expectation and explicitly writing the limit terms, the derivative of $\bpsi$ is 
\begin{align*}
  \frac{d}{d\tau}\bpsi(\tau) = \E\Big[\sum_{x'\in\calX} K_{\bX(t),\bx'} \left( \bphi(\bx', t-\tau) - \bphi(\bX(\tau),t-\tau) \right) - D_t \bphi(\bX(\tau), t-\tau)\big].
\end{align*}
Note that by definition of the stochastic process the derivative with respect to time and the expectation are interchangeable. i.e. $\frac{d}{ds}\esp{\esp{\bphi(\bX(\tau + s),t-(\tau + s))\mid \bX(\tau)}\bigr|_{s=0}}$ is equal to \\ $\esp{\frac{d}{ds} \esp{\bphi(\bX(\tau + s),t-(\tau + s))\mid \bX(\tau)}\bigr|_{s=0}}$. As $\bphi(\bx,\cdot)$ is the solution of the ODE starting in $\bx$ at time $0$, we use\footnote{To see why, for $t,s\ge0$, the solution of the ODE satisfies $\frac{d}{ds}\bphi(\bx, t+s) = \frac{d}{ds}\bphi(\bphi(\bx,s), t)$. This shows that $\frac{d}{ds}\bphi(\bx, t+s) = \frac{d}{ds}\bphi(\bphi(\bx,s), t) = D_\bx\bphi(\bphi(\bx,s),t)f(\bphi(\bx,s),t)$. Evaluating this expression at time $s=0$ gives the result. Note that by definition, one also has $D_t\bphi(\bx,t)=f(\bphi(x,t))$ but the latter is hard to use in the analysis.} that $D_t\bphi(\bx,t)=D_\bx\bphi(\bx,t)f(\bx)$. The proof is concluded by rewriting\\
  $\esp{\bX(t) - \bphi(\bX(0),t)} = \bpsi(t) - \bpsi(0) = \int_0^t \frac{d}{d\tau}\bpsi(\tau) d\tau$.
\color{black}
\end{proof}

\subsection{Proof of Theorem~\ref{th:MF} (Mean Field Approximation)}
\label{sec:proof_mf}

By Lemma \ref{lemma:generator_comparison}, we have
\begin{align}
  \label{eq:MF:generators}
  \begin{split}
  \E[ \bX(t) - \bphi(\bX(0),t) ]
  = \int_0^t \E [ \sum_{\bx' \in \calX}&  K_{\bX(\tau),\bx'} \left(\bphi(\bx',t-\tau) - \bphi(\bX(\tau), t-\tau)\right) \\
  & - D_x\bphi(\bX(\tau), t-\tau)f(\bX(\tau)) ] d\tau.
  \end{split}
\end{align}
The above expression involves terms of the form $\bphi(\bx',\tau) - \bphi(\bx, \tau)$. By using a first order Taylor's expansion, we have:
\begin{align}
  \label{eq:expansion2}
  \bphi(\bx',\tau) - \bphi(\bx, \tau) = D_\bx \bphi(\bx, \tau)(\bx'-\bx) + R_1(\bx, \bx', \tau),
\end{align}
where $R_1(\bx,\bx',\tau)$ is a remainder term that can be expressed in integral form as 
\begin{align*}
  R_1(\bx,\bx',\tau) = \int_0^1 (1-\nu)\sum_{i,j\in\calI} \frac{\partial^2 \bphi}{\partial x_{i} \partial x_{j}}(\bx + \nu(\bx'-\bx) , \tau) (\bx'_i-\bx_i)(\bx'_{j}-\bx_{j}) d\nu.
\end{align*}
Moreover, by definition of the drift, one has $\sum_{\bx'\in\calX}K_{\bx,\bx'}(\bx'-\bx) = f(\bx)$. Combining this with \eqref{eq:expansion2} and plugging this into equation \eqref{eq:MF:generators} shows that
\begin{align}
  \label{eq:proof_R2}
& \E[\bX(t) - \bphi(\bX(0),t)] = \int_0^t \E[\sum_{\bx'\in\calX} K_{\bX(\tau), x'} R_1(\bX(\tau), \bx',t-\tau)]d\tau.
\end{align}
To conclude the proof, we show in Lemma~\ref{lemma:taylor_remainder} that $\sum_{\bx'\in\calX}K_{\bx,\bx'}R_1(\bx,\bx',\tau)$ is of order $O(1/n)$. Note that obtaining this bound is the most technical step of the proof as it requires to bound the second derivative of $\bphi$ as a function of the initial condition. This is where we use the assumptions on the rates $r$.




\subsection{Proof of Theorem~\ref{th:RMF} (Refined mean field approximation)}
\label{sec:proof_rmf}

The proof of Theorem~\ref{th:RMF} uses the same methodology as the proof of Theorem~\ref{th:MF} with two additional ideas: 
The first is to use a second-order Taylor expansion instead of the first order expansion used in \eqref{eq:expansion2}. The second is to express the refinement term $\bv$ as an integral of quantities that depend on the second derivative of $\bphi$.

By using a second order Taylor expansion of $\bphi$, it holds that
\begin{align}
  \label{eq:expansion3}
  \bphi(\bx',\tau) - \bphi(\bx, \tau) = D_\bx \bphi(\bx, \tau)(\bx'-\bx) + \sum_{i,j \in \calI} Q_{i,j}(\bx) \frac{\partial^2 \bphi}{\partial x_{i} \partial x_{j}} (\bx, \tau) +  R_2(\bx, \bx', \tau),
\end{align}
where the remainder term $R_2$ is equal to
\begin{align*}
  R_2(\bx,\bx',\tau)=\frac{1}{2} \int_0^1 (1{-}\nu)^2 \sum_{i,j,u\in\calI} \frac{\partial^3 \bphi}{\partial x_{i} \partial x_{j} \partial x_{u}} (\bx + \nu(\bx'{-}\bx) , \tau) (\bx'_i{-}\bx_i)(\bx'_j{-}\bx_j)(\bx'_u{-}\bx_u)d\nu
\end{align*} \color{myorange}
and $Q_{i,j}(\bx)$ is given by $\sum_{\bx' \in \calX} K_{\bx, \bx'}(x_i' - x_i)(x_j' - x_j) $ which we formally introduce in Appendix \ref{apx:rmf_def}. \color{black}

To simplify notations, let $\bg(\bx, \tau) := \sum_{i,j \in \calI} Q_{i,j}(\bx) \frac{\partial^2 \bphi}{\partial x_{i} \partial x_{j}} (\bx, \tau)$. Similarly to \eqref{eq:proof_R2}, we have
\begin{align}
  \label{eq:proof_R3}
& \E[\bX(t)] {-} \bphi(\bX(0),t)= \underbrace{\frac{1}{2}\int_0^t\!\! \E [ \bg(\bX(\tau),t{-}\tau)]d\tau}_{\text{$\approx \bv(\bx,t){+}O(\frac{1}{n^2})$ by \ref{lemma:connection_ode_integral_form} and \ref{lemma:refinement_bound}.}} + \int_0^t \!\!\E [ \underbrace{\sum_{\bx' \in \calX}K_{\bX(\tau),\bx'}R_2(\bX(\tau),\bx',t{-}\tau)}_{\text{$O(1/n^2)$ by Lemma~\ref{lemma:taylor_remainder}.}} ] d\tau.
\end{align}

By using the approach we used for $R_1$, we prove in Lemma~\ref{lemma:taylor_remainder} that the last term of the above equation (that involves a sum of $R_2$) is of order $O(1/n^2)$. This is quite technical and done by carefully bounding the first, second and third derivatives of $\bphi(\bx,t)$ with respect to its initial condition. We are then left with the first term of Equation~\eqref{eq:proof_R3}. By Lemma \ref{lemma:connection_ode_integral_form}, the refinement term $\bv$ can be expressed in integral form as $v_{(k,s)}(\bx, t) = \frac{1}{2}\int_0^t g_{(k,s)}(\bphi(\bx, \tau), \tau) d\tau$. This shows that
\begin{align*}
  \frac{1}{2} \int_0^t \E [g_{(k,s)}(\bX(\tau),\tau)]d\tau- v_{(k,s)}(\bX(0), t)&=\frac{1}{2} \int_0^t \E [g_{(k,s)}(\bX(\tau),\tau)- g_{(k,s)}(\bphi(\bX(0), \tau),\tau)]d\tau.
\end{align*}

We show in Lemma~\ref{lemma:refinement_bound} that the above term is of order $O(1/n^2)$. This requires to bound up to the fourth derivative of $\bphi$ with respect to its initial condition. Plugging everything into Equation~\eqref{eq:proof_R3} shows that $\E[\bX(t)] - \bphi(\bX(0),t) - \bv(\bX(0),t) = O(1/n^2)$ and concludes the proof.

\section{Conclusion}
\label{sec:conclusion}

In this paper, we show how to derive mean field and a refined mean field approximation for systems composed of $n$ heterogeneous objects. Most of the results which guarantee that mean field approximation is asymptotically correct assume that the system is composed of a population of $n$ homogeneous objects, or at least can be clustered into a finite number of classes of objects and let the number of objects in each class goes to infinity. A possible approach to derive a (refined) mean field approximation for a heterogeneous is to consider a scaled model with $C$ copies of each of the $n$ objects. Classical methods show that the (refined) mean field approximations are asymptotically exact as $C$ grows.

Our paper is the first to show that applying this method for the original system (with $C=1$ object of each of the $n$ class) is indeed valid. The main results of our paper, namely Theorem~\ref{th:MF} and \ref{th:RMF}, show that the accuracy of the mean field and refined mean field approximation is $O(1/n)$ and $O(1/n^2)$. We illustrate our results by considering two examples: a model of cache replacement policies, and a load balancing model. These examples show that the proposed approximations can be computed efficiently and are very accurate. They also show that taking heterogeneity into account is important to characterize precisely the quantitative behavior of such systems. 

When studying the performance of large computer systems, heterogeneity is often neglected because it increases the complexity of the model and because there are few tools to analyze such systems. We believe that our work has potential application in many models and will foster the development of the analysis of heterogeneous systems (such as load balancing or epidemic models).

\begin{acks}
  This work is supported by the French National Research Agency (ANR) through REFINO Project under Grant ANR-19-CE23-0015.  
\end{acks}

\bibliographystyle{ACM-Reference-Format}
\bibliography{references,biblio}


\begin{thebibliography}{37}


\ifx \showCODEN    \undefined \def \showCODEN     #1{\unskip}     \fi
\ifx \showDOI      \undefined \def \showDOI       #1{#1}\fi
\ifx \showISBNx    \undefined \def \showISBNx     #1{\unskip}     \fi
\ifx \showISBNxiii \undefined \def \showISBNxiii  #1{\unskip}     \fi
\ifx \showISSN     \undefined \def \showISSN      #1{\unskip}     \fi
\ifx \showLCCN     \undefined \def \showLCCN      #1{\unskip}     \fi
\ifx \shownote     \undefined \def \shownote      #1{#1}          \fi
\ifx \showarticletitle \undefined \def \showarticletitle #1{#1}   \fi
\ifx \showURL      \undefined \def \showURL       {\relax}        \fi
\providecommand\bibfield[2]{#2}
\providecommand\bibinfo[2]{#2}
\providecommand\natexlab[1]{#1}
\providecommand\showeprint[2][]{arXiv:#2}

\bibitem[\protect\citeauthoryear{Baccelli and Taillefumier}{Baccelli and
  Taillefumier}{2019}]%
        {baccelli2019replica}
\bibfield{author}{\bibinfo{person}{Fran{\c{c}}ois Baccelli} {and}
  \bibinfo{person}{Thibaud Taillefumier}.} \bibinfo{year}{2019}\natexlab{}.
\newblock \showarticletitle{Replica-mean-field limits for intensity-based
  neural networks}.
\newblock \bibinfo{journal}{\emph{SIAM Journal on Applied Dynamical Systems}}
  \bibinfo{volume}{18}, \bibinfo{number}{4} (\bibinfo{year}{2019}),
  \bibinfo{pages}{1756--1797}.
\newblock


\bibitem[\protect\citeauthoryear{Braverman}{Braverman}{2017}]%
        {bravermanSteinMethodSteadystate2017a}
\bibfield{author}{\bibinfo{person}{Anton Braverman}.}
  \bibinfo{year}{2017}\natexlab{}.
\newblock \showarticletitle{Stein's Method for Steady-State Diffusion
  Approximations}.
\newblock \bibinfo{journal}{\emph{arXiv:1704.08398}} (\bibinfo{date}{April}
  \bibinfo{year}{2017}).
\newblock
\showeprint[arxiv]{1704.08398}


\bibitem[\protect\citeauthoryear{Braverman and Dai}{Braverman and Dai}{2017}]%
        {braverman2017stein}
\bibfield{author}{\bibinfo{person}{Anton Braverman} {and} \bibinfo{person}{Jim
  Dai}.} \bibinfo{year}{2017}\natexlab{}.
\newblock \showarticletitle{Stein's method for steady-state diffusion
  approximations of $ M/Ph/n+ M $ systems}.
\newblock \bibinfo{journal}{\emph{The Annals of Applied Probability}}
  \bibinfo{volume}{27}, \bibinfo{number}{1} (\bibinfo{year}{2017}),
  \bibinfo{pages}{550--581}.
\newblock


\bibitem[\protect\citeauthoryear{Braverman, Dai, and Feng}{Braverman
  et~al\mbox{.}}{2017}]%
        {braverman2017stein2}
\bibfield{author}{\bibinfo{person}{Anton Braverman}, \bibinfo{person}{JG Dai},
  {and} \bibinfo{person}{Jiekun Feng}.} \bibinfo{year}{2017}\natexlab{}.
\newblock \showarticletitle{Stein's method for steady-state diffusion
  approximations: an introduction through the Erlang-A and Erlang-C models}.
\newblock \bibinfo{journal}{\emph{Stochastic Systems}} \bibinfo{volume}{6},
  \bibinfo{number}{2} (\bibinfo{year}{2017}), \bibinfo{pages}{301--366}.
\newblock


\bibitem[\protect\citeauthoryear{Braverman, Dai, and Fang}{Braverman
  et~al\mbox{.}}{2020}]%
        {bravermanHighOrderSteadystate2020}
\bibfield{author}{\bibinfo{person}{Anton Braverman}, \bibinfo{person}{J.~G.
  Dai}, {and} \bibinfo{person}{Xiao Fang}.} \bibinfo{year}{2020}\natexlab{}.
\newblock \showarticletitle{High Order Steady-State Diffusion Approximations}.
\newblock \bibinfo{journal}{\emph{arXiv:2012.02824}} (\bibinfo{date}{Dec.}
  \bibinfo{year}{2020}).
\newblock
\showeprint[arxiv]{2012.02824}


\bibitem[\protect\citeauthoryear{Casale and Gast}{Casale and Gast}{2020}]%
        {casale2020performance}
\bibfield{author}{\bibinfo{person}{Giuliano Casale} {and}
  \bibinfo{person}{Nicolas Gast}.} \bibinfo{year}{2020}\natexlab{}.
\newblock \showarticletitle{Performance analysis methods for list-based caches
  with non-uniform access}.
\newblock \bibinfo{journal}{\emph{IEEE/ACM Transactions on Networking}}
  (\bibinfo{year}{2020}).
\newblock


\bibitem[\protect\citeauthoryear{Che, Wang, and Tung}{Che
  et~al\mbox{.}}{2001}]%
        {che2001analysis}
\bibfield{author}{\bibinfo{person}{Hao Che}, \bibinfo{person}{Zhijun Wang},
  {and} \bibinfo{person}{Ye Tung}.} \bibinfo{year}{2001}\natexlab{}.
\newblock \showarticletitle{Analysis and design of hierarchical web caching
  systems}. In \bibinfo{booktitle}{\emph{INFOCOM 2001. Twentieth Annual Joint
  Conference of the IEEE Computer and Communications Societies. Proceedings.
  IEEE}}, Vol.~\bibinfo{volume}{3}. IEEE, \bibinfo{pages}{1416--1424}.
\newblock


\bibitem[\protect\citeauthoryear{Dan and Towsley}{Dan and Towsley}{1990}]%
        {dan1990approximate}
\bibfield{author}{\bibinfo{person}{Asit Dan} {and} \bibinfo{person}{Don
  Towsley}.} \bibinfo{year}{1990}\natexlab{}.
\newblock \showarticletitle{An approximate analysis of the LRU and FIFO buffer
  replacement schemes}. In \bibinfo{booktitle}{\emph{Proceedings of the 1990
  ACM SIGMETRICS conference on Measurement and modeling of computer systems}}.
  \bibinfo{pages}{143--152}.
\newblock


\bibitem[\protect\citeauthoryear{{de Arruda}, Rodrigues, and Moreno}{{de
  Arruda} et~al\mbox{.}}{2018}]%
        {dearrudaFundamentalsSpreadingProcesses2018}
\bibfield{author}{\bibinfo{person}{Guilherme~Ferraz {de Arruda}},
  \bibinfo{person}{Francisco~A. Rodrigues}, {and} \bibinfo{person}{Yamir
  Moreno}.} \bibinfo{year}{2018}\natexlab{}.
\newblock \showarticletitle{Fundamentals of Spreading Processes in Single and
  Multilayer Complex Networks}.
\newblock \bibinfo{journal}{\emph{Physics Reports}}  \bibinfo{volume}{756}
  (\bibinfo{date}{Oct.} \bibinfo{year}{2018}), \bibinfo{pages}{1--59}.
\newblock
\showISSN{03701573}
\urldef\tempurl%
\url{https://doi.org/10.1016/j.physrep.2018.06.007}
\showDOI{\tempurl}
\showeprint[arxiv]{1804.08777}


\bibitem[\protect\citeauthoryear{Fagin}{Fagin}{1977}]%
        {fagin1977asymptotic}
\bibfield{author}{\bibinfo{person}{Ronald Fagin}.}
  \bibinfo{year}{1977}\natexlab{}.
\newblock \showarticletitle{Asymptotic miss ratios over independent
  references}.
\newblock \bibinfo{journal}{\emph{J. Comput. System Sci.}}
  \bibinfo{volume}{14}, \bibinfo{number}{2} (\bibinfo{year}{1977}),
  \bibinfo{pages}{222--250}.
\newblock


\bibitem[\protect\citeauthoryear{Fricker, Robert, and Roberts}{Fricker
  et~al\mbox{.}}{2012}]%
        {fricker2012versatile}
\bibfield{author}{\bibinfo{person}{Christine Fricker},
  \bibinfo{person}{Philippe Robert}, {and} \bibinfo{person}{James Roberts}.}
  \bibinfo{year}{2012}\natexlab{}.
\newblock \showarticletitle{A versatile and accurate approximation for {LRU}
  cache performance}. In \bibinfo{booktitle}{\emph{2012 24th international
  teletraffic congress (ITC 24)}}. IEEE, \bibinfo{pages}{1--8}.
\newblock


\bibitem[\protect\citeauthoryear{Gast}{Gast}{2017}]%
        {gast2017expected}
\bibfield{author}{\bibinfo{person}{Nicolas Gast}.}
  \bibinfo{year}{2017}\natexlab{}.
\newblock \showarticletitle{Expected values estimated via mean-field
  approximation are 1/N-accurate}.
\newblock \bibinfo{journal}{\emph{Proceedings of the ACM on Measurement and
  Analysis of Computing Systems}} \bibinfo{volume}{1}, \bibinfo{number}{1}
  (\bibinfo{year}{2017}), \bibinfo{pages}{17}.
\newblock


\bibitem[\protect\citeauthoryear{Gast}{Gast}{2018}]%
        {gastRefinedMeanField2018}
\bibfield{author}{\bibinfo{person}{Nicolas Gast}.}
  \bibinfo{year}{2018}\natexlab{}.
\newblock \showarticletitle{Refined {{Mean Field Tool}}}.
\newblock \bibinfo{howpublished}{\url{https://github.com/ngast/rmf\_tool}}.
\newblock  (\bibinfo{year}{2018}).
\newblock


\bibitem[\protect\citeauthoryear{Gast, Bortolussi, and Tribastone}{Gast
  et~al\mbox{.}}{2019}]%
        {gastSizeExpansionsMean2019}
\bibfield{author}{\bibinfo{person}{Nicolas Gast}, \bibinfo{person}{Luca
  Bortolussi}, {and} \bibinfo{person}{Mirco Tribastone}.}
  \bibinfo{year}{2019}\natexlab{}.
\newblock \showarticletitle{Size Expansions of Mean Field Approximation:
  {{Transient}} and Steady-State Analysis}.
\newblock \bibinfo{journal}{\emph{Performance Evaluation}}
  \bibinfo{volume}{129} (\bibinfo{date}{Feb.} \bibinfo{year}{2019}),
  \bibinfo{pages}{60--80}.
\newblock
\showISSN{01665316}
\urldef\tempurl%
\url{https://doi.org/10.1016/j.peva.2018.09.005}
\showDOI{\tempurl}


\bibitem[\protect\citeauthoryear{Gast and Van~Houdt}{Gast and
  Van~Houdt}{2015}]%
        {gastTransientSteadystateRegime2015}
\bibfield{author}{\bibinfo{person}{Nicolas Gast} {and} \bibinfo{person}{Benny
  Van~Houdt}.} \bibinfo{year}{2015}\natexlab{}.
\newblock \showarticletitle{Transient and {{Steady}}-State {{Regime}} of a
  {{Family}} of {{List}}-Based {{Cache Replacement Algorithms}}}. In
  \bibinfo{booktitle}{\emph{Proceedings of the 2015 {{ACM SIGMETRICS
  International Conference}} on {{Measurement}} and {{Modeling}} of {{Computer
  Systems}} - {{SIGMETRICS}} '15}}. \bibinfo{publisher}{{ACM Press}},
  \bibinfo{address}{{Portland, Oregon, USA}}, \bibinfo{pages}{123--136}.
\newblock
\showISBNx{978-1-4503-3486-0}
\urldef\tempurl%
\url{https://doi.org/10.1145/2745844.2745850}
\showDOI{\tempurl}


\bibitem[\protect\citeauthoryear{Gast and Van~Houdt}{Gast and
  Van~Houdt}{2017a}]%
        {gastRefinedMeanField2017}
\bibfield{author}{\bibinfo{person}{Nicolas Gast} {and} \bibinfo{person}{Benny
  Van~Houdt}.} \bibinfo{year}{2017}\natexlab{a}.
\newblock \showarticletitle{A {{Refined Mean Field Approximation}}}.
\newblock \bibinfo{journal}{\emph{Proceedings of the ACM on Measurement and
  Analysis of Computing Systems}} \bibinfo{volume}{1}, \bibinfo{number}{2}
  (\bibinfo{date}{Dec.} \bibinfo{year}{2017}), \bibinfo{pages}{33:1--33:28}.
\newblock
\urldef\tempurl%
\url{https://doi.org/10.1145/3154491}
\showDOI{\tempurl}


\bibitem[\protect\citeauthoryear{Gast and Van~Houdt}{Gast and
  Van~Houdt}{2017b}]%
        {gast2017refined}
\bibfield{author}{\bibinfo{person}{Nicolas Gast} {and} \bibinfo{person}{Benny
  Van~Houdt}.} \bibinfo{year}{2017}\natexlab{b}.
\newblock \showarticletitle{A Refined Mean Field Approximation}.
\newblock \bibinfo{journal}{\emph{Proc. ACM Meas. Anal. Comput. Syst}}
  \bibinfo{volume}{1} (\bibinfo{year}{2017}).
\newblock


\bibitem[\protect\citeauthoryear{Gomes, Aguas, Corder, King, Langwig,
  {Souto-Maior}, Carneiro, Ferreira, and {Penha-Gon{\c c}alves}}{Gomes
  et~al\mbox{.}}{2020}]%
        {gomesIndividualVariationSusceptibility2020}
\bibfield{author}{\bibinfo{person}{M.~Gabriela~M. Gomes},
  \bibinfo{person}{Ricardo Aguas}, \bibinfo{person}{Rodrigo~M. Corder},
  \bibinfo{person}{Jessica~G. King}, \bibinfo{person}{Kate~E. Langwig},
  \bibinfo{person}{Caetano {Souto-Maior}}, \bibinfo{person}{Jorge Carneiro},
  \bibinfo{person}{Marcelo~U. Ferreira}, {and} \bibinfo{person}{Carlos
  {Penha-Gon{\c c}alves}}.} \bibinfo{year}{2020}\natexlab{}.
\newblock \bibinfo{booktitle}{\emph{Individual Variation in Susceptibility or
  Exposure to {{SARS}}-{{CoV}}-2 Lowers the Herd Immunity Threshold}}.
\newblock \bibinfo{type}{Preprint}. \bibinfo{institution}{{Epidemiology}}.
\newblock
\urldef\tempurl%
\url{https://doi.org/10.1101/2020.04.27.20081893}
\showDOI{\tempurl}


\bibitem[\protect\citeauthoryear{Grima}{Grima}{2010}]%
        {grima2010effective}
\bibfield{author}{\bibinfo{person}{Ramon Grima}.}
  \bibinfo{year}{2010}\natexlab{}.
\newblock \showarticletitle{An effective rate equation approach to reaction
  kinetics in small volumes: Theory and application to biochemical reactions in
  nonequilibrium steady-state conditions}.
\newblock \bibinfo{journal}{\emph{The Journal of chemical physics}}
  \bibinfo{volume}{133}, \bibinfo{number}{3} (\bibinfo{year}{2010}),
  \bibinfo{pages}{07B604}.
\newblock


\bibitem[\protect\citeauthoryear{Grima, Thomas, and Straube}{Grima
  et~al\mbox{.}}{2011}]%
        {grima2011}
\bibfield{author}{\bibinfo{person}{Ramon Grima}, \bibinfo{person}{Philipp
  Thomas}, {and} \bibinfo{person}{Arthur~V. Straube}.}
  \bibinfo{year}{2011}\natexlab{}.
\newblock \showarticletitle{How accurate are the nonlinear chemical
  {Fokker}-{Planck} and chemical {Langevin} equations?}
\newblock \bibinfo{journal}{\emph{The Journal of Chemical Physics}}
  \bibinfo{volume}{135}, \bibinfo{number}{8} (\bibinfo{year}{2011}).
\newblock


\bibitem[\protect\citeauthoryear{Hazewinkel, Calogero, Manin, Rinnooy~Kan, and
  Rota}{Hazewinkel et~al\mbox{.}}{1987}]%
        {hazewinkelStochasticAnalysisComputer1987}
\bibfield{author}{\bibinfo{person}{M Hazewinkel}, \bibinfo{person}{F Calogero},
  \bibinfo{person}{Yu.~I Manin}, \bibinfo{person}{A.~H.~G Rinnooy~Kan}, {and}
  \bibinfo{person}{G.-C Rota}.} \bibinfo{year}{1987}\natexlab{}.
\newblock \bibinfo{booktitle}{\emph{Stochastic {{Analysis}} of {{Computer
  Storage}}}}.
\newblock \bibinfo{publisher}{{Springer Netherlands}},
  \bibinfo{address}{{Dordrecht}}.
\newblock
\showISBNx{978-90-277-2515-8 978-0-585-27373-0}


\bibitem[\protect\citeauthoryear{Hirade and Osogami}{Hirade and
  Osogami}{2010}]%
        {hirade1}
\bibfield{author}{\bibinfo{person}{R. Hirade} {and} \bibinfo{person}{T.
  Osogami}.} \bibinfo{year}{2010}\natexlab{}.
\newblock \showarticletitle{Analysis of Page Replacement Policies in the Fluid
  Limit}.
\newblock \bibinfo{journal}{\emph{Oper. Res.}}  \bibinfo{volume}{58}
  (\bibinfo{date}{July} \bibinfo{year}{2010}), \bibinfo{pages}{971--984}.
\newblock
\showISSN{0030-364X}
\urldef\tempurl%
\url{https://doi.org/10.1287/opre.1090.0761}
\showDOI{\tempurl}


\bibitem[\protect\citeauthoryear{Hodgkinson, McVinish, and Pollett}{Hodgkinson
  et~al\mbox{.}}{2018}]%
        {hodgkinsonNormalApproximationsDiscretetime2018}
\bibfield{author}{\bibinfo{person}{Liam Hodgkinson}, \bibinfo{person}{Ross
  McVinish}, {and} \bibinfo{person}{Philip~K. Pollett}.}
  \bibinfo{year}{2018}\natexlab{}.
\newblock \showarticletitle{Normal Approximations for Discrete-Time Occupancy
  Processes}.
\newblock \bibinfo{journal}{\emph{arXiv:1801.00542}} (\bibinfo{date}{Nov.}
  \bibinfo{year}{2018}).
\newblock
\showeprint[arxiv]{1801.00542}


\bibitem[\protect\citeauthoryear{Jiang, Nain, and Towsley}{Jiang
  et~al\mbox{.}}{2018}]%
        {jiang2018convergence}
\bibfield{author}{\bibinfo{person}{Bo Jiang}, \bibinfo{person}{Philippe Nain},
  {and} \bibinfo{person}{Don Towsley}.} \bibinfo{year}{2018}\natexlab{}.
\newblock \showarticletitle{On the convergence of the {TTL} approximation for
  an {LRU} cache under independent stationary request processes}.
\newblock \bibinfo{journal}{\emph{ACM Transactions on Modeling and Performance
  Evaluation of Computing Systems (TOMPECS)}} \bibinfo{volume}{3},
  \bibinfo{number}{4} (\bibinfo{year}{2018}), \bibinfo{pages}{1--31}.
\newblock


\bibitem[\protect\citeauthoryear{Kolokoltsov, Li, and Yang}{Kolokoltsov
  et~al\mbox{.}}{2012}]%
        {kolokoltsovMeanFieldGames2012}
\bibfield{author}{\bibinfo{person}{Vassili~N. Kolokoltsov},
  \bibinfo{person}{Jiajie Li}, {and} \bibinfo{person}{Wei Yang}.}
  \bibinfo{year}{2012}\natexlab{}.
\newblock \showarticletitle{Mean {{Field Games}} and {{Nonlinear Markov
  Processes}}}.
\newblock \bibinfo{journal}{\emph{arXiv:1112.3744}} (\bibinfo{date}{April}
  \bibinfo{year}{2012}).
\newblock
\showeprint[arxiv]{1112.3744}


\bibitem[\protect\citeauthoryear{Kurtz}{Kurtz}{1978}]%
        {kurtzStrongApproximationTheorems1978}
\bibfield{author}{\bibinfo{person}{Thomas~G. Kurtz}.}
  \bibinfo{year}{1978}\natexlab{}.
\newblock \showarticletitle{Strong Approximation Theorems for Density Dependent
  {{Markov}} Chains}.
\newblock \bibinfo{journal}{\emph{Stochastic Processes and their Applications}}
  \bibinfo{volume}{6}, \bibinfo{number}{3} (\bibinfo{date}{Feb.}
  \bibinfo{year}{1978}), \bibinfo{pages}{223--240}.
\newblock
\showISSN{03044149}
\urldef\tempurl%
\url{https://doi.org/10.1016/0304-4149(78)90020-0}
\showDOI{\tempurl}


\bibitem[\protect\citeauthoryear{M{\'e}zard, Parisi, and Virasoro}{M{\'e}zard
  et~al\mbox{.}}{1987}]%
        {mezard1987sk}
\bibfield{author}{\bibinfo{person}{Marc M{\'e}zard}, \bibinfo{person}{Giorgio
  Parisi}, {and} \bibinfo{person}{Miguel~Angel Virasoro}.}
  \bibinfo{year}{1987}\natexlab{}.
\newblock \showarticletitle{Spin glass theory and beyond: An Introduction to
  the Replica Method and Its Applications}.
\newblock \bibinfo{journal}{\emph{World Scientific Publishing Company}}
  (\bibinfo{year}{1987}), \bibinfo{pages}{232--237}.
\newblock


\bibitem[\protect\citeauthoryear{Mitzenmacher}{Mitzenmacher}{2001}]%
        {mitzenmacherPowerTwoChoices2001}
\bibfield{author}{\bibinfo{person}{M. Mitzenmacher}.}
  \bibinfo{year}{Oct./2001}\natexlab{}.
\newblock \showarticletitle{The Power of Two Choices in Randomized Load
  Balancing}.
\newblock \bibinfo{journal}{\emph{IEEE Transactions on Parallel and Distributed
  Systems}} \bibinfo{volume}{12}, \bibinfo{number}{10}
  (\bibinfo{year}{Oct./2001}), \bibinfo{pages}{1094--1104}.
\newblock
\showISSN{10459219}
\urldef\tempurl%
\url{https://doi.org/10.1109/71.963420}
\showDOI{\tempurl}


\bibitem[\protect\citeauthoryear{Montalb{\'a}n, Corder, and
  Gomes}{Montalb{\'a}n et~al\mbox{.}}{2020}]%
        {montalbanHerdImmunityIndividual2020}
\bibfield{author}{\bibinfo{person}{Antonio Montalb{\'a}n},
  \bibinfo{person}{Rodrigo~M. Corder}, {and} \bibinfo{person}{M.~Gabriela~M.
  Gomes}.} \bibinfo{year}{2020}\natexlab{}.
\newblock \showarticletitle{Herd Immunity under Individual Variation and
  Reinfection}.
\newblock \bibinfo{journal}{\emph{arXiv:2008.00098}} (\bibinfo{date}{Nov.}
  \bibinfo{year}{2020}).
\newblock
\showeprint[arxiv]{2008.00098}


\bibitem[\protect\citeauthoryear{Mukhopadhyay and Mazumdar}{Mukhopadhyay and
  Mazumdar}{2015}]%
        {mukhopadhyayAnalysisLoadBalancing2015}
\bibfield{author}{\bibinfo{person}{Arpan Mukhopadhyay} {and}
  \bibinfo{person}{Ravi~R. Mazumdar}.} \bibinfo{year}{2015}\natexlab{}.
\newblock \showarticletitle{Analysis of {{Load Balancing}} in {{Large
  Heterogeneous Processor Sharing Systems}}}.
\newblock \bibinfo{journal}{\emph{arXiv:1311.5806}} (\bibinfo{date}{Feb.}
  \bibinfo{year}{2015}).
\newblock
\showeprint[arxiv]{1311.5806}


\bibitem[\protect\citeauthoryear{Stein}{Stein}{1986}]%
        {stein1986approximate}
\bibfield{author}{\bibinfo{person}{Charles Stein}.}
  \bibinfo{year}{1986}\natexlab{}.
\newblock \showarticletitle{Approximate computation of expectations}.
\newblock \bibinfo{journal}{\emph{Lecture Notes-Monograph Series}}
  \bibinfo{volume}{7} (\bibinfo{year}{1986}), \bibinfo{pages}{i--164}.
\newblock


\bibitem[\protect\citeauthoryear{Tsukada, Hirade, and Miyoshi}{Tsukada
  et~al\mbox{.}}{2012}]%
        {tsukada1}
\bibfield{author}{\bibinfo{person}{N. Tsukada}, \bibinfo{person}{R. Hirade},
  {and} \bibinfo{person}{N. Miyoshi}.} \bibinfo{year}{2012}\natexlab{}.
\newblock \showarticletitle{Fluid Limit Analysis of {FIFO} and {RR} Caching for
  Independent Reference Models}.
\newblock \bibinfo{journal}{\emph{Perform. Eval.}} \bibinfo{volume}{69},
  \bibinfo{number}{9} (\bibinfo{date}{Sept.} \bibinfo{year}{2012}),
  \bibinfo{pages}{403--412}.
\newblock
\showISSN{0166-5316}
\urldef\tempurl%
\url{https://doi.org/10.1016/j.peva.2012.05.008}
\showDOI{\tempurl}


\bibitem[\protect\citeauthoryear{Van~Houdt}{Van~Houdt}{2013}]%
        {vanhoudtMeanFieldModel2013}
\bibfield{author}{\bibinfo{person}{Benny Van~Houdt}.}
  \bibinfo{year}{2013}\natexlab{}.
\newblock \showarticletitle{A Mean Field Model for a Class of Garbage
  Collection Algorithms in Flash-Based Solid State Drives}.
\newblock \bibinfo{journal}{\emph{ACM SIGMETRICS Performance Evaluation
  Review}} \bibinfo{volume}{41}, \bibinfo{number}{1} (\bibinfo{date}{June}
  \bibinfo{year}{2013}), \bibinfo{pages}{191--202}.
\newblock
\showISSN{0163-5999}
\urldef\tempurl%
\url{https://doi.org/10.1145/2494232.2465543}
\showDOI{\tempurl}


\bibitem[\protect\citeauthoryear{van Kampen}{van Kampen}{2007}]%
        {vankampen2007}
\bibfield{author}{\bibinfo{person}{N.~G. van Kampen}.}
  \bibinfo{year}{2007}\natexlab{}.
\newblock \bibinfo{booktitle}{\emph{Stochastic processes in physics and
  chemistry}}.
\newblock \bibinfo{publisher}{Elsevier}, \bibinfo{address}{Amsterdam; Boston;
  London}.
\newblock
\showISBNx{978-0-444-52965-7 0-444-52965-9 978-0-08-047536-3 0-08-047536-1}


\bibitem[\protect\citeauthoryear{Ying}{Ying}{2015}]%
        {ying2015rate}
\bibfield{author}{\bibinfo{person}{Lei Ying}.} \bibinfo{year}{2015}\natexlab{}.
\newblock \showarticletitle{On the Rate of Convergence of Mean-Field Models:
  Stein's Method Meets the Perturbation Theory}.
\newblock \bibinfo{journal}{\emph{arXiv preprint arXiv:1510.00761}}
  (\bibinfo{year}{2015}).
\newblock


\bibitem[\protect\citeauthoryear{Ying}{Ying}{2016}]%
        {ying2016rate}
\bibfield{author}{\bibinfo{person}{Lei Ying}.} \bibinfo{year}{2016}\natexlab{}.
\newblock \showarticletitle{On the Approximation Error of Mean-Field Models}.
  In \bibinfo{booktitle}{\emph{Proceedings of the 2016 ACM SIGMETRICS
  International Conference on Measurement and Modeling of Computer Science}}.
  ACM, \bibinfo{pages}{285--297}.
\newblock


\bibitem[\protect\citeauthoryear{Ying}{Ying}{2017}]%
        {ying2017stein}
\bibfield{author}{\bibinfo{person}{Lei Ying}.} \bibinfo{year}{2017}\natexlab{}.
\newblock \showarticletitle{Stein's Method for Mean Field Approximations in
  Light and Heavy Traffic Regimes}.
\newblock \bibinfo{journal}{\emph{Proceedings of the ACM on Measurement and
  Analysis of Computing Systems}} \bibinfo{volume}{1}, \bibinfo{number}{1}
  (\bibinfo{year}{2017}), \bibinfo{pages}{12}.
\newblock


\end{thebibliography}

\appendix 

\section{Notation List}
\label{apx:notations}

\begin{center}
\begin{tabular}{l l}
$n$ & system size\\
$a_{\{\ldots\}}, b_{\{\ldots\}}$ & unilateral  and pairwise transition parameters; Equations \eqref{eq:unilateral}, \eqref{eq:interactions} \\
$\bS(t)$ & Markov chain describing the heterogeneous population model; Section \ref{ssec:interaction_model_and_definition} \\
$\bX(t)$ & binary based representation of the model described by $\bS(t)$; Section \ref{ssec:state_representation} \\
$\calS $ & finite state space of objects in the population model described by $\bS$ and $\bX$ \\
$s, \hat{s}, s', s_1$ & notation for states \\
$k, \hat{k}, k', k_1$ & notation for objects \\
$X_{(k,s)}(t)$ & entry of the Markov chain $\bX$ indicating if object~$k$ is in state~$s$ at time $t$ \\
$\bdrift, \bQ, \bR $ &  drift of the stochastic system and related tensors; Section \ref{ssec:drift_mean_field}, Appendix \ref{apx:rmf_def}, Lemma \ref{lemma:taylor_remainder} \\
$\bphi(\bx, t)$ & solution to the ODE given by the drift $\bdrift$ of the system; Section \ref{ssec:drift_mean_field} \\
$\bv(\bx, t)$ & refinement term; Section \ref{ssec:rmf:definition}, Appendix \ref{apx:rmf_def}\\
$\bw(\bx, t)$ & solution to the second set of differential equations of the refinement; Appendix \ref{apx:rmf_def} \\
$t, \tau, \nu$ & time and integration variables \\ 
$\calI, \calI_k$ & sets of object-state pairs defined by $\{1,\ldots,n\}\times \calS$ and $\{k\}\times \calS$ respectively\\
$ i,j,w,u,l $ & indices used for elements of the sets $\calI$ and $\calI_k$ \\
$\bx, \by, \bz$ & initial conditions for Markov chain $\bX$ and corresponding mean field approximation \\
$K_{\bx,\bx'}$ & transition rate for Markov chain $\bX$ from state $\bx$ to $\bx'$, Lemma \ref{lemma:taylor_remainder} \\
$L_{1, 2}$ & bounds for first or second partial derivatives of the drift $\bdrift$ \\
$ \bx \otimes\by $ & Kronecker product of $\bx$ and $\by$ \\
$\bx^{\otimes2}, \bx^{\otimes3}$ & Kronecker product of $\bx$ with itself ($\bx^{\otimes2} = \bx \otimes\bx; \  \bx^{\otimes3} = \bx \otimes\bx\otimes\bx$); Lemma \ref{lemma:taylor_remainder} \\
$D, D^2, D_{\bx}$ & 1st and 2nd order derivative, derivative with respect to $\bx$ \\
$ \frac{\partial }{\partial x_i} $ & partial derivative with respect to $x_i$ \\
$ \frac{d }{dt}, \frac{\partial }{\partial t}, \dot{x}(t)$ & derivative with respect to time $t$ \\
$ \mathbf{1}_{\{ a>b \}} $ & indicator function, i.e. indicating if $a>b$ \\
\end{tabular}
\end{center}

\section{Equation for the Mean Field and Refined Mean Field Approximations} 
\color{myorange}

\subsection{General drift definition}
\label{apx:drift_def}

For completeness, we give the general form of the drift $f\toN$ for a heterogeneous interaction model having up to $d_{\max}$ interacting objects. The drift in $(k,s)$ is derived from interactions that imply either object~$k$ transitions into state~$s$ ($s' \rightarrow s$) or leaves state~$s$ ($s \rightarrow s'$). By considering all these interactions, the $(k,s)$ component of the drift $f\toN(\bx)$ is
\begin{align}
  & \sum_{s' \ne s} (r\toN_{k,(s') \rightarrow (s)} x_{(k,s')} - r\toN_{k,(s) \rightarrow (s')} x_{(k,s)})  \label{eq:drift_apx}\\
  & \phantom{+} + \frac1n \sum_{s',k_1, s_1,s_1'} (r\toN_{k,k_1,(s',s_1')\rightarrow (s,s_1)} x_{(k,s')}x_{(k_1,s_1)} - r\toN_{k,k_1,(s,s_1) \rightarrow (s',s'_1)} x_{(k,s)}x_{(k_1,s_1')}) \nonumber \\
  & \phantom{++} + \sum_{d=3,\hdots,d_{\max}} \frac{1}{n^{d-1}}\sum_{{k_1,...,k_{d-1} \atop s_1,...,s_{d-1}} \atop s', s_1', \hdots, s_{d-1}'} r_{k,\hdots,k_{d-1},(s',\hdots,s_{d-1}') \rightarrow (s,\hdots,s_{d-1})}\toN x_{(k,s')}^{(n)}\hdots x_{(k_{d-1},s_{d-1}')}^{(n)} \nonumber \\
  & \phantom{++ + \sum_{d=3,\hdots,d_{\max}}^{n} \frac{1}{n^{d-1}}\sum_{{k_1,...,k_{d-1} \atop s_1,...,s_{d-1}} \atop s', s_1', \hdots, s_{d-1}'}++++++} - r\toN_{k,\hdots,k_{d-1},(s,\hdots,s_{d-1}) \rightarrow (s',\hdots,s_{d-1}')} x_{(k,s)}^{(n)}\hdots x_{(k_{d-1},s_{d-1})}^{(n)} \nonumber
\end{align} 
In the above equation, we sum over all permutations such that the first objects is fixed to $k$. This counters the factor $1/d$ of the definition of the rates. Therefore the sum can be written as $\frac{1}{n^{d-1}}\sum_{{k_1,...,k_{d-1} \atop s_1,...,s_{d-1}} \atop s',s_1, \hdots, s_{d-1}'} r_{k,\hdots,k_{d-1},(s',\hdots,s_{d-1}') \rightarrow (s,\hdots,s_{d-1})}\toN x_{(k,s')}^{(n)}\hdots x_{(k_{d-1},s_{d-1}')}^{(n)}$. Without loss of generality, we fix the index order $k,k_1,\hdots$ to simplify the mathematical notations. This simplification comes from the fact that we assumed for any permutation $\sigma$ of the set $\{1\dots d\}$, the rates satisfy $r_{k_1, \dots, k_{d}, (s_1, \dots, s_{d}) \rightarrow (s_1', \dots, s_{d}')}\toN = r_{k_{\sigma(1)}, \dots, k_{\sigma(d)}, (s_{\sigma(1)}, \dots, s_{\sigma(d)}) \rightarrow (s_{\sigma(1)}', \dots, s_{\sigma(d)}')}\toN$.

\subsection{Definition of the Refined Mean Field Approximation}
\label{apx:rmf_def}



In this section we show how the definition of the refinement term $\bv$ from \cite{gastSizeExpansionsMean2019} can be adapted and how it can be computed using the rates of the model introduced in section \ref{sec:model}. In \cite{gastSizeExpansionsMean2019}, the refinement term is based on a density representation of the stochastic system and therefore independent of the state of individual objects. Since our model representation takes the state of each object into account, we extend the definition of their refinement term $\bv$ to object-state pairs with the following set of ODEs (for better readability we suppress the dependence on~$n$ in the definitions)
\begin{align*}
\dot{v}_{(k_1,s_1)}(\bx,t) &= \sum_{u \in \calI} \frac{\partial \bdrift_{(k_1,s_1)}}{\partial x_u}(\bphi(\bx,t)) v_{u}(\bx,t) + \frac{1}{2}\sum_{u,l \in \calI} \frac{\partial^2\bdrift_{(k_1,s_1)}}{\partial x_l \partial x_u}(\bphi(\bx,t))w_{u,l}(\bx,t), \\
\dot{w}_{(k_1,s_1),(k_2,s_2)}(\bx,t) &= \sum_{u \in \calI}w_{u,(k_2,s_2)}(\bx,t) \frac{\partial \bdrift_{(k_1,s_1)}}{\partial x_{u}}(\bphi(\bx,t))  \\
& \qquad + \sum_{u \in \calI}w_{u,(k_1,s_1)}(\bx,t) \frac{\partial \bdrift_{(k_2,s_2)}}{\partial x_u}(\bphi(\bx,t)) + Q_{(k_1,s_1),(k_2,s_2)}(\bphi(\bx,t)),
\end{align*}
with initial conditions $\bv(\bx, 0) = 0, \bw(\bx, 0) = 0$. The values of $\bv$ and $\bw$ should be interpreted as the leading correction terms for the first moment and covariance of $\bX(t) - \bphi(\bx,t)$. The value of $\bQ$ is given by the expected change of the covariance of the stochastic system which, for a given state $\bx$, is
\begin{align*}
  \bQ_{(k,s),(k',s')}(\bx) = \sum_{\bx'\in\calX} K_{\bx,\bx'}(\bx'_{(k,s)}-\bx_{(k,s)})(\bx'_{(k',s')}-\bx_{(k',s')}).
\end{align*}
To make the definitions less abstract, we give explicit formulas of $\bQ$ and the derivative of the drift $f$ when considering a heterogeneous model having at most pairwise interactions. We start by characterizing the elements of $\bQ$ evaluated at $\bx$ (here $s\ne s'$ and $k\ne k'$) 
\begin{align*}
Q_{(k,s),(k,s)}(\bx) &{=} \sum_{s'}r_{k,(s)\rightarrow (s')}x_{(k,s)} {+} r_{k,(s')\rightarrow (s)}x_{(k,s')} \\ 
& \phantom{+++} {+} \frac1n \!\!\sum_{k_1, s_1,s_1',s'}\!\!\!\! r_{k,k_1,(s,s_1)\rightarrow (s',s_1')}x_{(k,s)}x_{(k_1,s_1)} {+} r_{k,k_1,(s',s_1')\rightarrow (s,s_1)}x_{(k,s')}x_{(k_1,s_1')}, \\
Q_{(k,s),(k,s')}(\bx) &{=} {-} r_{k,(s)\rightarrow (s')}x_{(k,s)} {-} r_{k,(s') \rightarrow (s)}x_{(k,s')} \\
& \phantom{+++} {+} \frac1n \!\sum_{k_1,s_1,s_1'}\!\!\!\!- r_{k,k_1,(s,s_1)\rightarrow (s',s_1')}x_{(k,s)}x_{(k_1,s_1)} {-} r_{k,k_1,(s',s_1) \rightarrow (s,s_1')}x_{(k,s')}x_{(k_1,s_1)}, \\
Q_{(k,s),(k',s_1)}(\bx) &{=} \frac1n \sum_{s',s_1'} r_{k,k',(s,s_1) \rightarrow (s',s_1')}x_{(k,s)}x_{(k',s_1)} - r_{k,k',(s',s_1') \rightarrow (s,s_1)}x_{(k,s')}x_{(k',s_1')}, \ k \ne k'.
\end{align*} 
The first and second partial derivatives of the drift $f$ are given by

\begin{align*}
& \frac{\partial f_{(k,s)}}{\partial x_{(k,s)}}(\bx) = {-} \sum_{s'} r_{k,(s)\rightarrow(s')} {-} \frac1n \sum_{k_1\ne k, s_1, s_1', s'\neq s} r_{k,k_1, (s, s_1)\rightarrow (s', s_1')}x_{(k_1,s_1)} \ , \\
& \frac{\partial f_{(k,s)}}{\partial x_{(k,\tilde{s})}}(\bx) = r_{k,(\tilde{s})\rightarrow (s)} +\frac1n \sum_{k_1 \ne k,s_1,s_1'} r_{k,k_1,(\tilde{s},s_1)\rightarrow (s ,s_1')}x_{(k_1,s_1)} && \tilde{s} \ne s, \\
& \frac{\partial f_{(k,s)}}{\partial x_{(\tilde{k},\tilde{s})}}(\bx) = \frac1n\sum_{s',s_1'} r_{k,\tilde{k},(s',\tilde{s})\rightarrow(s,s_1')}x_{(k,s')} {-} r_{k,\tilde{k},(s,\tilde{s})\rightarrow (s',s_1')}x_{(k,s)} && \tilde{k} \ne k,
\end{align*}

\begin{align*}
& \frac{\partial^2 f_{(k,s)}}{\partial x_{(\hat{k},\hat{s})} \partial x_{(k,s)}}(\bx) = {-} \frac1n \sum_{s_1', s'} r_{k,\hat{k}, (s, \hat{s}) \rightarrow (s', s_1')} && \hat{k} \ne k,\\
& \frac{\partial^2 f_{(k,s)}}{\partial x_{(\hat{k},\hat{s})} \partial x_{(k,\tilde{s})}}(\bx) = \phantom{-}\frac1n \sum_{s_1'} r_{k,\hat{k}, (\tilde{s}, \hat{s})\rightarrow (s', s_1')} && \hat{k} \ne k,\tilde{s} \ne s, \\
& \frac{\partial^2 f_{(k,s)}}{\partial x_{(k,\hat{s})} \partial x_{(\tilde{k},\tilde{s})}}(\bx) = \phantom{-}\frac1n \sum_{s_1'} r_{k,\tilde{k}, (\hat{s}, \tilde{s})\rightarrow (s, s_1')} && \tilde{k} \ne k,\\ 
& \frac{\partial^2 f_{(k,s)}}{\partial x_{(k,\hat{s})} \partial x_{(\tilde{k},\tilde{s})}}(\bx) = {-} \frac1n \sum_{s', s_1'} r_{k,\tilde{k}, (s, s') \rightarrow (\tilde{s}, s_1')} && \tilde{k} \ne k,\\
& \frac{\partial^2 f_{(k,s)}}{\partial x_{(\hat{k},\hat{s})} \partial x_{(\tilde{k},\tilde{s})}}(\bx) = 0 && \tilde{k}, \hat{k} \ne k.
\end{align*}

Note that if interactions of more than two objects occur, the above formulations include additional rates and higher order derivatives of the drift are non-zero.

\color{black}

\section{Cache Replacement Policies}
\subsection{Computation of the exact steady-state probabilities}
\label{apx:cache_adaptation}

It is shown in \cite{gastTransientSteadystateRegime2015} that the steady-state distribution of the RANDOM($\bm$) cache replacement policy has a product-form, which the authors use to derive the per-object miss probability. Here, we show how to adapt the same methodology to compute the steady-state probability for an object to be in list $s$. Our approach is very similar to the one developed in \cite{gastTransientSteadystateRegime2015} but leads to a slightly different recurrence equation.

Recall that $S_k$ denotes the list in which object $k$ is (where $0$ means that the object is not in the cache). We say that a state $\bS$ is admissible for $\bm$ if the number of objects in list $s$ is exactly $m_s$ for all $s\in\{1,\dots ,S\}$. Theorem~6 of \cite{gastTransientSteadystateRegime2015} can be rephrased as follows: For any admissible state $\bS$, the steady-state probability of $\bS$ is equal to
\begin{align*}
  \pi(\bS) = \frac1{C(\bm,n)}\prod_{k=1}^n (\lambda_k)^{S_k},
\end{align*}
where $C(\bm,n)=\sum_{\text{$\bS$ admissible for $\bm$}}\prod_{k=1}^n (\lambda_k)^{S_k}$ is a constant such that the probabilities $\pi(\bS)$ sum to one. Note that the constant $C(\bm,n)$ is not the same as the constant $E(\bm,n)$ defined in \cite{gastTransientSteadystateRegime2015} because our configuration $\bS$ does not take into account the position in a list in which an object is but only takes into account the list in which an object is: there is a $\prod_s m_s!$ factor between the two.

By decomposing the set of admissible configurations depending on the list in which object $n$ is (either outside the cache or in list $s$), we get that:
\begin{align*}
  C(\bm,n)=\sum_{\text{$\bS$ admissible for $\bm$}}\prod_{k=1}^n (\lambda_k)^{S_k}
  &=\sum_{s=0}^S \sum_{\text{$\bS$ admissible for $\bm$ and object $n$ is in list $s$}}\prod_{k=1}^n (\lambda_k)^{S_k}\\
  &= C(\bm,n-1)+\sum_{s=1}^S (\lambda_n)^{s}C(\bm-\be_s,n-1),
\end{align*}
with the convention that $C(\bm,n)=0$ if $\bm=\mathbf{0}$ or if $\sum_s m_s>n$. Indeed, there is a bijection between the admissible configurations for $\bm$ with $n$ objects in which object $n$ is in list $s$ and the configurations for $\bm-\be_s$ with $n-1$ objects.

Similarly, the probability for object $n$ to be in list $s$ is the sum over all admissible configurations such that object $n$ is in list $s$ which corresponds to the set of admissible configurations for $\bm-\be_s$ with $n-1$ objects. Hence, we have
\begin{align}
  \label{eq:cache_exact}
  \pi_{n,s}^{\text{exact}} = \frac{(\lambda_n)^s C(\bm-\be_s, n-1)}{C(\bm,n)}.
\end{align}
The above recurrence equations can be used to compute the exact value of $\Proba{S_k^{(n)}=s}$ for all $s$. By reordering the objects, it can be also used to compute the recurrence equation for all objects $k$. The naive complexity of such an equation grows in $O(n^2\prod_{s}m_s)$ and can be lowered to $O(n\log n\prod_{s}m_s)$ by carefully reordering the objects. This means that for relatively small values of $\bm$, it is possible to compute an exact value for $\Proba{S^{(n)}_k=s}$. Note that in practice, the complexity is quite large as soon as the list sizes grow. For instance, our implementation does not allow us to calculate the values for more than $3$ lists of size $10$.

{\color{myorange}
\subsection{Theorem~6 of \cite{gastTransientSteadystateRegime2015} is a consequence of our results (and can be refined)}
\label{apx:cache_adaptation2}

The cache replacement policy RAND($\bm$) that we study in Section~\ref{ssec:cache} is essentially the same\footnote{One difference between the two model is that we consider a continuous time model where object $k$ is requested at rate $\lambda_k$ and the authors of \cite{gastTransientSteadystateRegime2015} consider a discrete-time model where object $k$ is requested with probability $p_k=\lambda_k / \sum_{\ell}\lambda_\ell$. Up to re-normalizing the time by $\sum_{\ell}\lambda_\ell$, these two models are essentially equivalent.} as the one studied in \cite{gastTransientSteadystateRegime2015}. In \cite{gastTransientSteadystateRegime2015}, the authors denote by $H_s(t)=\sum_{k}p_k X_{k,s}(t)$ the sum of the items' popularity that are in list $s$ at time $t$, and by $\rho_s(t)=\sum_{k}p_k x_{k,s}(t)$ its mean field approximation, where $p_k=\lambda_k/(\sum_{k'}\lambda_{k'})$ is the request probability for object $k$. Theorem~6 of \cite{gastTransientSteadystateRegime2015} implies that for $t\le T$
\begin{align*}
  \esp{\norm{ H_s(t) - \rho_s(t)}^2} = O( \max_k p_k + \max_s \frac1{m_s}),
\end{align*}
which implies that if the popularities of items are such that $\lambda_k=O(1)$ and the list size are such that $m_s=O(1/n)$, then
\begin{align*}
  \esp{\norm{ H_s(t) - \rho_s(t)}^2} = O(1/n)
\end{align*}
This is a $O(1/\sqrt{n})$ convergence result because it implies that
\begin{align*}
  \esp{\norm{ H_s(t) - \rho_s(t)}} = O(1/\sqrt{n}).
\end{align*}
We do not think that the proof of the main result of \cite{gastTransientSteadystateRegime2015} is correct. Lemma 1 implies is that $M(t)$ is a Martingale such that $\esp{ \norm{M(t+1)-M(t)}^2} \le c$. Later in the proof, the authors argue that this implies that $\esp{\norm{M(t)}^2} \le ct$. This would hold if the norm could be written as a scalar product $\norm{M}_2^2 =\langle M,M\rangle$. Indeed, in such a case one would have:
\begin{align}
  \label{eq:L2_inequality}
  \esp{\norm{M(t+1)}_2^2} &= \esp{\norm{M(t)}_2^2 + 2\underbrace{\cro{M(t+1)-M(t), M(t)}}_{=0}+ \norm{M(t+1)-M(t)}_2^2}\\
  &\le \esp{\norm{M(t)}_2^2} + c,\nonumber
\end{align}
where the second term equals $0$ because $\esp{M(t+1)-M(t)\mid M(t)}=0$. A direct recurrence would imply that $\esp{\norm{M(t+1)}_2^2} \le ct$. 

The problem is that the norm used in \cite{gastTransientSteadystateRegime2015} can be written as a supremum norm (it is a supremum norm) and we not think that it can be written as a scalar product. This implies that one cannot use the reasoning of Equation~\eqref{eq:L2_inequality}, which means that this inequality does not hold for their case.

Yet, we claim that the result of their Theorem~6 holds, and can in fact be refined by using our approach. To see that, we rewrite the difference between $H$ and $\rho$ as:
\begin{align}
  \label{eq:L2_equality}
  \esp{\norm{ H_s(t) - \rho_s(t)}^2} &= \sum_{k_1,k_2} p_{k_1} p_{k_2} \esp{(X_{k_1,s}(t)-x_{k_1,s}(t))(X_{k_2,s'}(t)-x_{k_2,s'}(t))}
\end{align}
We claim that the proof of Theorem~\ref{th:RMF} can be adapted to show that:
\begin{align*}
  \esp{(X_{k_1,s}(t)-x_{k_1,s}(t))(X_{k_2,s'}(t)-x_{k_2,s'}(t))} &= \left\{
    \begin{array}{ll}
      w_{(k_1,s),(k_2,s')} + O(1/n) & \text{if $k_1=k_2$}\\
      w_{(k_1,s),(k_2,s')} + O(1/(n^2)) & \text{if $k_1\ne k_2$},
    \end{array}
  \right.
\end{align*}
where $w$ is defined in Appendix~\ref{apx:rmf_def} and is such that:
\begin{align*}
  w_{(k_1,s),(k_2,s')} = \left\{
  \begin{array}{ll}
    O(1) & \text{if $k_1=k_2$}\\
    O(1/n) & \text{if $k_1\ne k_2$},
  \end{array}
  \right.
\end{align*}

This implies that
\begin{align*}
  \esp{\norm{ H_s(t) - \rho_s(t)}^2}  = \underbrace{\sum_{k_1,k_2}  p_{k_1} p_{k_2} W_{(k_1,s),(k_2,s)}(t)}_{=O(1/n)} + O(1/n^2).
\end{align*}
The above equation refined Theorem~6 of \cite{gastTransientSteadystateRegime2015} by not only proving that the term is of order $O(1/n)$ but also by providing the expansion term.
}
\subsection{Cache replacement policies: time to compute the fixed point.}
\label{apx:cache_computation}

{\color{myorange}
In Table~\ref{tab:cache_steady-state}, we show that computing the fixed point of the refined mean field approximation takes less than $50$ms for a cache replacement model with two lists and $n=50$ heterogeneous objects.  To explore further how this computation time scales with $n$ or the number of lists of the cache, we report in Table~\ref{table:computation_time} the time to compute the fixed point of the mean field and refined mean field approximation for up to $1000$ items and between $2$ to $4$ lists. The total number of values to be computed here is $n\times S$ where $S$ is the number of lists. We observe that the mean field approximation is relatively fast to compute for all considered values. The refined mean field takes more time but remains reasonable when we have at most $n=1000$ objects. For $n=1000$ and $2$ lists, the computation times is much larger (more than $10$ times larger). We believe that this huge increase of computation time might be due to memory contention when \texttt{scipy} tries to solve a very big linear system (with $|n\calS|^2 = 4\text{ millions}$ of variables). 

\begin{table}[ht]
  \caption{Time to compute the fixed point of the mean filed and refined mean field approximation for the RAND($\bm$) model for various values of $n$ and $\bm$.}
  \label{table:computation_time}
  \color{myorange}
  \begin{tabular}{|ccc|c|c|}
    \hline
    $n$ & $\mathbf{m}$ & $n|\calS|$ & Time (mean field) & Time (Refined mean field)\\\hline
    30 & [6, 6, 6] & 90 & 40ms &50ms\\\hline
    50 & [10, 10, 10] & 150 & 50ms &72ms\\\hline
    100 & [20, 20, 20, 20] & 400 & 263ms &458ms\\\hline
    200 & [40, 40, 40] & 600 & 137ms &881ms\\\hline
    200 & [40, 40, 40, 40] & 800 & 370ms &2s\\\hline
    300 & [60, 60, 60] & 900 & 186ms &4s\\\hline
    500 & [150, 150] & 1000 & 121ms &6s\\\hline
    1000 & [300, 300] & 2000 & 222ms &71s\\\hline
  \end{tabular}
\end{table}

}


\section{Technical lemmas}
\label{apx:proofs}

\subsection{Bounds for Partial Derivatives of $\bphi$}
\label{apx:proof_bounds_partial_derivatives}

In Lemma \ref{lemma:properties_phi} we analyze the properties of the partial derivatives of $\phi_{(k,s)}(\bx,t)$ with respect to the initial condition $\bx$. We introduce the set $\calI := \{1, \dots,n \} \times \calS$ and $\calI_{k} = \{k \} \times \calS$ to simplify notations for frequently appearing sums in the lemma and proof. The set $\calI$ encompasses all object-state tuples $(k,s)$, the set $\calI_{k}$ includes tuples $(k,s)$ with fixed object $k$. We emphasize that  the bounds for the partial derivatives differ substantially depending on whether $\phi_{(k,s)}(\bx,t)$ is derived with respect to $(k,\hat{s}) \in \calI_k$ or $(\tilde{k},\tilde{s}) \in \calI \setminus \calI_k$. Our results show that if the sum over the states of the absolute values of the partial derivatives of $\phi_{(k,s)}(\bx,t)$, $\sum_{\hat{s}\in\calS} \abs{\frac{\phi_{(k,\hat{s})} }{x_{(k,\hat{s})}} (\bx,t)}$, is derived with respect to the same object $k$, i.e. in direction of a object-state pair $(k,\hat{s})$, it can be bounded independent of $n$. However, if the same sum is derived with respect to a pair $(\tilde{k},\tilde{s}) \in \calI \setminus \calI_k$ it is of order $O(1/n)$. Subsequently, we can show similar properties for sums of higher partial derivatives such as $\sum_{\hat{s}\in\calS} \abs{\frac{\partial \phi_{(k,\hat{s})}}{\partial x_{(\tilde{k},\tilde{s})} \partial x_{(\hat{k},\hat{s})} }(\bx,t)}$. For the second partial derivatives we see that if at least one of the derivative direction is in $\calI_k$ the sum is bounded by $O(1/n)$ and otherwise, if $(\tilde{k},\tilde{s}),(\hat{k},\hat{s}) \in \calI \setminus \calI_k$, the sum is of order $O(1/n^2)$. Our analysis considers partial derivatives up to the fourth order for which we establish bounds with likewise properties. A direct consequence we frequently use is that the absolute value of the partial derivative of $\phi_{(k,s)}(\bx,t)$ can be bounded by the previously mentioned sums, for example $\abs{\frac{\partial \phi_{(k,s)}}{\partial x_{(\tilde{k},\tilde{s})}}(\bx,t)} \leq \sum_{\hat{s}\in\calS} \abs{\frac{\phi_{(k,\hat{s})} }{\partial x_{(\hat{k},\hat{s})}} (\bx,t)}$. Thus, same bounds hold for $\abs{\frac{\partial \phi_{(k,s)}}{\partial x_{(\tilde{k},\tilde{s})}}(\bx,t)}$ and absolute values of higher order partial derivatives.

\begin{lemma}
\label{lemma:properties_phi}
Given the solution $\bphi$ of the ODE defined in section \ref{ssec:drift_mean_field}. For the partial derivatives of $\bphi_{(k,s)}$ with respect to the initial condition $\bx \in \calX$ and $(k,s)\in\calI = \{1,\ldots,n\} \times \calS$ the following properties hold:

\begin{enumerate}[label=(\alph*)]
  \item If $i,j,w$ and $l$ are in  $\calI \setminus \calI_k = \{1,\ldots,k-1,k+1,\ldots,n\}\times \calS$, i.e. none of the tuples $i,j,w$ or $l$ refer to object~$k$, then
\begin{align*}
(a.1)\quad& \sum_{s\in\calS} \abs{\frac{\partial \phi_{(k,s)}}{\partial x_{i}}(\bx,t)} = O(1/n),
& (a.2)\quad& \sum_{s\in\calS} \abs{\dd{\phi_{(k,s)}}{i}{j}(\bx,t)} = O(1/n^2), \\
(a.3)\quad& \sum_{s\in\calS} \abs{\ddd{\phi_{(k,s)}}{i}{j}{w}(\bx,t)} = O(1/n^3),
& (a.4)\quad& \sum_{s\in\calS} \abs{\dddd{\phi_{(k,s)}}{i}{j}{w}{l}(\bx,t)} = O(1/n^4).
\end{align*}
  \item Otherwise, if any tuple $i,j,w$ or $l$ is in  $\calI_k= \{k\}\times \calS$ then, for the same summations, 
\begin{align*}
(b.1)\quad& \sum_{s\in\calS} \abs{\frac{\partial \phi_{(k,s)}}{\partial x_{i}}(\bx,t)} = O(1),
& (b.2)\quad& \sum_{s\in\calS} \abs{\dd{\phi_{(k,s)}}{i}{j}(\bx,t)} = O(1/n), \\
(b.3) \quad & \sum_{s\in\calS} \abs{\ddd{\phi_{(k,s)}}{i}{j}{w}(\bx,t)} = O(1/n^2),
& (b.4) \quad& \sum_{s\in\calS} \abs{\dddd{\phi_{(k,s)}}{i}{j}{w}{l}(\bx,t)} = O(1/n^3).
\end{align*}

\end{enumerate}

\end{lemma}

\begin{proof}
  We will prove this lemma by bounding the derivative with respect to time of $\frac{\phi_{(k,s)}}{\partial x_i}(\bx,t)$ (and of the derivative of the higher order terms). The result will then follow by using Grönwall's Lemma in differential form.

  \textbf{First derivative -- proof of (a.1) and (b.1)} -- Recall that $\bphi(\bx,t)$ satisfies the differential equation $\frac{d}{dt}\bphi(\bx,t) = f(\bphi(\bx,t))$. Hence, the partial derivatives of $\frac{\phi_{(k,s)}}{\partial x_i}(\bx,t), \ i \in \calI$ with respect to the time $t$ are
\begin{align*}
\frac{d}{dt}\frac{\partial \phi_{(k,s)}}{\partial x_i}(\bx,t) = \frac{\partial}{\partial x_i} \frac{d\phi_{(k,s)}}{dt} (\bx,t) = \frac{\partial (f_{(k,s)} \circ \bphi)}{\partial x_i} (\bx,t) = \sum_{u\in \calI} \frac{\partial f_{(k,s)}}{\partial \phi_{u}}(\bphi(\bx,t))\frac{\partial \phi_{u}}{\partial x_{i}}(\bx,t). \\
\end{align*}
Having a closer look at the partial derivatives of $f$ (see Appendix~\ref{apx:rmf_def}), we see that

\begin{align*}
& \frac{\partial f_{(k,s)}}{\partial x_u}(\bx) \leq C_1 
&& \text{ for } u \in \calI_{k}, \\
& \frac{\partial f_{(k,s)}}{\partial x_u}(\bx) \leq C_2/n && \text{ for } u \in \calI \setminus \calI_{k}.
\end{align*}

Let $L_1 := \max\{C_1,C_2\}$ and define $c^k_{i}:=\begin{cases}1 \text{ if } i \in \calI_k \\ 1/n \text{ otherwise}\end{cases}$. It follows that
\begin{align*}
\sum_{u\in \calI} \frac{\partial f_{(k,s)}}{\partial \phi_{u}}(\bphi(\bx,t))\frac{\partial \phi_{u}}{\partial x_{i}}(\bx,t) \leq L_{1}\sum_{u\in \calI} c_u^k \abs{\frac{\partial \phi_{u}}{\partial x_{i}}(\bx,t)}.
\end{align*}

To obtain the bound on the max term, we start by bounding the change of the $\max$ with respect to time for the first partial derivatives. 
\begin{align*}
\frac{d}{dt}\max_{i\in\calI, s\in \calS} & \sum_{k=1}^{n} \abs{\frac{\partial \phi_{(k,s)}}{\partial x_{i}}(\bx,t)} 
\leq L_{1} \max_{i\in\calI, s\in \calS}\sum_{k=1}^{n} \sum_{u\in \calI}c_u^k \abs{\frac{\partial \phi_{u}}{\partial x_{i}}(\bx,t)} \\
& \leq L_{1}\abs{\calS} \max_{i\in\calI, s\in \calS} \sum_{k=1}^{n} \abs{\frac{\partial \phi_{(k,s)}}{\partial x_{i}}(\bx,t)} + L_{1}\frac{1}{n}\abs{\calS}(n-1) \max_{i\in\calI, s\in \calS} \sum_{k=1}^{n} \abs{\frac{\partial \phi_{(k,s)}}{\partial x_{i}}(\bx,t)} \\
& \leq 2L_{1}\abs{\calS} \max_{i\in\calI, s\in \calS} \sum_{k=1}^{n} \abs{\frac{\partial \phi_{(k,s)}}{\partial x_{i}}(\bx,t)}
\end{align*}

Furthermore, for $t$ equal to zero, $\bphi(\bx, 0)=\bx$ which implies that $\frac{\partial \phi_{(k,s)}(\bx, 0)}{\partial x_{i}}=1$ if $(k,s)=i$ and $0$ otherwise. From this it follows directly that $\max_{i\in\calI, s\in S} \sum_{k=1}^{n} \abs{\frac{\partial \phi_{(k,s)}}{\partial x_{i}}(\bx,0)} = 1$. Applying Grönwalls Lemma to the obtained results yields $ \max_{i\in\calI, s\in \calS} \sum_{k=1}^{n} \abs{\frac{\partial \phi_{(k,s)}}{\partial x_{i}}(\bx,t)} \leq \exp\bigl( 2L_{1} \abs{\calS} t \bigr) = O(1)$.

We bound $\sum_{s\in \calS} \abs{\frac{\partial \phi_{(k,s)}}{\partial x_i}(\bx,t)}$ in the same manner. First, for the time derivative
\begin{align*}
\frac{d}{dt}\sum_{s\in\calS} \abs{\frac{\partial \phi_{(k,s)}}{\partial x_i}(\bx,t)} &\leq L_{1} \abs{\calS} \sum_{u\in \calI} c^k_{u} \abs{\frac{\partial \phi_{(k,s)}}{\partial x_i}(\bx,t)} \\
& \leq L_{1} \abs{\calS} \left( \sum_{s\in\calS} \abs{\frac{\partial \phi_{(k,s)}}{\partial x_i}(\bx,t)} + \frac{1}{n}\abs{\calS}^2 \max_{i\in\calI, s\in \calS} \sum_{k=1}^{n} \abs{\frac{\partial \phi_{(k,s)}}{\partial x_{i}}(\bx,t)} \right).
\end{align*}

Here we know that the second summand is $O(1)$. By definition of $\bphi$, at time zero  $\sum_{s\in\calS} \abs{\frac{\partial \phi_{(k,s)}}{\partial x_i}(\bx,0)}$ is equal to one if $i$ is in $\calI_k$ and zero otherwise. Using Grönwalls Lemma it follows
\begin{align*}
  \sum_{s\in\calS} \abs{\frac{\partial \phi_{(k,s)}}{\partial x_i}(\bx,t)} = O(c^k_i) = 
  \begin{cases}
    O(1)  & \text{ for } i \in \calI_k, \\ 
    O(1/n)  & \text{ otherwise. }
  \end{cases}
\end{align*}
This shows (a.1) and (b.1). Note that as an important direct consequence, the same is true for $\abs{\frac{\partial \phi_{(k,s)}}{\partial x_i}(\bx,t)}$.
\medskip

\textbf{Second derivative -- proof of (a.2) and (b.2)} -- For the second partial derivatives we repeat the procedure of bounding the second derivative of $\bphi$ with respect to time $t$. Deriving the second partial derivative of $\bphi$ with respect to time $t$ gives

\begin{align}
   \frac{d}{dt} \frac{\partial^2 \phi_{(k,s)}}{\partial x_j\partial x_i}(\bx, t) &= \frac{\partial^2 f_{(k,s)}}{\partial x_j \partial x_i}(\bphi (\bx,t)) 
  = \frac{\partial}{\partial x_j} \left( \sum_{u \in \calI} \frac{\partial f_{(k,s)}}{\partial x_u}(\bphi (\bx,t)) \frac{\partial \phi_{u}}{\partial x_i}(\bx,t) \right)\label{eq:second_derivative_phi} \\
  & = \sum_{u,v \in \calI} \dd{f_{(k,s)}}{u}{v}(\bphi (\bx,t)) \frac{\partial \phi_{u}}{\partial x_i}(\bx,t) \frac{\partial \phi_{v}}{\partial x_j}(\bx,t) + \sum_{u \in \calI} \frac{\partial f_{(k,s)}}{\partial x_u}(\bphi (\bx,t)) \dd{\phi_u}{j}{i}(\bx,t).
  \nonumber
\end{align}
\color{myorange}
To bound the above term, we observe that 

\begin{align*}
\dd{f_{(k,s)}}{i}{j}(\bx) & \leq 
\begin{cases}
C_3/n \text{ if } i \text{ or } j \in \calI_k \\
C_4/n^2 \text{ otherwise} 
\end{cases} \\
\end{align*}
for $C_3, C_4 \geq 0$. We define$ L_2 = \max\{C_3,C_4\}$ with which we bound the first sum by

\begin{align*}
  & \sum_{u,v \in \calI} \abs{\dd{f_{(k,s)}}{u}{v}(\bphi (\bx,t))} \abs{\frac{\partial \phi_{u}}{\partial x_i}(\bx,t)} \abs{\frac{\partial \phi_{v}}{\partial x_j}(\bx,t)} \\
  & \leq \frac{L_{2}}{n}\left( \sum_{ u \in \calI_k} \abs{\frac{\partial \phi_{u}}{\partial x_i}(\bx,t)} \sum_{v \in \calI }\abs{\frac{\partial \phi_{v}}{\partial x_j}(\bx,t)} + \sum_{ v \in \calI_k} \abs{\frac{\partial \phi_{v}}{\partial x_i}(\bx,t)} \sum_{u \in \calI }\abs{\frac{\partial \phi_{u}}{\partial x_j}(\bx,t)}\right) \\
  & \phantom{=}+ \frac{L_2}{n^2}\left( \sum_{u,v \in \calI \setminus \calI_k} \abs{\frac{\partial \phi_{u}}{\partial x_i}(\bx,t)} \abs{\frac{\partial \phi_{v}}{\partial x_j}(\bx,t)} \right)
\end{align*}

By previous observations $\sum_{v \in \calI }\abs{\frac{\partial \phi_{v}}{\partial x_j}(\bx, t)} = O(1)$ and $\sum_{ u \in \calI_k} \abs{\frac{\partial \phi_{u}}{\partial x_i}(\bx,t)} = O(c^k_i)$ which implies that the first sum is $\frac{1}{n}L_{2}(O(c^k_i)+O(c^k_j)) + O(\frac{1}{n^2}) = O(\frac{1}{n}(c^k_i + c^k_j))$. 
\color{black}

The second sum can be bounded, similar to the first partial derivatives, by
\begin{align*}
\sum_{u \in \calI} \frac{\partial f^{(k,s)}}{\partial x_u}(\bphi (\bx,t)) \dd{\phi_u}{j}{i}(\bx,t) \leq L_{1} \sum_{u\in \calI} c_u^k\abs{\dd{\phi_{u} }{j}{i}(\bx,t)} .
\end{align*}
Now we can derive bounds for $\max_{i,j \in \calI, s\in S} \sum_k \abs{\dd{\phi_{(k,s)}}{i}{j}(\bx,t)}$ and $\sum_{s} \abs{\dd{\phi_{(k,s)}}{i}{j}(\bx,t)}$. The procedure is the same as before. For the $\max$ term we get
\begin{align*}
\frac{d}{dt}\max_{i,j\in\calI, s\in \calS} \sum_{k=1}^{n} \abs{\frac{\partial^2 \phi_{(k,s)}}{\partial x_j\partial x_i}(\bx, t) } & \leq \max_{i,j\in\calI, s\in \calS} \{\sum_{k=1}^{n} O(\frac{1}{n}(c^k_i + c^k_j)) + L_{1} \sum_{k=1}^{n}\sum_{u\in \calI} c_u^k\abs{\dd{\phi_{u} }{j}{i}(\bx,t)}\} \\
& \leq O(1/n) + 2L_{1}\abs{\calS} \max_{i,j\in\calI, s\in \calS} \sum_{k=1}^{n} \abs{\frac{\partial^2 \phi_{(k,s)}}{\partial x_j\partial x_i}(\bx, t)}.
\end{align*} 
Recall that $\bphi(\bx, 0)=\bx$ which implies that $\frac{\partial^2 \phi_{(k,s)}(\bx, 0)}{\partial x_{i}\partial x_j}=0$.  Hence, $\max_{i,j\in\calI, s\in \calS} \sum_{k=1}^{n} \abs{\frac{\partial^2 \phi_{(k,s)}}{\partial x_j\partial x_i}(\bx, 0) } = 0$ which allows to conclude, by applying Grönwalls Lemma, that
$\max_{i,j\in\calI, s\in \calS} \sum_{k=1}^{n} \abs{\frac{\partial^2 \phi_{(k,s)}}{\partial x_j\partial x_i}(\bx, t) } = O(1/n)$.
For $\sum_{s\in\calS} \abs{\dd{\phi_{(k,s)}}{i}{j}(\bx,t)}$ we infer
\begin{align*}
\frac{d}{dt}\sum_{s\in\calS} \abs{\dd{\phi_{(k,s)}}{i}{j}(\bx,t)} &\leq \sum_{s\in \calS} O(\frac{1}{n}(c^k_i + c^k_j)) + L_{1} \sum_{s\in\calS}\sum_{u\in \calI} c_u^k\abs{\dd{\phi_{u} }{j}{i}(\bx,t)} \\
& \leq O(\frac{1}{n}(c^k_i + c^k_j)) + L_{1}\abs{\calS} \sum_{s\in \calS} \abs{\dd{\phi_{(k,s)}}{i}{j}(\bx,t)} + \frac{1}{n}L_{1}\abs{\calS}^2 \max_{i,j\in\calI, s\in \calS} \sum_{k=1}^{n} \abs{\frac{\partial^2 \phi_{(k,s)}}{\partial x_j\partial x_i}(\bx, t)} \\
& = O(\frac{1}{n}(c^k_i + c^k_j)) + L_{1}\abs{\calS} \sum_{s\in \calS} \abs{\dd{\phi_{(k,s)}}{i}{j}(\bx,t)} + O(1/n^2).
\end{align*}
With $\sum_{s\in\calS} \abs{\dd{\phi_{(k,s)}}{i}{j}(\bx,0)} = 0$ and Grönwall, we have $\sum_{s\in\calS} \abs{\dd{\phi_{(k,s)}}{i}{j}(\bx,t)} = O(\frac{1}{n}(c^k_i + c^k_j))$. 

\medskip

\color{myorange}
\textbf{Third and Fourth derivatives} -- For the higher order partial derivatives the proof procedure stays the same as for the first and second partial derivatives. First, we calculate the time derivative for the partial derivatives of third and fourth order of $\bphi$. In order to obtain bounds for the $\max$ term and the sum over states, we bound the derivatives. The time derivatives of $\ddd{\phi_{(k,s)}}{i}{j}{w}(\bx,t)$ and $\dddd{\phi_{(k,s)}}{i}{j}{w}{l}(\bx,t)$, with $i,j,w,l \in \calI$, are given by 

\begin{align*}
& \frac{d}{dt} \ddd{\phi_{(k,s)}}{i}{j}{w}(\bx,t) \\
& = \frac{\partial}{\partial x_w} \left( \sum_{u,v \in \calI} \dd{f_{(k,s)}}{u}{v}(\bphi (\bx,t)) \frac{\partial \phi_{u}}{\partial x_i}(\bx,t) \frac{\partial \phi_{v}}{\partial x_j}(\bx,t) + \sum_{u \in \calI} \frac{\partial f_{(k,s)}}{\partial x_u}(\bphi (\bx,t)) \dd{\phi_u}{j}{i}(\bx,t)  \right) \\
& = \sum_{u,v,o\in \calI} \ddd{f_{(k,s)}}{u}{v}{o}(\bphi (\bx,t)) \frac{\partial \phi_{u}}{\partial x_i}(\bx,t) \frac{\partial \phi_{v}}{\partial x_j}(\bx,t) \frac{\partial \phi_{o}}{\partial x_w}(\bx,t) \\
& + \sum_{u,v \in \calI} \dd{f_{(k,s)}}{u}{v}(\bphi (\bx,t)) \left(\dd{\phi_{u}}{i}{w}(\bx,t) \frac{\partial \phi_{v}}{\partial x_j}(\bx,t) + \frac{\partial \phi_{u}}{\partial x_i}(\bx,t) \dd{\phi_{v}}{j}{w}(\bx,t) \right) \\
& \quad + \sum_{u,v \in \calI} \dd{f_{(k,s)}}{u}{v}(\bphi (\bx,t)) \dd{\phi_{u}}{i}{j}(\bx,t) \frac{\partial \phi_{v}}{\partial x_w}(\bx,t) + \sum_{u \in \calI} \frac{\partial f_{(k,s)}}{\partial x_u}(\bphi (\bx,t)) \ddd{\bphi_{u}}{j}{i}{w}(\bx,t)
\end{align*}
and
\begin{align*}
& \frac{d}{dt} \dddd{\phi_{(k,s)}}{i}{j}{w}{l}(\bx,t) \\
& = \frac{\partial}{\partial x_l} \bigl( \sum_{u,v,o\in \calI} \ddd{f_{(k,s)}}{u}{v}{o}(\bphi (\bx,t)) \frac{\partial \phi_{u}}{\partial x_i}(\bx,t) \frac{\partial \phi_{v}}{\partial x_j}(\bx,t) \frac{\partial \phi_{o}}{\partial x_w}(\bx,t) \\
& + \sum_{u,v \in \calI} \dd{f_{(k,s)}}{u}{v}(\bphi (\bx,t)) \left(\dd{\phi_{u}}{i}{w}(\bx,t) \frac{\partial \phi_{v}}{\partial x_j}(\bx,t) + \frac{\partial \phi_{u}}{\partial x_i}(\bx,t) \dd{\phi_{v}}{j}{w}(\bx,t) \right) \\
& \quad + \sum_{u,v \in \calI} \dd{f_{(k,s)}}{u}{v}(\bphi (\bx,t)) \dd{\phi_{u}}{i}{j}(\bx,t) \frac{\partial \phi_{v}}{\partial x_w}(\bx,t) + \sum_{u \in \calI} \frac{\partial f_{(k,s)}}{\partial x_u}(\bphi (\bx,t)) \ddd{\bphi_{u}}{j}{i}{w}(\bx,t) \bigr):
\end{align*}
The above quantity is equal to
\begin{align*}
& \sum_{u,v,o,p\in \calI} \dddd{f_{(k,s)}}{u}{v}{o}{p}(\bphi (\bx,t)) \frac{\partial \phi_{u}}{\partial x_i}(\bx,t) \frac{\partial \phi_{v}}{\partial x_j}(\bx,t) \frac{\partial \phi_{o}}{\partial x_w}(\bx,t)\frac{\partial \phi_{p}}{\partial x_l}(\bx,t)  \\
& \phantom{= }+ \sum_{u,v,o \in \calI} \ddd{f_{(k,s)}}{u}{v}{o}(\bphi (\bx,t)) \left( \dd{\phi_{u}}{i}{l}(\bx,t) \frac{\partial \phi_{v}}{\partial x_j}(\bx,t) \frac{\partial \phi_{o}}{\partial x_w}(\bx,t) + \frac{\partial \phi_{u}}{\partial x_i}(\bx,t) \dd{\phi_{v}}{j}{l}(\bx,t) \frac{\partial \phi_{o}}{\partial x_w}(\bx,t) \right. \\ 
& \phantom{\sum_{u,v,o \in \calI} \ddd{f_{(k,s)}}{u}{v}{o}(\bphi (\bx,t))} + \left. \frac{\partial \phi_{u}}{\partial x_i}(\bx,t) \frac{\partial \phi_{v}}{\partial x_j}(\bx,t) \dd{\phi_{o}}{w}{l}(\bx,t)\right)\\
& + \sum_{u,v \in \calI} \dd{f_{(k,s)}}{u}{v}(\bphi (\bx,t)) \left( \ddd{\phi_{u}}{i}{w}{l}(\bx,t) \frac{\partial \phi_{v}}{\partial x_j}(\bx,t) + \frac{\partial \phi_{u}}{\partial x_i}(\bx,t) \ddd{\phi_{v}}{j}{l}{w}(\bx,t) \right. \\
& \quad \phantom{\sum_{u,v \in \calI 123}  \dd{f_{(k,s)}}{u}{v}(\phi (\bx,t))}  \left. + \dd{\phi_{u}}{i}{w}(\bx,t) \dd{\phi_{v}}{j}{l}(\bx,t) + \dd{\phi_{u}}{i}{l}(\bx,t) \dd{\phi_{v}}{j}{w}(\bx,t) \right) \\
& \quad + \sum_{u,v \in \calI} \dd{f_{(k,s)}}{u}{v} (\bphi (\bx,t)) \left( \dd{\phi_{u}}{w}{l}(\bx,t) \dd{\phi_{v}}{i}{j}(\bx,t) + \frac{\partial \phi_{v}}{\partial x_w}(\bx,t) \ddd{\phi_{u}}{i}{j}{l}(\bx,t) \right) \\
& \quad + \sum_{u \in \calI} \frac{\partial f_{(k,s)}}{\partial x_u}(\bphi (\bx,t)) \dddd{\phi_{u}}{j}{i}{w}{l}(\bx,t).
\end{align*}
\color{black}

For the third partial derivatives, we use the above equation to show first that \\$\max_{i,j,w \in \calI,s \in \calS} \sum_{k=1}^{n} \abs{\ddd{\phi_{(k,s)}}{i}{j}{w}(\bx,t)}$ is of order $O(\frac{1}{n^2})$ and that $\sum_{s\in \calS} \abs{\ddd{\phi_{(k,s)}}{i}{j}{w}(\bx,t)}$ is of order $O\left( \frac{1}{n^2}(c^k_{i} + c^k_{w} + c^k_{j}) \right)$. To obtain a bound for the $\max$ term, we use the results obtained by the analysis of the first and second partial derivatives of $\bphi$. The overall aim is to apply Grönwalls Lemma. 
We bound the first three sums of the derivative, which include first and second order partial derivatives of $\bphi$. We use previous analysis and bounds on the drift derivatives to obtain the following asymptotic properties. The third order partial derivatives of the drift can be bounded by 

\color{myorange}
\begin{align*}
\ddd{f_{(k,s)}}{i}{j}{w}(\bx) & \leq 
\begin{cases}
C_5/n^2 \text{ if } i,j \text{ or } w \in \calI_k \\
C_6/n^3 \text{ otherwise,} 
\end{cases} 
\end{align*}
 with $C_5,C_6 \geq 0$ and we define $L_3 = \max\{C_5,C_6 \}$.

For the first sum
\begin{align*}
& \sum_{u,v,o\in \calI} \abs{\ddd{f_{(k,s)}}{u}{v}{o}(\bphi (\bx,t))} \abs{\frac{\partial \phi_{u}}{\partial x_i}(\bx,t)} \abs{\frac{\partial \phi_{v}}{\partial x_j}(\bx,t)} \abs{\frac{\partial \phi_{o}}{\partial x_w}(\bx,t)} \\
& \leq \frac{L_3}{n^3} \sum_{u,v,o \in \calI \setminus \calI_k} \abs{\frac{\partial \phi_{u}}{\partial x_i}(\bx,t)} \abs{\frac{\partial \phi_{v}}{\partial x_j}(\bx,t)} \abs{\frac{\partial \phi_{o}}{\partial x_w}(\bx,t)}
+ \frac{L_3}{n^2} \sum_{u \in \calI_k, v,o \in \calI} \abs{\frac{\partial \phi_{u}}{\partial x_i}(\bx,t)} \abs{\frac{\partial \phi_{v}}{\partial x_j}(\bx,t)} \abs{\frac{\partial \phi_{o}}{\partial x_w}(\bx,t)} \\
& \phantom{\leq} + \frac{L_3}{n^2} \sum_{v \in \calI_k, u,o \in \calI} \abs{\frac{\partial \phi_{u}}{\partial x_i}(\bx,t)}  \abs{\frac{\partial \phi_{v}}{\partial x_j}(\bx,t)} \abs{\frac{\partial \phi_{o}}{\partial x_w}(\bx,t)} + \frac{L_3}{n^2} \sum_{o \in \calI_k, u,v \in \calI} \abs{\frac{\partial \phi_{u}}{\partial x_i}(\bx,t)} \abs{\frac{\partial \phi_{v}}{\partial x_j}(\bx,t)} \abs{\frac{\partial \phi_{o}}{\partial x_w}(\bx,t)} \\
& = O\bigl(\frac{1}{n^3}\bigr) + O\left(\frac{1}{n^2}(c^k_i + c^k_j + c^k_w)\right) = O\left(\frac{1}{n^2}(c^k_i + c^k_j + c^k_w)\right)
\end{align*}
\color{black}
and for the second type of sums 
\begin{align*}
& \sum_{u,v \in \calI} \abs{\dd{f_{(k,s)}}{u}{v}(\bphi (\bx,t))} \abs{\dd{\phi_{u}}{i}{w}(\bx,t)} \abs{\frac{\partial \phi_{v}}{\partial x_j}(\bx,t)} \\
& \leq L_{2}\frac{1}{n^2} \sum_{u,v \in \calI \setminus \calI_k} \abs{\dd{\phi_{u}}{i}{w}(\bx,t)} \abs{\frac{\partial \phi_{v}}{\partial x_j}(\bx,t)}
 + L_{2}\frac{1}{n} \sum_{u\in \calI_k ,v \in \calI} \abs{\dd{\phi_{u}}{i}{w}(\bx,t)} \abs{\frac{\partial \phi_{v}}{\partial x_j}(\bx,t)} \\
& \phantom{\leq} + L_{2}\frac{1}{n} \sum_{u\in \calI,v \in \calI_k} \abs{\dd{\phi_{u}}{i}{w}(\bx,t)} \abs{\frac{\partial \phi_{v}}{\partial x_j}(\bx,t)} \\
& \leq O(\frac{1}{n^3}) + L_{2}\frac{1}{n}\left( O\bigl(\frac{1}{n}(c^k_{i} + c^k_{w})\bigr)\abs{S}O(1)) + O(c^k_{j})\abs{S}O(\frac{1}{n}) \right)
 = O\left(\frac{1}{n^2}(c^k_{i} + c^k_{w} + c^k_{j})\right).
\end{align*}
The above statement also holds for any permutation of $i,j$ and $w$. By summing the above terms over $k$ we see, by definition of the $c^k_{i}$'s, that $\sum_{k=1}^n\sum_{u,v \in \calI} \abs{\dd{f_{(k,s)}}{u}{v}(\bphi (\bx,t))} \abs{\dd{\phi_{u}}{i}{w}(\bx,t)} \abs{\frac{\partial \phi_{v}}{\partial x_j}(\bx,t)} = O(\frac{1}{n^2})$. The third sum of interest can be bounded by
\begin{align*}
 \sum_{k=1}^{n} \sum_{u \in \calI} \abs{\frac{\partial f_{(k,s)}}{\partial x_u}(\bphi (\bx,t))} \abs{\ddd{\phi_{u}}{j}{i}{w}(\bx,t)} \leq 2 \abs{S} K_f \max_{i,j,w\in \calI,s\in \calS} \{ \sum_{k=1}^{n} \abs{\ddd{\phi_{(k,s)}}{i}{j}{k}(\bx,t)} \}.
\end{align*}
We furthermore note that at time $t{=}0$ the third partial derivatives of $\phi_{(k,s)}$ are zero. In combination with the obtained bounds for the sums and by applying Grönwall it is shown that $\max_{i,j,w \in \calI,s \in \calS} \sum_{k=1}^{n} \abs{\ddd{\phi_{(k,s)}}{i}{j}{w}(\bx,t)} = O(\frac{1}{n^2})$. 

Next, we show that the sum over the states $\sum_{s\in \calS} \abs{\ddd{\phi_{(k,s)}}{i}{j}{w}(\bx,t)}$ is bounded by $O\left( \frac{1}{n^2}(c^k_{i} + c^k_{w} + c^k_{j}) \right)$. First, we recall that the first two sums of the corresponding derivative are bounded by $O\left( \frac{1}{n^2}(c^k_{i} + c^k_{w} + c^k_{j}) \right)$. Second, 
\begin{align*}
& \sum_{s\in \calS} \sum_{u \in \calI} \abs{\frac{\partial f_{(k,s)}}{\partial x_u}(\bphi (\bx,t))} \abs{\ddd{\phi_{u}}{j}{i}{w}(\bx,t)} \\
& \leq  K_f \sum_{s\in \calS} \left( \sum_{u \in \calI_k} \abs{\ddd{\phi_{u}}{j}{i}{w}(\bx,t)} + \frac{1}{n} \sum_{u \in \calI \setminus \calI_k}  \abs{\ddd{\phi_{u}}{j}{i}{w}(\bx,t)} \right) \\
&  = K_f\abs{S} \sum_{s\in \calS} \abs{\ddd{\phi_{(k,s)}}{j}{i}{w}(\bx,t)} +  O(1/n^3) .
\end{align*}
Summarized, we bound $\sum_{s\in \calS} \abs{\ddd{\phi_{(k,s)}}{i}{j}{w}(\bx,t)}$ by 
\begin{align*}
\sum_{s\in \calS} \abs{\ddd{\phi_{(k,s)}}{i}{j}{w}(\bx,t)} \leq O\left( \frac{1}{n^2}(c^k_{i} + c^k_{w} + c^k_{j}) \right) + O(1/n^3) + \sum_{s\in \calS} \abs{\ddd{\phi_{(k,s)}}{i}{j}{w}(\bx,t)}.
\end{align*}
Using $\sum_{s\in \calS} \abs{\ddd{\phi_{(k,s)}}{i}{j}{w}(\bx,0)} = 0$ and applying Grönwalls Lemma proofs the claim. 

For the $\max$ term and the sum over the states of the fourth partial derivatives we repeat the same steps. First, we show that $\max_{i,j,w,l \in \calI s\in \calS} \sum_{k=1}^{n} \abs{\dddd{\phi_{(k,s)}}{i}{j}{w}{l}(\bx,t)}$ is bounded by $O(1/n^3)$. We bound the sums which contain first, second and third partial derivatives of $\bphi$. We use bounds on the derivatives of the drift up to the fourth order, for which \color{myorange}
\begin{align*}
\dddd{f_{(k,s)}}{i}{j}{w}{l}(\bx) & \leq 
\begin{cases}
C_7/n^3 \text{ if } i,j,w \text{ or } l \in \calI_k \\
C_8/n^4 \text{ otherwise,}
\end{cases} 
\text{and define } L_4 = \max\{C_7,C_8 \}.
\end{align*}
The bounds are
\begin{align*}
& \sum_{u,v,o,p\in \calI} \abs{\dddd{f_{(k,s)}}{u}{v}{o}{p}(\bphi (\bx,t))} \abs{\frac{\partial \phi_{u}}{\partial x_i}(\bx,t)} \abs{\frac{\partial \phi_{v}}{\partial x_j}(\bx,t)} \abs{\frac{\partial \phi_{o}}{\partial x_w}(\bx,t)} \abs{\frac{\partial \phi_{p}}{\partial x_l}(\bx,t)} \\
& \leq \frac{L_4}{n^4} \sum_{u,v,o,p\in \calI \setminus \calI_k} \abs{\frac{\partial \phi_{u}}{\partial x_i}(\bx,t)} \abs{\frac{\partial \phi_{v}}{\partial x_j}(\bx,t)} \abs{\frac{\partial \phi_{o}}{\partial x_w}(\bx,t)} \abs{\frac{\partial \phi_{p}}{\partial x_l}(\bx,t)} \\
& \phantom{\leq} + \frac{L_4}{n^3}\sum_{u \in \calI_k v,o,p \in \calI}  \abs{\frac{\partial \phi_{u}}{\partial x_i}(\bx,t)} \abs{\frac{\partial \phi_{v}}{\partial x_j}(\bx,t)} \abs{\frac{\partial \phi_{o}}{\partial x_w}(\bx,t)} \abs{\frac{\partial \phi_{p}}{\partial x_l}(\bx,t)} \\
& \phantom{\leq} + \hdots + \frac{L_4}{n^3}\sum_{p \in \calI_k u,v,o \in \calI} \abs{\frac{\partial \phi_{u}}{\partial x_i}(\bx,t)} \abs{\frac{\partial \phi_{v}}{\partial x_j}(\bx,t)} \abs{\frac{\partial \phi_{o}}{\partial x_w}(\bx,t)} \abs{\frac{\partial \phi_{p}}{\partial x_l}(\bx,t)} \\
& = O(\frac{1}{n^4}) + O\bigl(\frac{1}{n^3}(c^k_i + c^k_j + c^k_w + c^k_l)\bigr),
\end{align*}

\begin{align*}
& \sum_{u,v,o \in \calI}\abs{ \ddd{f_{(k,s)}}{u}{v}{o}(\bphi (\bx,t))} \abs{\dd{\phi_{u}}{i}{l}(\bx,t)} \abs{\frac{\partial \phi_{v}}{\partial x_j}(\bx,t)} \abs{\frac{\partial \phi_{o}}{\partial x_w}(\bx,t)} \\
& \leq \frac{L_3}{n^3} \sum_{u,v,o \in \calI \setminus \calI_k} \abs{\dd{\phi_{u}}{i}{l}(\bx,t)} \abs{\frac{\partial \phi_{v}}{\partial x_j}(\bx,t)} \abs{\frac{\partial \phi_{o}}{\partial x_w}(\bx,t)} + \frac{L_3}{n^2} \sum_{u \in \calI_k, v,o \in \calI} \abs{\dd{\phi_{u}}{i}{l}(\bx,t)} \abs{\frac{\partial \phi_{v}}{\partial x_j}(\bx,t)} \abs{\frac{\partial \phi_{o}}{\partial x_w}(\bx,t)} \\
& \phantom{\leq} + \hdots + \frac{L_3}{n^2} \sum_{o \in \calI_k, u,v \in \calI} \abs{\dd{\phi_{u}}{i}{l}(\bx,t)} \abs{\frac{\partial \phi_{v}}{\partial x_j}(\bx,t)} \abs{\frac{\partial \phi_{o}}{\partial x_w}(\bx,t)} \\
& = O(\frac{1}{n^4}) + O\left(\frac{1}{n^3}(c^k_i + c^k_j + c^k_w + c^k_l)\right),
\end{align*}

\begin{align*}
& \sum_{u,v \in \calI} \abs{\dd{f_{(k,s)}}{u}{v}(\bphi (\bx,t))} \abs{\ddd{\phi_{u}}{i}{w}{l}(\bx,t)} \abs{\frac{\partial \phi_{v}}{\partial x_j}(\bx,t)} \\
& \leq L_{2}\frac{1}{n^2}\sum_{u, v \in \calI \setminus \calI_k} \abs{\ddd{\phi_{u}}{i}{w}{l}(\bx,t)}\abs{\frac{\partial \phi_{v}}{\partial x_j}(\bx,t)} \\
& \phantom{\leq} L_{2}\frac{1}{n}\left(\sum_{u \in \calI_k,v\in \calI} \abs{\ddd{\phi_{u}}{i}{w}{l}(\bx,t)} \abs{\frac{\partial \phi_{v}}{\partial x_j}(\bx,t)} + \sum_{u \in \calI,v\in \calI_k} \abs{\ddd{\phi_{u}}{i}{w}{l}(\bx,t)} \abs{\frac{\partial \phi_{v}}{\partial x_j}(\bx,t)}\right) \\
& = O(\frac{1}{n^4}) + L_{2}\frac{1}{n}\left( O(\frac{1}{n^2}(c^k_{i} + c^k_{w} + c^k_{l}))O(1) + O(\frac{1}{n^2})O(c^k_{j}) \right) = O(\frac{1}{n^3}(c^k_{i} + c^k_{w} + c^k_{l} + c^k_{j}))
\end{align*}

and

\begin{align*}
& \sum_{u,v \in \calI} \abs{\dd{f_{(k,s)}}{u}{v}(\bphi (\bx,t))} \abs{\dd{\phi_{u}}{i}{w}(\bx,t)} \abs{\dd{\phi_{v}}{j}{l}(\bx,t)} \\
& \leq L_{2}\frac{1}{n^2}\sum_{u, v\in \calI \setminus \calI_k} \abs{\dd{\phi_{u}}{i}{w}(\bx,t)} \abs{\dd{\phi_{v}}{j}{l}(\bx,t)} \\
& \phantom{\leq}+ L_{2}\frac{1}{n}\left( \sum_{u \in \calI_k,v\in \calI} \abs{\dd{\phi_{u}}{i}{w}(\bx,t)} \abs{\dd{\phi_{v}}{j}{l}(\bx,t)} + \sum_{u \in \calI,v\in \calI_k} \abs{\dd{\phi_{u}}{i}{w}(\bx,t)} \abs{\dd{\phi_{v}}{j}{l}(\bx,t)} \right) \\
& = L_{2}\frac{1}{n}\left( O(\frac{1}{n}(c^k_{i} + c^k_{w}) )O(\frac{1}{n}) + O(\frac{1}{n}(c^k_{j} + c^k_{l}) )O(\frac{1}{n})  \right) \\
& = O(\frac{1}{n^3}(c^k_{i} + c^k_{w} + c^k_{j} + c^k_{l})).
\end{align*}

\color{black}

Note that the results hold for permutations of $i,j,w,l$.
The remaining sum which appears in $\frac{d}{dt} \dddd{\phi_{(k,s)}}{i}{j}{w}{l}(\bx,t)$ is $\sum_{u \in \calI} \frac{\partial f_{(k,s)}}{\partial x_u}(\bphi (\bx,t)) \dddd{\phi_{u}}{j}{i}{w}{l}(\bx,t)$. We see that by summing over $k$ and applying the $\max$, this term is bounded by
\begin{align*}
& \max_{i,j,w,l\in \calI ,s\in \calS} \sum_{k=1}^{n}  \sum_{u \in \calI} \abs{\frac{\partial f_{(k,s)}}{\partial x_u}(\bphi (\bx,t))} \abs{\dddd{\phi_{u}}{j}{i}{w}{l}(\bx,t)} \\
& \qquad \leq 2\abs{S}K_f \max_{i,j,w,l\in \calI,s\in \calS} \{\sum_{k=1}^{n} \abs{\dddd{\phi_{u}}{j}{i}{w}{l}(\bx,t)}\}.
\end{align*} 
Furthermore, $\max_{i,j,w,l\in \calI,s\in \calS} \{\sum_{k=1}^{n} \abs{\dddd{\phi_{u}}{j}{i}{w}{l}(\bx,0)}\}$ is zero. We see that $\sum_{k=1}^{n} O(\frac{1}{n^3}(c^k_{i} + c^k_{w} + c^k_{j} + c^k_{l})) = O(\frac{1}{n^3})$ and, by applying Grönwall, it follows that $\max_{i,j,w,l \in \calI s\in \calS} \sum_{k=1}^{n} \abs{\dddd{\phi_{(k,s)}}{i}{j}{w}{l}(\bx,t)} = O(1/n^3)$. At last, we show that $\sum_{s\in \calS} \abs{\dddd{\phi_{(k,s)}}{i}{j}{w}{l}(\bx,t)}$ is bounded by $O\left( \frac{1}{n^3}(c^k_{i} + c^k_{j} + c^k_{w} + c^k_{l}) \right)$.

The proof follows the same principles as before for the third partial derivatives. The term $\frac{d}{dt} \dddd{\phi_{(k,s)}}{i}{j}{w}{l}(\bx,t)$ can be separated into sums which are of order $O(\frac{1}{n^3}(c^k_{i} + c^k_{w} + c^k_{l} + c^k_{j}))$ and the additional term $\sum_{s\in \calS} \sum_{u \in \calI} \abs{\frac{\partial f_{(k,s)}}{\partial x_u}(\bphi (\bx,t))} \abs{\dddd{\phi_{u}}{j}{i}{w}{l}(\bx,t)}$. The latter is bounded by 
\begin{align*}
& \sum_{s\in \calS} \sum_{u \in \calI} \abs{\frac{\partial f_{(k,s)}}{\partial x_u}(\bphi (\bx,t))} \abs{\dddd{\phi_{u}}{j}{i}{w}{l}(\bx,t)} \leq  K_f\sum_{s\in \calS} \left( \sum_{u \in \calI_k} \abs{\dddd{\phi_{u}}{j}{i}{w}{l}(\bx,t)} \right. \\
& \quad \left. + \frac{1}{n} \sum_{u \in \calI \setminus \calI_k}  \abs{\dddd{\phi_{u}}{j}{i}{w}{l}(\bx,t)} \right) 
\leq \abs{S} \sum_{s\in \calS}  \abs{\dddd{\phi_{(k,s)}}{j}{i}{w}{l}(\bx,t)} + \frac{1}{n}\abs{S}^2 O(\frac{1}{n^3}).
\end{align*}

To conclude, we use the same steps as before and see that  
\begin{align*}
\sum_{s\in \calS} \abs{\dddd{\phi_{(k,s)}}{i}{j}{w}{l}(\bx,t)} = O\left( \frac{1}{n^3}(c^k_{i} + c^k_{j} + c^k_{w} + c^k_{l}) \right). 
\end{align*}
\end{proof}

\subsection{Bounds for Taylor Remainders}

Lemma \ref{lemma:taylor_remainder} gives bounds for sums of weighted remainder terms appearing in the proofs of Theorem \ref{th:MF} and \ref{th:RMF}. We respectively bound the weighted sums of the first and second order Taylor remainder term by two suprema which are of order $O(1/n)$ and $O(1/n^2)$. 

\begin{lemma}
  \label{lemma:taylor_remainder}
  For $\bx \in \calX$ and $\tau \in \R_+$, the remainder terms satisfy
  \begin{align*}
    \E [ \sum_{\bx' \in \calX}K_{\bx,\bx'}R_1(\bx,\bx',\tau) ] &\le \sup_{\bx,\by\in\conv{\calX}} \frac{1}{2}\sum_{i,j\in \calI} \abs{\frac{\partial^2 \phi^{(k,s)}}{\partial x_{i} \partial x_{j} }(\by, \tau)}\abs{Q_{i,j}(\bx)} &= O(1/n),\quad\\
    \E [ \sum_{\bx' \in \calX}K_{\bx,\bx'}R_2(\bx,\bx',\tau) ] &\le \sup_{\bx,\by\in\conv{\calX}} \frac{1}{6} \sum_{i,j,u\in \calI}\abs{\frac{\partial^3 \phi^{(k,s)}}{\partial x_{i} \partial x_{j} \partial x_{u}}(\by, \tau)}\abs{R_{i,j,u}(\bx)} &= O(1/n^2).
  \end{align*}
\end{lemma}

Before starting the proof, recall that the first and second order remainder terms $R_1$ and $R_2$ defined in Section~\ref{sec:proofs} are expressed as:
\begin{align*}
& R_1(\bx,\bx',\tau) = \phantom{\frac{1}{2}} \int_0^1 (1{-}\nu)\sum_{i,j\in\calI} \frac{\partial^2 \bphi}{\partial x_{i} \partial x_{j}}(\bx + \nu(\bx'-\bx) , \tau) (\bx'_i{-}\bx_i)(\bx'_{j}{-}\bx_{j}) d\nu, \\
& R_2(\bx,\bx',\tau)=\frac{1}{2} \int_0^1 (1{-}\nu)^2 \sum_{i,j,u\in\calI} \frac{\partial^3 \bphi}{\partial x_{i} \partial x_{j} \partial x_{u}} (\bx + \nu(\bx'{-}\bx) , \tau) (\bx'_i{-}\bx_i)(\bx'_j{-}\bx_j)(\bx'_u{-}\bx_u)d\nu,
\end{align*}
and that, as defined in Appendix~\ref{apx:rmf_def}, $\bQ$ and $\bR$ are given by:
\begin{align*}
  \bQ(\bx) &= \sum_{\bx' \in \calX} K_{\bx,\bx'}(\bx'{-}\bx)^{\otimes 2} \text{\qquad and \qquad}
  \bR(\bx) = \sum_{\bx' \in \calX} K_{\bx,\bx'}(\bx'{-}\bx)^{\otimes 3},
\end{align*}
where $(\bx'{-}\bx)^{\otimes 2}$ and $(\bx'{-}\bx)^{\otimes 3}$ are Kronecker products of $(\bx'{-}\bx)$ with itself, i.e. $(\bx'{-}\bx)^{\otimes 2}_{i,j} = (\bx'_i{-}\bx_i)(\bx'_j{-}\bx_j)$ and $(\bx'{-}\bx)^{\otimes 3}_{i,j,u} = (\bx'_i{-}\bx_i)(\bx'_j{-}\bx_j)(\bx'_u{-}\bx_u)$. The two tensors $\bQ$ and $\bR$ can be naturally extended to $\conv{\calX}$ due to their entries being polynomials. 

\begin{proof}
To prove the two statements we first introduce some simplifying notations. We define $c^k_{k_1}$, with $k,k_1 \in \{1,\ldots,n\}$, to be one if $k$ equals $k_1$ and $1/n$ otherwise. By the definition of $\bQ$ it follows that for $\bx\in\conv{\calX}$:
\begin{align*}
  \abs{Q_{(k,s),(k_1,s_1)}(\bx)} = O(c^{k}_{k_1}) =
  \begin{cases}
    O(1)   & \text{ if } k = k_1, \\
    O(1/n) & \text{ otherwise.}
  \end{cases} 
\end{align*}
This can be seen by writing the elements of $\bQ$ based on the two transition types \eqref{eq:unilateral} and \eqref{eq:interactions}, as shown in Section~\ref{apx:rmf_def}. 

From Lemma \ref{lemma:properties_phi} we know $\abs{\dd{\phi_{(k,s)}}{(k_1,s_1)}{(k_2,s_2)}(\by,\tau)} = O(\frac{1}{n}(c^k_{k_1} + c^k_{k_2}))$ which holds for any $\by \in \conv{\calX}$. Here, the big $O$ notation hides the dependence on $\tau, L_{1}, L_{2}$ and $\abs{\calS}$. By noting that $\sum_{k_1,k_2}c_{k_2}^{k_1}c^{k}_{k_1} = O(1)$, we conclude 
\begin{align*}
\sum_{(k_1,s_1),(k_2,s_2)\in \calI}\abs{Q_{(k_1,s_1),(k_2,s_2)}(\bx)} \abs{\frac{\partial^2 \phi^{(k,s)}}{\partial x_{(k_1,s_1)} \partial x_{(k_2,s_2)} }(\by, \tau)} = \sum_{(k_1,s_1),(k_2,s_2)\in \calI} O(c^{k_1}_{k_2}) O(\frac{1}{n}(c^k_{k_1} + c^k_{k_2})) = O(1/n),
\end{align*}
where we hide the dependence on $\abs{\calS}$. \color{myorange} To prove the second statement we define 

\begin{align*}
c_{k_1,k_2,k_3} = 
\begin{cases}
1 & \text{ if } k_1 = k_2 = k_3\\
\frac{1}{n} &   \text{ if } k_1 = k_2 \ne k_3 \text{ or } k_2 = k_3 \ne k_1 \text{ or } k_1 = k_3 \ne k_2 \\
\frac{1}{n^2} & \text{ otherwise. }
\end{cases}
\end{align*}

By explicitly rewriting the entries of $\bR$ as done in Section~\ref{apx:rmf_def} for $\bQ$, the tensor $\mathbf{R}$ is such that
\begin{align*}
  \abs{R_{(k_1,s_1),(k_2,s_2),(k_3,s_3)}(\bx)} = O(c_{k_1,k_2,k_3})  
\end{align*}
Lemma \ref{lemma:properties_phi} states that the third partial derivatives of $\bphi$ are bounded by $\abs{\frac{\partial^3 \phi^{(k,s)}}{\partial x_{(k_1,s_1)} \partial x_{(k_2,s_2)} \partial x_{(k_3,s_3)}}(\by, \tau)} = O(\frac{1}{n^2}(c^k_{k_1}+ c^k_{k_2}+c^k_{k_3}))$. From $\sum_{k_1,k_2,k_3}c_{k_1,k_2,k_3}(c^k_{k_1}+ c^k_{k_2}+c^k_{k_3}) = O(1)$, it follows that the sum of the two terms above behaves as \color{black}
\begin{align*}
\sum_{(k_1, s_1),(k_2, s_2),(k_3, s_3) \in \calI } \abs{R_{(k_1,s_1),(k_2,s_2),(k_3,s_3)}(\bx)}\abs{\frac{\partial^3 \phi^{(k,s)}}{\partial x_{(k_1,s_1)} \partial x_{(k_2,s_2)} \partial x_{(k_3,s_3)}}(\by, \tau)} = O(1/n^2).
\end{align*}
\end{proof}

\subsection{Connection of Differential and Integral Form for the Refinement Term}

The following Lemma \ref{lemma:connection_ode_integral_form} shows how to express the refinement term $v$ and $w$ in integral form. Both representations are of importance since we exploit the differential form for numerical computations whereas we use the integral form in the proofs of Theorem \ref{th:RMF} and Lemma \ref{lemma:refinement_bound} which are related to the accuracy of the refined mean field approximation. 
\begin{lemma}
  \label{lemma:connection_ode_integral_form}
  The solutions to the system of ODEs 
  \begin{align*}
    \frac{d}{dt}{v}_{(k,s)}(\bx,t) &= \sum_{u \in \calI} \frac{\partial \bdrift_{(k,s)}}{\partial x_u}(\bphi(\bx,t)) v_{u}(\bx,t) + \frac{1}{2}\sum_{u,l \in \calI} \frac{\partial^2\bdrift_{(k,s)}}{\partial x_l \partial x_u}(\bphi(\bx,t)))w_{u,l}(\bx,t), \\
    \frac{d}{dt}{w}_{(k_1,s_1),(k_2,s_2)}(\bx,t) &= \sum_{u \in \calI}w_{u,(k_2,s_2)}(\bx,t) \frac{\partial \bdrift_{(k_1,s_1)}}{\partial x_{u}}(\bphi(\bx,t)) + \sum_{u \in \calI}w_{u,(k_1,s_1)}(\bx,t) \frac{\partial \bdrift_{(k_2,s_2)}}{\partial x_u}(\bphi(\bx,t)) \\
    & \phantom{\dot{w}_{(k_1,s_1),(k_2,s_2)}(\bx,t) == } + Q_{(k_1,s_1),(k_2,s_2)}(\bphi(\bx,t))
  \end{align*}
  can be expressed in integral form as
  \begin{align*}
    v_{(k,s)}(\bx, t) & = \frac{1}{2}\int_0^t\sum_{i,j\in\calI} Q_{i,j}(\bphi(\bx,\tau))\dd{\phi_{(k,s)}}{i}{j}(\bphi(\bx,\tau),t-\tau)d\tau, \\
    w_{(k_1,s_1),(k_2,s_2)}(\bx, t) & = \int_0^t \sum_{i,j \in \calI} Q_{i,j}(\bphi(\bx,\tau))\frac{\partial \phi_{(k_1,s_1)}}{\partial x_i}(\bphi(\bx,\tau),t-\tau) \frac{\partial \phi_{(k_2,s_2)}}{\partial x_j}(\bphi(\bx,\tau),t-\tau) d\tau.
  \end{align*}

\end{lemma}

\begin{proof}
For a sufficiently differentiable function $h:\R \times \R \mapsto \R$ we have 
\begin{align*}
  \frac{d}{dt}\int_0^t h(\tau,t)d\tau = h(t,t) + \int_0^t \frac{\partial h}{\partial t} (\tau,t)d\tau.
\end{align*}
We define $h(\tau, t) = \sum_{i,j\in\calI} Q_{i,j}(\bphi(\bx,\tau))\dd{\phi_{(k,s)}}{i}{j}(\bphi(\bx,\tau),t-\tau)$. Recall that $\dd{\phi_{(k,s)}}{i}{j}(\bphi(x,t),0) = 0$ which implies $h(t,t) = 0$. To calculate $\frac{\partial h}{\partial t}(\tau, t)$, we use the identity 
\begin{align*}
\frac{d}{dt}\dd{\phi_{(k,s)}}{i}{j}(\bphi(\bx,\tau),t-\tau) &= \dd{}{i}{j} \frac{d}{dt} \phi_{(k,s)} (\bphi(\bx,\tau),t-\tau)\\
&= \dd{}{i}{j}f_{(k,s)}(\bphi(\bphi(\bx,\tau),t-\tau)) \\
&= \dd{}{i}{j}f_{(k,s)}(\bphi(\bx,t))\\
& = \sum_{u,l\in \calI} \dd{f_{(k,s)}}{u}{l}(\bphi(\bx,t)) \frac{\partial \phi_u}{\partial x_i}(\bphi(\bx,\tau),t-\tau) \frac{\partial \phi_l}{\partial x_j}(\bphi(\bx,\tau),t-\tau)\\
& \qquad  + \sum_{u\in\calI} \frac{\partial f_{(k,s)}}{\partial x_u}(\bphi(\bx,t)) \dd{\phi_u}{i}{j}(\bphi(\bx,\tau),t-\tau),
\end{align*}
where the last term is the same as the one derived in \eqref{eq:second_derivative_phi}. 

Combining these results and rearranging terms leads to 
\begin{align*}
\frac{d}{dt} &\frac{1}{2}\int_0^t\sum_{i,j\in\calI} Q_{i,j}(\bphi(\bx,\tau))\dd{\phi_{(k,s)}}{i}{j}(\bphi(\bx,\tau),t-\tau)d\tau \\
& = \sum_{u \in \calI} \frac{\partial f_{(k,s)}}{\partial x_u}(\bphi(\bx,t)) \underbrace{\frac{1}{2}\int_0^t \sum_{i,j\in\calI} Q_{i,j}(\bphi(\bx,\tau)) \dd{\phi_u}{i}{j}(\bphi(\bx,\tau),t-\tau) d\tau}_{v_u(\bx,t)} \\
& \quad + \sum_{u,l\in \calI} \dd{f_{(k,s)}}{u}{l}(\bphi(\bx,t)) \frac{1}{2}\underbrace{\int_0^t \sum_{i,j\in\calI} Q_{i,j}(\bphi(\bx,\tau)) \frac{\partial \phi_u}{\partial x_i}(\bphi(\bx,\tau),t-\tau) \frac{\partial \phi_l}{\partial x_j}(\bphi(\bx,\tau),t-\tau) d\tau}_{w_{u,l}},
\end{align*}
which is the ODE describing $v_{(k,s)}(\bx,t)$. We obtain the integral form for $w_{(k_1,s_1),(k_2,s_2)}(\bx, t)$ by application of the same steps. 

\end{proof}

\subsection{Comparison of the Refinement Term $\bv$ and the Quadratic Taylor Term}

In Lemma~\ref{lemma:refinement_bound} below, we bound the difference of the refinement term $\bv$ and the quadratic term of the second order Taylor expansion appearing in the proof of Theorem \ref{th:RMF}. By defining $g_{(k,s)}(\by, \tau) = \sum_{i,j \in \calI} Q_{i,j}(\by) \frac{\partial^2 \phi_{(k,s)}}{\partial x_{i} \partial x_{j}} (\by, t-\tau)$ we see that the entries of the refinement term $\bv$ in integral form can be expressed as $v_{(k,s)}(\bx, t) = \frac{1}{2}\int_0^t g_{(k,s)}(\bphi(\bx, \tau), \tau) d\tau$. Similarly, the time integral over the expectation of the quadratic term of the Taylor expansion is given by $\frac{1}{2}\int_0^t\!\! \ \E [ g_{(k,s)}(\bX(\tau),\tau)]d\tau$. The latter arises due to the comparison of generator approach used in the proof of Theorem \ref{th:RMF} and the subsequent Taylor expansion of order two. The lemma shows that the difference of the two terms decreases quadratically with the system size~$n$ and allows, in combination with Lemma~\ref{lemma:taylor_remainder}, to obtain the accuracy bounds for the refined mean field approximation.

\begin{lemma}
\label{lemma:refinement_bound}

Define $g_{(k,s)}(\by, \tau) = \sum_{i,j \in \calI} Q_{i,j}(\by) \frac{\partial^2 \phi_{(k,s)}}{\partial x_{i} \partial x_{j}} (\by, t-\tau)$ with $\bphi$ being the solution to the ODE defined in Section \ref{ssec:drift_mean_field} and $\bQ$ as defined in Appendix \ref{apx:rmf_def}. Then
  \begin{align*}
   \frac{1}{2}\int_0^t \E [ g_{(k,s)}(\bX(\tau),\tau) - g_{(k,s)}(\bphi(\bx, \tau),\tau)]d\tau = O(1/n^2). 
  \end{align*}
\end{lemma}

\begin{proof}

We follow a similar proof concept as in Theorem~\ref{th:MF}. First, we define $h_\tau(\by) = g_{(k,s)}(\by,\tau)$ and rewrite 
\begin{align}
\label{eq:difference_g}
& \frac{1}{2}\int_0^t \E [ g_{(k,s)}(\bX(\tau),\tau) - g_{(k,s)}(\bphi(\bx, \tau),\tau)]d\tau = \frac{1}{2} \int_0^t \E [h_\tau (\bX(\tau)) - h_\tau(\bphi(\bx, \tau))]d\tau. 
\end{align}
By definition of $g_{(k,s)}$, $h_\tau$ is twice continuously differentiable. Second, using Lemma \ref{lemma:generator_comparison}, we see that $\E [h_\tau (\bX(\tau)) - h_\tau(\bphi(\bx, \tau))]$ is equal to 
\begin{align}
  \int_0^\tau \E[ \sum_{\bx'\in \calX} K_{\bX(\nu),\bx'} \bigl( h_\tau(\bphi(\bx',\tau{-}\nu)) -  h_\tau(\bphi(\bX(\nu),\tau{-}\nu)) \bigr) 
  - D_x(h_\tau \circ \bphi)(\bX(\nu),\tau{-}\nu) f(\bX(\nu)) ]d\nu.
  \label{eq:difference_h}
\end{align}


We use a second order Taylor expansion to express $h_\tau(\bphi(\bx',\tau-\nu)) = (h_\tau \circ \bphi)(\bx', \tau - \nu)$ around $\bX(\nu)$. The constant and linear term of the expansion are $h_\tau(\bphi(\bX(\nu),\tau-\nu))$ and \\ $D_x (h_\tau \circ \bphi)(\bX(\nu),\tau-\nu)\Delta X(\nu)$ respectively. By realizing that the sum $\sum_{\bx'\in \calX} K_{\bX(\nu),\bx'}D_x (h_\tau \circ \bphi)(\bX(\nu),\tau-\nu) \Delta X(\nu)$ is equal to $D_x(h_\tau \circ \bphi)(\bX(\nu),\tau-\nu)f(\bX(\nu))$, it follows that equation \eqref{eq:difference_h} is equal to the remainder of the Taylor expansion
\begin{align*}
  \int_0^\tau \E[ \sum_{i,j \in \calI}\sum_{\bx' \in \calX} K_{\bX(\nu),\bx'} \Delta X_i(\nu)\Delta X_j (\nu) \int_0^1 (1-\omega)\frac{\partial^2}{\partial x_i \partial x_j}(h_\tau \circ \bphi)(\bX(\nu) + \omega\Delta \bX(\nu), \tau - \nu) d\omega]d\nu.
\end{align*}

Taking the supremum over all possible values of $\bX(\tau)$ as well as $\bx$ and using the definition of $\bQ$, the above term is bounded by
\begin{align}
\label{eq:bound_refinement_taylor}
\frac{1}{2}\int_0^\tau \sup_{\by,\bz \in \conv{\calX}} \sum_{i,j\in \calI}\abs{Q_{i,j}(\by)}\abs{\frac{\partial^2}{\partial x_i \partial x_j} (h_\tau \circ \bphi)(\bz,\tau-\nu)} d\nu.
\end{align}

The rest of the proof is then essentially a careful analysis of the above sum. For that, we use again Lemma~\ref{lemma:properties_phi} but also need bounds on up to the second derivative of $\bQ$ (This is needed because in the above expression the function $h_\tau$ is defined as a function of $\bQ$).  The latter makes the rest of the proof long and technical but the main ideas are essentially similar to the ones used in Lemma~\ref{lemma:properties_phi}. 

\medskip

The second derivative of $h\circ \bphi$ satisfies (for $i,j\in\calI$):
\begin{align}
& \frac{\partial^2}{\partial x_i \partial x_j} (h_\tau \circ \bphi)(\bz,\tau{-}\nu) = 
\frac{\partial}{\partial x_j} \bigl( \sum_{u \in \calI} \frac{\partial h_\tau}{\partial x_u}(\bphi (\bz, \tau{-}\nu)) \frac{\partial \phi_u}{\partial x_i}(\bz,\tau{-}\nu) \bigr) \nonumber\\
& = \sum_{u \in \calI} \frac{\partial h_\tau}{\partial x_u}(\bphi (\bz, \tau{-}\nu))\frac{\partial^2 \phi_u}{\partial x_i \partial x_j}(\bz,\tau{-}\nu) + \sum_{u, r \in \calI} \frac{\partial^2 h_\tau}{\partial x_u \partial x_r}(\bphi (\bz, \tau{-}\nu)) \frac{\partial \phi_u}{\partial x_i}(\bz,\tau{-}\nu) \frac{\partial \phi_r}{\partial x_j}(\bz,\tau{-}\nu) \label{eq:partial_deriv_g_ks}.
\end{align}
To bound the above term, we need to study $\frac{\partial h_\tau}{\partial x_u}(\bz) = \frac{\partial g_{(k,s)}}{\partial x_u}(\bz,\tau)$ and $\frac{\partial^2 h_\tau}{\partial x_u \partial x_r}(\bz) = \frac{\partial^2 g_{(k,s)}}{\partial x_u \partial x_r}(\bz,\tau)$. Applying the chain rule to the definition of $h_\tau$ shows that the first partial derivative of $h_\tau$ is
\begin{align*}
  \frac{\partial h_\tau}{\partial x_u}(\bz) & = \frac{\partial }{\partial x_u}\left( \sum_{q, l} Q_{m,l}(\bz)\frac{\partial \phi_{(k,s)}}{\partial x_m \partial x_l}(\bz, t{-}\tau) \right)\\
  & = \sum_{q,l \in \calI} \frac{\partial Q_{q,l}}{\partial x_u}(\bz) \frac{\partial^2 \phi_{(k,s)}}{\partial x_q \partial x_l}(\bz,t-\nu) + \sum_{q,l \in \calI} Q_{q,l}(\bz) \frac{\partial^3 \phi_{(k,s)}}{\partial x_q \partial x_l \partial x_u}(\bz,t-\nu).
\end{align*}

Similarly, the second partial derivative is
\begin{align*}
\frac{\partial^2 h_\tau}{\partial x_r \partial x_u}(\bz) &= \frac{\partial }{\partial x_r} \left( \sum_{q,l \in \calI} \frac{\partial Q_{q,l}}{\partial x_u}(\bz) \frac{\partial^2 \phi_{(k,s)}}{\partial x_q \partial x_l}(\bz,t-\nu) + \sum_{q,l \in \calI} Q_{q,l}(\bz) \frac{\partial^3 \phi_{(k,s)}}{\partial x_q \partial x_l \partial x_u}(\bz,t-\nu) \right) \\
& = \sum_{q,l\in \calI} \frac{\partial^2 Q_{q,l}}{\partial x_u \partial x_r}(\bz) \dd{\phi_{(k,s)}}{q}{l}(\bz, t-\nu) + \frac{\partial Q_{q,l}}{\partial x_u}(\bz) \ddd{\phi_{(k,s)}}{q}{l}{r}(\bz, t - \nu) \\
& \qquad + \frac{\partial Q_{q,l}}{\partial x_r}(\bz) \ddd{\phi_{(k,s)}}{q}{l}{u}(\bz, t-\nu) + Q_{q,l}(\bz) \dddd{\phi_{(k,s)}}{q}{l}{u}{r}(\bz, t-\nu).
\end{align*}

\color{myorange}
What remains is to bound the sums appearing in the above derivatives. We use the notations $c^k_{k_1}$ and $c_{k_1,k_2,k_3}$, as in the proof of Lemma~\ref{lemma:taylor_remainder}. From the representation of $\bQ$ given in Appendix~\ref{apx:rmf_def}, it can be seen that $Q_{(k_1,s_1),(k_2,s_2)}(z) = O(c^{k_1}_{k_2})$, $\frac{\partial Q_{(k_1,s_1),(k_2,s_2)}}{\partial x_{(k_3,s_3)}} = O(c_{k_1,k_2,k_3})$ and

\begin{align*}
\dd{Q_{(k_1,s_1),(k_2,s_2)}}{(k_3,s_3)}{(k_4,s_4)} & = \begin{cases}
O(1/n) & \text{ if } (k_1, k_2) = (k_3,k_4) \text{ or } (k_4,k_3), \\
O(1/n^2) & \text{ if } k_1 = k_3, k_4 \text{ or } k_2 = k_3,  k_4, \\
O(1/n^3) & \text{ otherwise. }
\end{cases}
\end{align*}

Lemma \ref{lemma:properties_phi} gives bounds for the partial derivatives of $\bphi$. This enables us to develop an upper bound for $|\frac{\partial h_\tau}{\partial x_u}(z)|$,

\begin{align*}
&\abs{\frac{\partial h_\tau}{\partial x_{(k',s')}}(z)} \leq \sum_{(k_1,s_1),(k_2,s_2) \in \calI} \abs{\frac{\partial Q_{(k_1,s_1),(k_2,s_2)}}{\partial x_{(k',s')}}(z)} \abs{\frac{\partial^2 \phi_{(k,s)}}{\partial x_{(k_1,s_1)} \partial x_{(k_2,s_2)}}(z,t-\nu)} \\
& \phantom{\frac{\partial h_\tau}{\partial x_{(k',s')}}(z) \leq} \quad + \sum_{(k_1,s_1),(k_2,s_2) \in \calI} \abs{Q_{(k_1,s_1),(k_2,s_2)}(z)} \abs{\frac{\partial^3 \phi_{(k,s)}}{\partial x_{(k_1,s_1)} \partial x_{(k_2,s_2)} \partial x_{(k',s')}}(z,t-\nu)} \\
& = \sum_{(k_1,s_1),(k_2,s_2) \in \calI} O\bigl(c_{k_1,k_2,k'}\bigr)O\bigl(\frac{1}{n}(c^k_{k_1}+c^k_{k_2})\bigr) + \sum_{(k_1,s_1),(k_2,s_2) \in \calI} O\bigl( c^{k_1}_{k_2} \bigr)O\bigl( \frac{1}{n^2} (c^{k}_{k_1} + c^{k}_{k_2} + c^{k}_{k_3})\bigr) 
= O(\frac1n c^{k}_{k'}).
\end{align*}
With the above observations we bound the first sum of \eqref{eq:partial_deriv_g_ks} by
\begin{align*}
\sum_{u \in \calI} \abs{\frac{\partial h_\tau}{\partial x_u}(\bphi (\bz, \tau-\nu))} \abs{\frac{\partial^2 \phi_u}{\partial x_i \partial x_j}(\bz,\tau-\nu)}
& =\sum_{(k',s') \in \calI} \abs{\frac{\partial g_{(k,s)}}{\partial x_{(k',s')}}(\bphi (\bz, \tau-\nu), \tau)} \abs{\frac{\partial^2 \phi_{(k',s')}}{\partial x_{(k_1,s_1)} \partial x_{(k_2,s_2)}}(\bz,\tau-\nu)} \\
& = \sum_{(k',s')\in\calI} O(\frac1n c^k_{k'})O(\frac1n (c^{k'}_{k_1} + c^{k'}_{k_2})) = O(\frac{1}{n^2}(c^k_{k_1} + c^k_{k_2})).
\end{align*}
For the second partial derivatives of $h_\tau(\bz) = g_{(k,s)}(\bz,t-\tau)$, we note that all sums which appear in the explicit form of the partial derivative are bounded by $O(\frac{1}{n^2}(c^k_{k'} + c^k_{\hat{k}}))$. Using the bounds for $\bQ$ and $\bphi$ and their respective partial derivatives we see that
\begin{align*}
&\sum_{(k_1,s_1),(k_2,s_2)}  \abs{Q_{(k_1,s_1),(k_2,s_2)(\bz)}} \abs{\dddd{\phi_{(k,s)}}{(k_1,s_1)}{(k_2,s_2)}{(k',s')}{(\hat{k},\hat{s})}(\bz, t-\nu)} \\
&= \sum_{(k_1,s_1),(k_2,s_2)} O(c^{k_1}_{k_2})O(\frac{1}{n^3}( c^k_{k_1} + c^k_{k_2} + c^k_{k'} + c^k_{\hat{k}})) = O(\frac{1}{n^2}(c^k_{k'} + c^k_{\hat{k}})),
\end{align*}
and that
\begin{align*}
\sum_{(k_1,s_1),(k_2,s_2)}  \abs{\frac{\partial Q_{(k_1,s_1),(k_2,s_2)}}{\partial x_{(\hat{k},\hat{s})}}(\bz)} \abs{\ddd{\phi_{(k,s)}}{(k_1,s_1)}{(k_2,s_2)}{(k',s')}(\bz, t-\nu)}\\
= \sum_{(k_1,s_1),(k_2,s_2)} O(c_{k_1,k_2,\hat{k}})O(\frac{1}{n^2}( c^k_{k_1} + c^k_{k_2} + c^k_{k'}))
&= O(\frac{1}{n^2}(c^k_{k'} + c^k_{\hat{k}})),
\end{align*}
and

\begin{align*}
  \sum_{(k_1,s_1),(k_2,s_2)}  \abs{\frac{\partial^2 Q_{(k_1,s_1),(k_2,s_2)}}{\partial x_{(\hat{k},\hat{s})} \partial x_{(k',s')}}(z)} \abs{\dd{\phi_{(k,s)}}{(k_1,s_1)}{(k_2,s_2)}(\bz, t-\nu)}
  = O(\frac{1}{n^2}(c^k_{k'} + c^k_{\hat{k}})).
\end{align*}

The last bound follows with careful case-by-case analysis for the second derivative of $\frac{\partial^2 Q_{(k_1,s_1),(k_2,s_2)}}{\partial x_{(\hat{k},\hat{s})}\partial x_{(k',s')}}$. As a direct consequence $\abs{\frac{\partial^2 h_\tau}{\partial x_{(k',s')} \partial x_{(\hat{k},\hat{s})}}(\bphi (\bz, \tau-\nu))} = O(\frac{1}{n^2}(c^k_{k'} + c^k_{\hat{k}}))$. This enables us to establish a bound for the second sum of \eqref{eq:partial_deriv_g_ks},
\begin{align*}
  & \sum_{u, r \in \calI} \abs{\frac{\partial^2 h_\tau}{\partial x_u \partial x_r}(\bphi (\bz, \tau-\nu))} \abs{\frac{\partial \phi_u}{\partial x_i}(\bz,\tau-\nu)} \abs{\frac{\partial \phi_r}{\partial x_j}(\bz,\tau-\nu)} \\
  & = \sum_{(k',s'), (\hat{k},\hat{s}) \in \calI} \abs{\frac{\partial^2  g_{(k,s)}}{\partial x_{(k',s')} \partial x_{(\hat{k},\hat{s})}}(\bphi (\bz, \tau-\nu),\tau)} \abs{\frac{\partial \phi_{(k',s')}}{\partial x_{(k_1,s_1)}}(\bz,\tau-\nu)} \abs{\frac{\partial \phi_{(\hat{k},\hat{s})}}{\partial x_{(k_2,s_2)}}(\bz,\tau-\nu)} \\
  & = \sum_{(k',s'), (\hat{k},\hat{s}) \in \calI} O(\frac{1}{n^2}(c^k_{k'} + c^k_{\hat{k}}))O(c^{k'}_{k_1})O(c^{\hat{k}}_{k_2}) = O(\frac{1}{n^{2}}(c^k_{k_1} + c^k_{k_2})).
\end{align*}
\color{black}

For the last part of the proof, we use the obtained results to bound $\sum_{i,j\in \calI}\abs{Q_{i,j}(\by)}\abs{\frac{\partial^2}{\partial x_i \partial x_j} (h_\tau \circ \bphi)(\bz,\tau-\nu)}$. 
By equation \eqref{eq:partial_deriv_g_ks} we see that \eqref{eq:bound_refinement_taylor} is equal to 
\begin{align*}
& \frac{1}{2}\int_0^\tau \sup_{\by,\bz \in \conv{\calX}} \sum_{i,j\in \calI}\abs{Q_{i,j}(\by)} \left| \sum_{u \in \calI} \frac{\partial h_\tau}{\partial x_u}(\bphi (\bz, \tau-\nu))\frac{\partial^2 \phi_u}{\partial x_i \partial x_j}(\bz,\tau-\nu) \right. \\
&  \qquad \left. + \sum_{u, r \in \calI} \frac{\partial^2 h_\tau}{\partial x_u \partial x_r}(\bphi (\bz, \tau-\nu)) \frac{\partial \phi_u}{\partial x_i}(\bz,\tau-\nu) \frac{\partial \phi_r}{\partial x_j}(\bz,\tau-\nu) \right| d\nu.
\end{align*}

Indeed, we bound the supremum by the two following terms
\begin{align}
&  \sum_{(k_1,s_1),(k_2,s_2)\in \calI} \abs{Q_{(k_1,s_1),(k_2,s_2)}(\by)}\sum_{(k',s') \in \calI} \abs{\frac{\partial g_{(k,s)}}{\partial x_{(k',s')}}(\bphi (\bz, \tau-\nu), \tau)} \abs{\frac{\partial^2 \phi_{(k',s')}}{\partial x_{(k_1,s_1)} \partial x_{(k_2,s_2)}}(\bz,\tau-\nu)} \nonumber\\
& \phantom{\sum_{(k_1,s_1),(k_2,s_2)}} = \sum_{(k_1,s_1),(k_2,s_2)} O(c^{k_1}_{k_2})O(\frac{1}{n^2}(c^k_{k_1} + c^k_{k_2})) = O(1/n^2) \label{eq:refinement_remainder_bound:1}
\end{align}
and
\begin{align}
& \sum_{(k_1,s_1),(k_2,s_2)\in\calI} \abs{Q_{(k_1,s_1),(k_2,s_2)}(\by)}\sum_{(k',s'), (\hat{k},\hat{s}) \in \calI} \abs{\frac{\partial^2  g_{(k,s)}}{\partial x_{(k',s')}\partial x_{(\hat{k},\hat{s})}}(\bphi (\bz, \tau-\nu),\tau)} \nonumber \\
& \phantom{\sum_{(k_1,s_1),(k_2,s_2)} \abs{Q_{(k_1,s_1),(k_2,s_2)}(z)}\sum_{(k',s'), (\hat{k},\hat{s}) \in \calI} \ldots}  \times \abs{\frac{\partial \phi_{(k',s')}}{\partial x_{(k_1,s_1)}}(\bz,\tau-\nu)} \abs{\frac{\partial \phi_{(\hat{k},\hat{s})}}{\partial x_{(k_2,s_2)}}(\bz,\tau-\nu)} \nonumber \\
&\phantom{\sum_{(k_1,s_1),(k_2,s_2)}} = \sum_{(k_1,s_1),(k_2,s_2)} O(c^{k_1}_{k_2})O(\frac{1}{n^2}(c^k_{k_1} + c^k_{k_2})) = O(1/n^2). \label{eq:refinement_remainder_bound:2}
\end{align}

The bounds \eqref{eq:refinement_remainder_bound:1} and \eqref{eq:refinement_remainder_bound:2} show that \eqref{eq:bound_refinement_taylor} is of order $O(1/n^2)$, where the hidden constant depends on $\tau, \bar{r}, \abs{\calS}$, from which the claim of the Lemma follows. \end{proof}

\end{document}